\documentclass[11pt,a4paper,preview]{article}
\usepackage[utf8]{inputenc}
\usepackage{amsmath}
\usepackage{amsfonts}
\usepackage{amssymb}
\usepackage{amsthm}

\usepackage{graphicx, caption, subcaption}

\usepackage{xcolor}
\usepackage[left=2.5cm,right=2.5cm,top=2cm,bottom=2cm]{geometry}
\usepackage{hyperref}
\usepackage{apacite}
\usepackage{bbm}
\usepackage{booktabs}
\usepackage{color}
\usepackage{multicol}
\usepackage[justification=centering]{caption}
\usepackage{subcaption}
\usepackage{multirow}
\usepackage{stackrel}
\usepackage{hhline}
\usepackage{makecell}
\usepackage{nicefrac}
\usepackage[hang]{footmisc}
\usepackage{authblk}
\usepackage{todonotes}
\usepackage{csquotes}
\usepackage{chngcntr} 
\usepackage[margin=true,inline=false,index=true]{fixme} 
\usepackage{endnotes}
\usepackage{comment}

\let\footnote=\endnote
\fxsetup{theme=color}
\definecolor{fxtarget}{rgb}{0.8000,0.0000,0.0000}
\allowdisplaybreaks 

\counterwithin{figure}{section}
\counterwithin{table}{section}

\numberwithin{equation}{section}
\theoremstyle{definition}
\newtheorem{assumption}{Assumption}[section]

\newtheorem{lemma}[assumption]{Lemma}

\newtheorem{corollary}[assumption]{Corollary}
\newtheorem{proposition}[assumption]{Proposition}

\newcommand{\E}[1]{\mathbb{E}\left[ #1\right]}

\newcommand{\tin}{t \in [0, T]}
\newcommand{\Q}{\mathbb{Q}}

\newcommand{\Ft}{\mathcal{F}_t}

\DeclareMathOperator*{\argmax}{argmax}

\newcommand{\rBrackets}[1]{\left( #1 \right)}

\title{Risk sharing in equity-linked insurance products: Stackelberg equilibrium between an insurer and a reinsurer}

\author[a,b]{Yevhen Havrylenko}
\author[c]{Maria Hinken}
\author[a]{Rudi Zagst}
\affil[a]{Technical University of Munich, TUM School of Computation, Information and Technology, Department of Mathematics, Chair of Mathematical Finance, Parkring 11, 87548, Garching, Germany}
\affil[b]{University of Copenhagen, Denmark, Department of Mathematical Sciences, Section for Insurance and Economics, Universitetsparken 5, 2100, Copenahgen, Denmark}
\affil[c]{Ulm University, Germany; Institute of Insurance Science}
\date{}                     
\begin{document}
\maketitle

\begin{abstract}
We study the optimal investment-reinsurance problem in the context of equity-linked insurance products. Such products often have a capital guarantee, which can motivate insurers to purchase reinsurance. Since a reinsurance contract implies an interaction between the insurer and the reinsurer, we model the optimization problem as a Stackelberg game. The reinsurer is the leader in the game and maximizes its expected utility by selecting its optimal investment strategy and a safety loading in the reinsurance contract it offers to the insurer. The reinsurer can assess how the insurer will rationally react on each action of the reinsurer. The insurance company is the follower and maximizes its expected utility by choosing its investment strategy and the amount of reinsurance the company purchases at the price offered by the reinsurer. In this game, we derive the Stackelberg equilibrium for general utility functions. For power utility functions, we calculate the equilibrium explicitly and find that the reinsurer selects the largest reinsurance premium such that the insurer may still buy the maximal amount of reinsurance. Since in the equilibrium the insurer is indifferent in the amount of reinsurance, in practice, the reinsurer should consider charging a smaller reinsurance premium than the equilibrium one. Therefore, we propose several criteria for choosing such a discount rate and investigate its wealth-equivalent impact on the expected utility of each party.
\end{abstract}




\section{Introduction}

\hspace{0.7cm}\textbf{Motivation.} This paper is driven by a practical problem that appears in life insurance, namely the optimal choice of investment strategies for equity-linked insurance products and the corresponding reinsurance strategies. Reinsurance in such products can help insurance companies to meet their obligations towards clients, i.e., pay back a guaranteed capital amount to the policyholder (buyer of an equity-linked insurance product) at the end of the investment horizon. For example, in the previous decade, the insurance company ERGO started offering a pension product \enquote{ERGO Rente Garantie} with a capital guarantee of 80\% or 100\% of clients contributions.  According to the product description\footnote{See \url{https://www.yumpu.com/de/document/view/22247401/expertenwissen-zur-ergo-rente-garantie}}, clients' contributions are invested in the money market and a target volatility fund (TVF) and are reinsured by Munich Re\footnote{See \url{https://www.focus.de/finanzen/steuern/ergo-ergo-rente-garantie_id_3550999.html}}. In that product, the investment strategy of the insurer is fully transparent to the reinsurer, as the TVF is constructed in full agreement with a reinsurer, i.e., the insurer does not have much investment freedom. Note that there are other ways of ensuring that capital guarantees are reached, e.g., a constant-proportion portfolio insurance (CPPI) strategy \textcolor{black}{or a put-option replicating strategy of the investment portfolio}. However, those ways are highly challenged in practice. A CPPI strategy does not leave much freedom for insurers to generate decent surplus for their policyholders and shareholders when interest rates are low. A put-option replicating strategy of the insurer's individual portfolio might be too expensive due to transaction costs. Furthermore, insurance companies, especially smaller ones, may not have a trading department necessary for dynamic hedging. If, on the other hand, an insurance company buys reinsurance from a reinsurance company, it can lower its transaction costs, does not need a trading department for hedging its position and can get rid of at least some of the risk in its books.

Next we describe in more detail the practical problem that motivates our research. An insurance company (insurer for short) sells an equity-linked product with a capital guarantee to a customer. The insurer is willing to run its own individual investment strategy in some sub-universe of a global liquid market. It does not want to disclose its investment strategy to a third party like other insurers or reinsurers. Focusing on its own investment strategy, the insurer can buy (partial) reinsurance on a simple investment strategy in the global liquid market or even another liquid sub-universe of the global market. To get a reasonable reinsurance contract, the investment strategy to be reinsured should be transparent, easy to understand as well as implement and sufficiently correlated with the individual investment strategy of the insurer. A typical candidate is a constant-mix strategy. The reinsurance contract (put option) usually has a long maturity, is not continuously traded in the market and, thus, has to be bought over-the-counter, i.e., from a reinsurance company. Due to a long maturity of the put option modelling the reinsurance contract and the impossibility to trade it on exchanges, the reinsurance company can charge an additional price margin above the expected discounted loss of the contract.

\textcolor{black}{Note that in practice it may be beneficial for the insurer to disclose its investment strategy to the reinsurer. In that case, the parties can agree on a reinsurable portfolio that has a higher correlation with the insurer's portfolio and/or has lower costs for hedging. As a result, the reinsurer may offer a lower price for reinsurance in practice. However, even if the insurance company discloses fully its investment strategy, the company  may want to deviate from the agreed reinsurance strategy in the future, e.g., once a financial-market shock happens and the company needs to adjust its strategy. Additionally, the individual investment strategy might not be standardized or explainable by simple rules and might also have discretionary components, so that the reinsurer may not be willing to reinsure it.}


As a reinsurance contract is an agreement between a primary insurance company and a reinsurance company,  it implies interaction/negotiation between involved parties. In the simplest case, the insurance company chooses the amount of reinsurance, which we call the reinsurance strategy of the insurance company. The reinsurance company selects the price of the reinsurance. If the reinsurer takes into account only its own preferences when creating a reinsurance contract, the reinsurance premium might be too high for the insurer to accept the contract and, therefore, the insurer may not be willing to buy as much reinsurance as the reinsurer expects. On the other hand, if the reinsurance premium is too low, the insurer chooses to buy reinsurance for more claims, but the net profit of this deal may be suboptimal for the reinsurer due to the low price. Hence, the outcome for the reinsurer would be worse than expected. The described interaction between the parties motivates the usage of game-theoretical concepts for modelling.


Since there are significantly more insurance companies than reinsurance companies worldwide (e.g., see \citeA{albrecher2017reinsurance}), the reinsurance company has a rather dominant position in the  negotiation process when setting up a reinsurance contract (e.g., see \citeA{Chen2018}, \citeA{bai2019hybrid}). In addition, the reinsurance company is often larger than the primary insurance company and acts internationally, whereas the insurer often acts on a national level (e.g., see \citeA{albrecher2017reinsurance}). Therefore, the reinsurer is likely to have more investment opportunities and has the ability to assess how the insurer reacts to various specifications of a reinsurance contract, in particular, its price. Due to this difference in negotiating power and information asymmetry, the reinsurance contract can be considered as a hierarchical game between both parties.

The above-mentioned situation suits well to the concept of Stackelberg games. These games  have a hierarchical structure where the leader  "dominates" the follower, i.e., the leader moves first and selects his/her strategy knowing the future optimal response of the follower, whereas the follower moves afterwards and chooses its strategy depending on the choice of the leader.

According to the above, we study a Stackelberg game between an insurer and a reinsurer in the context of an equity-linked insurance product with a capital guarantee.
\vskip 1\baselineskip
\textbf{Modeling aspects and research questions.} In our model, we consider one reinsurer, one insurer and one representative client. The representative client pays an initial contribution to the insurance company that invests the money -- on behalf of the client --  in financial markets.  The insurer can buy reinsurance from the reinsurance company only at the beginning of the investment period\footnote{In reality, the reinsurance contract can be adjusted at regular intervals, e.g., annually. To solve the corresponding problem with several discrete adjustments during the investment period, our model can be applied sequentially.}. In the reinsurance contract, the insurer chooses the amount of risk transferred to the reinsurer, which we call the reinsurance strategy, and the reinsurer selects the reinsurance price. We model reinsurance as a put option, which is written on the potential losses related to some benchmark investment portfolio. This benchmark portfolio is highly correlated with but not necessarily equal to the individual investment portfolio of the insurer. This models the situations when the reinsurance company is not willing to reinsure the insurer's individual investment strategy, e.g., due to its high riskiness, nontransparency, because the insurer does not want to disclose its strategy, or because the reinsurer only wants to sell reinsurance based on a standard market index.
The aim of each party is to maximize its expected utility of the total terminal wealth, where the total terminal wealth consists of the terminal portfolio value including the reinsurance payoff. However, the reinsurer (leader) moves first and chooses its strategy first knowing the optimal response of the insurer (follower). Afterwards, the insurer (follower) selects its strategy knowing the optimal strategy of the reinsurer (leader). The solution to a Stackelberg game is called the Stackelberg equilibrium.

Using the described setting, we answer the following research questions:
\begin{enumerate}
    \item How can we analytically derive the Stackelberg equilibrium?
    \item What happens with the parties' expected utilities in case one of the parties  deviates from the Stackelberg equilibrium, e.g., the insurer does not buy reinsurance or the reinsurer sets a reinsurance contract price lower than the equilibrium one?
    \item What is the impact of model parameters on the Stackelberg equilibrium?
\end{enumerate}

\textbf{Contribution.} We contribute to the existing literature as described next.$\;$\\
First, we formulate and analyze a novel Stackelberg game between a reinsurer and an insurer, which is more realistic than Stackelberg games previously studied in the literature. In contrast to them, in our model the reinsurance is written on potential losses within an insurance product rather than the whole surplus process of an insurer, and is purchased at the beginning of the investment period, not dynamically traded at future points. Additionally, we allow that the investment portfolio in the insurance product and the underlying portfolio in the reinsurance contract are not identical. Second, we draw insightful economic conclusions from the equilibrium of the considered game. In particular, we find that in the Stackelberg equilibrium the reinsurer selects the largest reinsurance premium such that the insurer may buy the maximal amount of reinsurance.
The insurer's expected utility becomes constant with respect to the amount of purchased reinsurance, which is why, in practice, the reinsurer should charge a lower (discounted) safety loading of the reinsurance premium than the equilibrium one in order to secure a deal with the maximal amount of reinsurance. In this regard, we show that the reinsurer's expected utility in case of a discounted safety loading is higher than the expected utility corresponding to the case when no reinsurance is sold, which can happen for the equilibrium safety loading without a discount.
Thus, we provide three different ways how the reinsurer can determine a reasonable discount and show that for the discounted safety loading the optimally acting insurer may significantly reduce its product costs without decreasing its expected utility.
For example, in the base parametrization of our model, if the reinsurer charges only $86.73\%$ of its equilibrium safety loading, then the insurer needs $16$ bp\footnote{$1$ bp = 1 basis point = $0.0001$} less initial capital if it switches from an \enquote{old} product with optimal investment strategy without reinsurance to a \enquote{new} product with reinsurance. The cost benefits can be substantially higher depending on the investment strategy in the former insurance product. 

\textit{Portfolio optimization literature.} Our contribution to this stream of literature lies in solving a novel portfolio optimization problem with non-standard features by combining existing advanced techniques. To the best of our knowledge, this problem has not been considered earlier. First, we solve a portfolio optimization problem with a fixed position in a put option using the idea of linking option and stock market in \citeA{korn1999optimal}\footnote{We use the idea from Remark 5.2 of \citeA{korn1999optimal} and prove a similar statement as Theorem 5.1 in \citeA{korn1999optimal}}. Second, we solve a portfolio optimization problem with an investment strategy constraint and a fixed-term investment in a put option by combining the concept of auxiliary markets from \citeA{cvitanic1992convex} and the generalized martingale approach from \citeA{desmettre2016optimal}.

\vskip 1\baselineskip
\textbf{Literature overview.} {\color{black} Next we provide an overview of the relevant literature and organize it in three groups. The first stream of literature is related to risk-sharing between insurance and reinsurance companies in the context of pension and equity-linked products. The second group of sources is devoted to Stackelberg games between an insurer and a reinsurer. The third stream of literature is related to advanced portfolio optimization techniques, which we combine to solve the novel Stackelberg game considered in this chapter. }

\textit{Insurance focus.} The first paper that considered the usage of reinsurance in portfolio selection in a life insurance context is \citeA{muller1985investment}. The author considers a static investment-reinsurance problem for a pension fund, where preferences are described by an exponential utility function and there are no constraints on the terminal wealth. However, the interaction between the reinsurer and the pension fund in terms of contract parameters is not taken into account.

\citeA{Escobar} revisit the usage of reinsurance for life-insurance products and solve in a Black-Scholes market a dynamic investment-reinsurance problem for an insurer with no-short-selling and Value-at-Risk constraints. The researchers assume that reinsurance is dynamically traded and fairly priced\footnote{The price of reinsurance is equal to the fair price of the corresponding put option, i.e., the reinsurer does not charge anything extra, e.g., a safety loading}. In our paper, we consider a similar but more realistic setting, namely the reinsurance is traded only at the initial time and its price can be higher than the fair price of the corresponding put option. Moreover, we model the interaction between the parties as a Stackelberg game. To make it analytically tractable, we drop the Value-at-Risk constraint, which was present in \citeA{Escobar} though.

{In the general scope of risk-management techniques for equity-linked insurance products and variable annuities, two types of hedging strategies are distinguished: static and dynamic. Our proposed approach based on a reinsurance contract on a suitable benchmark portfolio is an example of a static hedging strategy to protect the insurer's investment portfolio against the market risk. Another common approach to control the riskiness of a portfolio is a terminal-wealth constraint based on risk measures like the Value-at-Risk or the conditional tail expectation. However, for the reasons mentioned in the previous paragraph, the model in this paper does not incorporate such constraints. For a survey of risk-management techniques for equity-linked insurance, see \citeA{feng2022variable} and references therein.}

\textit{Stackelberg games in insurance.} There are different ways to model the interaction between a reinsurer and an insurer in a reinsurance contract. For example, \citeA{li2016optimal} consider a general insurance group which consists of an insurance and a reinsurance company. The authors solve a reinsurance-investment problem for the general insurance group which incorporates the decision of the insurance company and the reinsurance company. Another way to model a reinsurance contract is via the principal-agent problem, e.g., \citeA{gu2020optimal} solve a reinsurance-investment problem where the principal assumes a worst-case scenario. To the best of our knowledge, the first paper solving a reinsurance problem with investment returns in the context of a Stackelberg game is \citeA{Chen2018}. The researchers consider a reinsurer who offers reinsurance on the whole claim process of an insurer. The authors assume that the reinsurance contract can be adjusted dynamically by the reinsurer and the insurer. Further papers on Stackelberg games in the context of insurance-reinsurance interaction are \citeA{Chen2019}, \citeA{bai2019hybrid}, \citeA{chen2020continuous}, \citeA{yuan2021robust}, \citeA{yang2021robust}, and \citeA{bai2021stochastic}. In all of them, the reinsurance is offered on the whole portfolio of aggregated insurance obligations and the reinsurance can be dynamically adjusted over the investment horizon. \textcolor{black}{\citeA{Chen2018}, \citeA{Chen2019} and \citeA{chen2020continuous} assume that the reinsurer and the insurer invest only their surplus process risk-free, whereas in \citeA{bai2019hybrid} the parties invest in a more complex financial market that includes one risky asset.} \citeA{gavagan2021optimal} study the Stackelberg game between a reinsurer and an insurer with model uncertainty, assuming that the reinsurance is not dynamically adjusted. However, the authors do not consider investment strategies in the game. Another application of Stackelberg games in insurance, for example, is the model of \citeA{asmussen2019stackelberg}. The authors used the framework of a Stackelberg game to model the competition of two insurance companies for their clients. {\color{black} Last but not least, we would like to mention \citeA{Boonen2022}, where the authors derive the optimal static reinsurance premium and indemnity function in a Stackelberg game between an insurer and a reinsurer whose preferences are described by monotone risk measures. Using a model without investment opportunities, the researchers show that the Stackelberg equilibrium makes the insurer indifferent with the status quo of not buying reinsurance.} We are not aware of any publications focusing on an optimal investment-reinsurance problem in the context of a Stackelberg game where reinsurance is static and offered within an equity-linked insurance product. Our research fills this gap.

\color{black}
\textit{Portfolio optimization focus.} The seminal papers \citeA{merton1969lifetime} and \citeA{merton1975optimum} introduced and solved the classic continuous-time portfolio optimization problem. There are no portfolio constraints, no fixed-term investments and no investments in options. Since then, various extensions have been considered, which were tackled either by the Hamilton-Jacobi-Bellman (HJB) approach or the Martingale approach.

In our paper, the Stackelberg game consists of two utility maximization subproblems. Each of them has features that do not allow direct usage of standard solution methods.\\

When solving the insurer's optimization problem, we are dealing with two peculiarities:

$\quad$ 1. A fixed long position in the put option, which is also a control variable.

$\quad$ 2. The put option is not spanned by the assets that the insurer trades.\\

To solve such an optimization problem, we combine the ideas of \citeA{desmettre2016optimal} and \citeA{cvitanic1992convex}. \citeA{cvitanic1992convex} solve a constrained optimization problem, where the investment strategy constraint is given by a convex set. The authors construct a family of unconstrained optimization problems in different auxiliary markets and find the unconstrained optimization problem that models the required portfolio constraint in the original market. In our case, an auxiliary market is the global market with changed drift coefficients, where the reinsurance contract is spanned. This way we make the fixed investment spanned and manage to derive closed-form solutions via \citeA{desmettre2016optimal}. In \citeA{desmettre2016optimal}, the authors develop a generalized martingale approach to solve an optimization problem where the investor can additionally invest in a fixed-term security at the beginning of the investment period. The first step of the approach is to find the optimal investment in the risky asset for a fixed position in the fixed-term security, which requires the inversion of conditional random utility functions.\footnote{ If the fixed-term investment is not spanned, the inversion of conditional random utility functions rarely leads to closed-form solutions. However, we circumvent this problem by treating the insurer's problem as a portfolio optimization problem with an investment strategy constraint and transforming it to the unconstrained one in an auxiliary market, where the put option is spanned.} The second step is to find the optimal investment in the fixed-term security given the investment in the risky assets by maximizing the corresponding value function with respect to a fixed-term position.\\

When solving the optimization problem of the reinsurer, we face the following two challenges:

$\quad$ 1. A fixed short position in the put option, which is predetermined by the insurer.

$\quad$ 2. A safety loading as an additional control variable of the reinsurer.\\

In contrast to the insurer's case, the reinsurer as a larger player trades in the global market and thus can hedge its fixed position in the reinsurance contract. For this reason, we use the concept of replicating  strategies from \citeA{korn1999optimal}. Their idea is to link the stock and the option market via replicating strategies of options. Once the link is established, the researchers combine it with the classic martingale approach to solve the optimization problem that contains investment in stocks and/or options. To overcome the  second challenge of the reinsurer's problem, we solve it in two steps, where the second one is the maximization of the reinsurer's value function with respect to the safety loading. The difference to the setting of \citeA{desmettre2016optimal} is that the safety loading is the control variable, not the position in the fixed-term asset

\vskip 1\baselineskip
\textbf{Paper structure.} 
In Section \ref{sec_problem_setting} we describe formally the financial market as well as the Stackelberg game between the reinsurer and the insurer. The optimal solution to the Stackelberg game is derived in Section \ref{sec_solution_to_SG}. It is divided into the subsections devoted to the optimization problem of the insurer and the optimization problem of the reinsurer, as the corresponding solution approaches differ. In Section \ref{sec_example} we consider the special case of a power utility function and a fixed upper bound on the reinsurance strategy of the insurer. In Section \ref{sec_numerical_studies} we conduct numerical studies in the context of the German market. First, we analyze the Stackelberg equilibrium numerically and investigate its sensitivity with respect to the parties' risk aversion coefficients and with respect to the model parameters influencing the fair price of the put option. Second, we study how deviations from the Stackelberg equilibrium influence each party's expected utilities. Section \ref{sec_conclusions} concludes the paper. In Appendix \ref{app:proofs_main} we provide the proofs of the main results from Section \ref{sec_solution_to_SG}. Appendix \ref{app:proofs_example} contains the proofs of results from Section \ref{sec_example}. In supplementary materials, which are available online, we provide auxiliary lemmata, additional plots, and discussions on other utility functions as well as the potential inclusion of actuarial risks in our model.

\section{Problem setting}\label{sec_problem_setting}

\hspace{0.5cm} In our model, the relationship between the reinsurer and the insurer is modeled by a Stackelberg game, where the reinsurer is the leader and the insurer the follower of the game. Before we introduce our model, we give a general definition of a Stackelberg game. For more details on Stackelberg games, see \citeA{osborne1994course}, \citeA{fudenberg1991game}, or \citeA{bressan2011noncooperative}. For applications of Stackelberg games in the context of reinsurance see \citeA{Chen2018}, \citeA{Chen2019}, \citeA{chen2020continuous}, \citeA{asmussen2019stackelberg}, or \citeA{bai2019hybrid}. 
    
    A Stackelberg game is a game with two players, called a leader and a follower, and a hierarchical structure. The players have objective functions $J_i(a_L,a_F):A_L\times A_F\to\mathbb{R}$ for $i\in\{L,F\}$, respectively, where $a_L$ is an action of the leader chosen from the set of admissible actions $A_L$ and $a_F$ an action of the follower chosen from the set of admissible actions $A_F$. The aim of players is to maximize their objective function with respect to their action. The hierarchy in the game should be understood as follows. The leader of a Stackelberg game dominates the follower and, therefore, the leader chooses its action first, knowing the response of the follower, and afterwards the follower acts depending on the action chosen by the leader. 
   
    The Stackelberg game between the leader and the follower is defined as 
    \begin{align*}
        \max_{a_L\in A_L}&J_L(a_L,a_F^\ast)\\
        \text{s.t. }&a_F^\ast\in\underset{a_F\in A_F}{\text{arg max }}J_F(a_L,a_F).
    \end{align*}
    
    The procedure of solving a Stackelberg game is backward induction. This means that the follower first solves its optimization problem for all possible actions of the leader. Second, the leader solves its optimization problem knowing the optimal response actions of the follower.
    
    A solution $(a_L^\ast,a_F^\ast)$ to a Stackelberg game is called a Stackelberg equilibrium if it satisfies two conditions:
    \begin{equation}\label{eq:SE_condition_1}\tag{$SEC_1$}
        a_F^\ast\in\underset{a_F\in A_F}{\text{arg max }}J_F(a_L^\ast,a_F)
    \end{equation}
    \begin{equation}\label{eq:SE_condition_2}\tag{$SEC_2$}
        J_L(a_L,a_F)\leq J_L(a_L^\ast,a_F^\ast)\,\,\forall (a_L,a_F)\in A_L\times A_F \text{ s.t. } a_F\in\underset{\tilde{a}_F\in A_F}{\text{arg max }}J_F(a_L,\tilde{a}_F).
    \end{equation}
    
    Condition \eqref{eq:SE_condition_1} ensures that $a_F^\ast$ is an optimal solution to the optimization problem of the follower. The second condition \eqref{eq:SE_condition_2} states that whenever there are several optimal responses of the follower to an action of the leader, the follower chooses the best response for the leader (cf. \citeA{bressan2011noncooperative}). Hence, this is also called an optimistic bilevel optimization problem (cf. \citeA{zemkoho2016solving}, \citeA{wiesemann2013pessimistic}, \citeA{liu2018pessimistic}, \citeA{mallozzi1995weak}).
    
    Now, we introduce our model, namely specify the financial market, sets of admissible actions of each player and players' objective functions. Let $T>0$ be a finite time horizon and $(W(t))_{\tin}=((W_1(t),W_2(t))^\top)_{\tin}$ a two-dimensional Brownian motion on a filtered probability space $(\Omega,\mathcal{F},(\Ft)_{\tin},\Q)$. \textcolor{black}{We denote by $(x_1,x_2)^\top$ a two-dimensional column vector.} The basic financial market, which we denote by $\mathcal{M}$, consists of one risk-free asset $S_0$ and two risky assets $S_1$, $S_2$. The prices of these assets are given by
    \begin{align*}
        dS_0(t)=&S_0(t)rdt,\;S_0(0)=1,\\
        dS_1(t)=&S_1(t)(\mu_1dt+\sigma_1dW_1(t)),\;S_1(0)=s_1>0,\\
	    dS_2(t)=&S_2(t)(\mu_2dt+\sigma_2(\rho dW_1(t)+\sqrt{1-\rho^2}dW_2(t))),\;S_2(0)=s_2>0,
    \end{align*}
    with constants $r\in \mathbb{R},\;\mu_1>r,\;\mu_2>r,\;\sigma_1>0,\;\sigma_2>0$ and $\rho \in [-1,1]$. \textcolor{black}{Note that $r$ can be negative.} $S_0$ can be interpreted as a bank account, $S_1$ as a risky fund in the individual investment strategy of the insurer. This fund cannot be reinsured. $S_2$ is a risky fund, which is reinsurable. We also use the following notation:
    \begin{align*}
        \begin{matrix}
        \mu:=\begin{pmatrix}
		\mu_1\\
		\mu_2
    	\end{pmatrix} & \sigma:=\begin{pmatrix}
    		\sigma_1 & 0\\
    		\sigma_2\rho & \sigma_2\sqrt{1-\rho^2}
    	\end{pmatrix} & \mathbbm{1}:=(1,1)^\top;\\
        \gamma:=\sigma^{-1}(\mu-r\mathbbm{1}) & 	\widetilde{Z}(t):=e^{-(r+\frac{1}{2}\lVert\gamma\rVert^2)t-\gamma^\top W(t)},
        \end{matrix}&
    \end{align*}
    \textcolor{black}{where $\lVert\cdot\rVert$ denotes the Euclidean norm on $\mathbb{R}^2$.}
    
    We consider one representative client, one insurer and one reinsurer. The insurer sells an equity-linked product to the representative client. Therefore, the representative client pays an initial contribution $v_I>0$ to the insurer and expects to receive a capital guarantee $G_T>0$ at the end of the investment period. The insurer invests the initial contribution in assets $S_0$ and $S_1$. To increase the chances of delivering a capital guarantee $G_T$ to the client, the insurer can buy reinsurance from the reinsurer to cover the downside risk of its portfolio at the beginning of the investment period. At the end of the investment period, the insurer receives the payment of the reinsurance contract. The reinsurance is modeled by a put option with the strike price being equal to the capital guarantee $G_T$ and the underlying asset being a benchmark portfolio.
    
    We assume that the benchmark portfolio is a constant-mix portfolio with respect to $S_0$ and $S_2$. This is motivated by the equity-linked insurance product \enquote{ERGO Rente Garantie,} where the reinsurable fund is a target volatility fund. In the Black-Scholes market, constant-mix strategies and target-volatility strategies are equivalent. An advantage of a constant-mix portfolio is that the reinsurer can assess the potential loss quite well and can calculate the price easily. We denote the benchmark constant-mix strategy by $\pi_B(t)=(0,\pi^{CM})^\top$ for all $\tin$. Hence, the constant-mix strategy consists of investing $1-\pi^{CM}$ of the wealth in $S_0$ and $\pi^{CM}$ in $S_2$. So the price dynamics of the benchmark constant-mix portfolio\footnote{ Deviating from a constant-mix investment strategy for the benchmark portfolio makes the overall problem much more challenging, as then we may not be able to analytically derive the fair price of the option, which should be determined at the product inception date $t = 0$. Estimating the option price via Monte-Carlo simulations would be possible.} is given by
    \begin{align*}
    	dV^{v_I,\pi_B}(t)=&(1-\pi^{CM})V^{v_I,\pi_B}(t)\frac{dS_0(t)}{S_0(t)}+\pi^{CM}V^{v_I,\pi_B}(t)\frac{dS_2(t)}{S_2(t)}\\
    	=&V^{v_I,\pi_B}(t)((r+\pi^{CM}(\mu_2-r))dt+\pi^{CM}\sigma_2(\rho dW_1(t)+\sqrt{1-\rho^2}dW_2(t))),\\
    	V^{v_I,\pi_B}(0)=&v_I.
    \end{align*}
    
    Note that the benchmark portfolio is in general not equal to the individual investment strategy but they should have high correlation so that the reinsurance indeed provides downside protection to the insurer. {\color{black} In practice, the parties choose a suitable benchmark portfolio while negotiating the reinsurance terms. For the insurer, the benchmark portfolio must resemble its individual portfolio, whereas for the reinsurer the benchmark portfolio must be transparent, simple to hedge and convenient to price. For example, if the insurer knows its long-term target position in $S_1$, then $\pi^{CM}$ can be chosen to this level, given that $S_2$ has a high correlation with $S_1$ and is hedgeable for the reinsurer. Another aspect that can be considered for choosing the benchmark portfolio is how costly and complicated it is to hedge the put option for the reinsurer.}
    
    The reinsurance, which is modeled by a put option $P$ with a payoff given by
    \begin{align*}
        P(T)=(G_T-V^{v_I,\pi_B}(T))^+.
    \end{align*}
    The fair price of the put option $P$ at time $t$ in the basic financial market is calculated by
    \begin{equation*}
        P(t)=\widetilde{Z}(t)^{-1}\E{\widetilde{Z}(T)(G_T-V^{v_I,\pi_B}(T))^+\big|\mathcal{F}_t}.
    \end{equation*}
    
    We assume that the insurer can buy $\xi_I$ reinsurance contracts (put options) and call $\xi_I$ the reinsurance strategy of the insurer. Furthermore, we assume that the reinsurance contract is priced using the following principle: the price is given by $(1+\theta_R)P(0)$ at time $0$, where $\theta_R$ is the safety loading chosen by the reinsurer and $P(0)$ the fair price of the put option in the basic financial market\footnote{Note that this is almost the expected value premium principle. The expected value premium principle is given by $(1+\theta)\mathbb{E}[X]$, where $\theta$ is the safety loading and $\mathbb{E}[X]$ the expected non-discounted loss under the real-world measure. In our case, the reinsurance premium is given by the expected discounted loss of the benchmark portfolio under the risk-neutral measure. Note that the expected loss of the benchmark portfolio is not the expected loss of the insurer.}. The safety loading $\theta_R$ is assumed to satisfy $0 \leq \theta_R \leq \theta^{\max}$. Given an admissible $\theta_R$, the reinsurance amount is assumed to satisfy
    \begin{align*}
        0 \leq \xi_I \leq \xi^{\max} :=\min\bigg\{\frac{v_I}{(1+\theta_R)P(0)}, \bar{\xi}\bigg\},
    \end{align*}
    where $\bar{\xi}>0$ is a constant independent of $\theta_R$. The fraction in the definition of $\xi^{\max}$ is required to ensure that the insurer has enough initial capital to buy the desired amount of reinsurance. As for $\bar{\xi}$, it should be close to $1$ to prevent the insurer from speculating with the reinsurance by buying an excessive amount of it, which may be required by regulatory institutions in practice.
    
    We denote the insurer's relative portfolio process by $\pi_I(t)=(\pi_{I,1}(t),\pi_{I,2}(t))^\top$, $\tin$. The wealth process of the insurer is given by
    \begin{align}
    	dV_I^{v_{I,0}(\xi_I,\theta_R),\pi_I}(t)=&(1-\pi_{I,1}(t)-\pi_{I,2}(t))V_I^{v_{I,0}(\xi_I,\theta_R),\pi_I}(t)\frac{dS_0(t)}{S_0(t)}\label{IWP1}\\
    	&+\pi_{I,1}(t)V_I^{v_{I,0}(\xi_I,\theta_R),\pi_I}(t)\frac{dS_1(t)}{S_1(t)}+\pi_{I,2}(t)V_I^{v_{I,0}(\xi_I,\theta_R),\pi_I}(t)\frac{dS_2(t)}{S_2(t)},\nonumber\\
    	V_I^{v_{I,0}(\xi_I,\theta_R),\pi_I}(0)=&v_I-\xi_I(1+\theta_R)P(0)=:v_{I,0}(\xi_I,\theta_R). \nonumber
    \end{align}
    
    The total terminal wealth of the insurer is given by the insurer's terminal wealth plus the payment from the reinsurance:
    \begin{align*}
        V^{v_{I,0}(\xi_I,\theta_R),\pi_I}(T)+\xi_IP(T).
    \end{align*}
    
    Similarly to the insurer, the reinsurer's relative portfolio process is given by $\pi_R=(\pi_{R,1}(t),\pi_{R,2}(t))^\top$, $\tin$. Therefore, the reinsurer's wealth process is
    \begin{align}
    	dV_R^{v_{R,0}(\xi_I,\theta_R),\pi_R}(t)=&(1-\pi_{R,1}(t)-\pi_{R,2}(t))V_R^{v_{R,0}(\xi_I,\theta_R),\pi_R}(t)\frac{dS_0(t)}{S_0(t)}\label{RWP1}\\
    	&+\pi_{R,1}(t)V_R^{v_{R,0}(\xi_I,\theta_R),\pi_R}(t)\frac{dS_1(t)}{S_1(t)}+\pi_{R,2}(t)V_R^{v_{R,0}(\xi_I,\theta_R),\pi_R}(t)\frac{dS_2(t)}{S_2(t)},\nonumber\\
    	V_R^{v_{R,0}(\xi_I,\theta_R),\pi_R}(0)=&v_{R}+\xi_I(1+\theta_R)P(0)=:v_{R,0}(\xi_I,\theta_R), \nonumber
    \end{align}
    where $v_R>0$ is the initial wealth of the reinsurer before issuance of a reinsurance contract. The total terminal wealth of the reinsurer equals the terminal wealth less the reinsurance payoff:
    \begin{align*}
        V^{v_{R,0}(\xi_I,\theta_R),\pi_R}(T)-\xi_IP(T).
    \end{align*}
    
    We denote by $U_R$ and $U_I$ the utility functions of the reinsurer and the insurer, respectively. The utility functions are assumed to be strictly concave, continuously differentiable, strictly increasing and to fulfill the Inada conditions, i.e., for $i\in\{R,I\}$ it holds
    \begin{align*}
        U_i'(0):=\lim_{x\downarrow 0}U_i'(x)=\infty\text{ and }U_i'(\infty)=\lim_{x\uparrow \infty}U_i'(x)=0.
    \end{align*}
    
    The set of admissible strategies of the insurer is defined in the following way:
    \begin{align*}
    	\Lambda_I:=\{(\pi_I,\xi_I)|\;&\pi_I\text{ self-financing, }\pi_I(t)\in K\;\mathbb{Q}\text{-a.s. }\forall t\in[0,T],\;\xi_I\in[0,\xi^{\max}],\\
    	&V^{v_{I,0}(\xi_I,\theta_R),\pi_I}_I(t)\geq0\;\mathbb{Q}\text{-a.s. }\forall t\in[0,T]\text{ and }\\
    	&\mathbb{E}[U_I(V_I^{v_{I,0}(\xi_I,\theta_R),\pi_I}(T)+\xi_IP(T))^-]<\infty]\},
    \end{align*}
    where $K:=\mathbb{R}\times\{0\}$ defines the constraint on the insurer's investment strategy and models the insurer's preference to invest in the individual risky asset $S_1$ only, $x^-:=\max\{-x,0\}$ for $x\in\mathbb{R}$. 
    
    The set of admissible strategies of the reinsurer is defined by
    \begin{align*}
    	\Lambda_R:=\{(\pi_R,\theta_R)|\;&\pi_R\text{ self-financing, }\theta_R\in[0,\theta^{\max}],\\
    	&V_R^{v_{R,0}(\xi_I,\theta_R),\pi_R}(t)\geq0\;\mathbb{Q}\text{-a.s. }\forall t\in[0,T]\text{ and }\\
    	&\mathbb{E}[U_R(V_R^{v_{R,0}(\xi_I,\theta_R),\pi_R}(T)-\xi_IP(T))^-]<\infty\}.
    \end{align*}
    
    \textcolor{black}{The definitions of admissible strategies ensure that the considered optimization problems are well defined. Self-financing investment strategy means that no money is injected into or extracted from the portfolio during the investment period. The finiteness of the expectation ensures that the expected terminal utility exists. For more information, please, see \citeA{korn1999optimal}, \citeA{cvitanic1992convex}.}
    
    The aim of each party is to maximize its expected utility of the total terminal wealth. Hence, the Stackelberg game between the reinsurer (leader) and the insurer (follower) is given by
    \begin{align}
        \sup_{(\pi_R,\theta_R)\in\Lambda_R}&\mathbb{E}[U_R(V_R^{v_{R,0}(\xi_I^\ast(\theta_R),\theta_R),\pi_R}(T)-\xi_I^\ast(\theta_R)P(T))] \label{SG}\tag{SG}\\
		&\text{s.t. }(\xi_I^\ast(\theta_R),\pi_I^\ast(\cdot|\theta_R))\in\underset{(\pi_I,\xi_I)\in\Lambda_I}{\text{arg max }}\mathbb{E}[U_I(V_I^{v_{I,0}(\xi_I,\theta_R),\pi_I}(T)+\xi_IP(T))]. \nonumber
    \end{align}

\section{Solution to the Stackelberg game}\label{sec_solution_to_SG}

\hspace{0.5cm} In this section, we solve the Stackelberg game \eqref{SG} using the idea of backward induction. First, we solve the insurer's optimization problem
    \begin{align}
        \sup_{(\pi_I,\xi_I)\in\Lambda_I}\mathbb{E}[U_I(V_I^{v_{I,0}(\xi_I,\theta_R),\pi_I}(T)+\xi_IP(T))] \label{OPI}\tag{$\mathcal{P}_I$}
    \end{align}
    for each admissible safety loading $\theta_R\in[0,\theta^{\max}]$, i.e., we find the optimal response\\ $(\pi_I^\ast(\cdot|\theta_R),\xi_I^\ast(\theta_R))\in\Lambda_I$ of the insurer. Second, knowing the responses of the insurer, we solve the optimization problem of the reinsurer
    \begin{align}
        \sup_{(\pi_R,\theta_R)\in\Lambda_R}&\mathbb{E}[U_R(V_R^{v_{R,0}(\xi_I^\ast(\theta_R),\theta_R),\pi_R}(T)-\xi_I^\ast(\theta_R)P(T))], \label{OPR}\tag{$\mathcal{P}_R^{\pi_R,\theta_R}$}
    \end{align}
    i.e., we find the optimal strategy $(\pi_R^\ast,\theta_R^\ast)\in\Lambda_R$. The Stackelberg equilibrium in the game \eqref{SG} is then given by the solution $(\pi_R^\ast(\cdot),\theta_R^\ast,\pi_I^\ast(\cdot|\theta_R^\ast),\xi_I^\ast(\theta_R^\ast))$, as it fulfills \eqref{eq:SE_condition_1} and \eqref{eq:SE_condition_2}.

\subsection{Solution to the optimization problem of the insurer} \label{subsec_solution_to_OPI}

\hspace{0.5cm} In this subsection, we solve the optimization problem of the insurer for each admissible safety loading $\theta_R\in[0,\theta^{\max}]$. We face two challenges in solving the optimization problem \eqref{OPI}: the insurer cannot invest in $S_2$, hence, the optimization problem \eqref{OPI} is constrained and it has a fixed-term investment in the put option.
    
    The procedure of solving the optimization problem of the insurer \eqref{OPI} is as follows. First, we consider a family of auxiliary financial markets as per \citeA{cvitanic1992convex}. In each auxiliary market we can solve an unconstrained optimization problem with the generalized martingale method introduced by \citeA{desmettre2016optimal}, i.e., we find the optimal investment strategy and the optimal reinsurance strategy. Second, we find the optimal auxiliary market for which the solution to the unconstrained optimization problem coincides with the solution to the constrained optimization problem.

    \textbf{Auxiliary Market.} The financial market consisting of $S_0$, $S_1$ and $S_2$ is referred to as the basic financial market $\mathcal{M}$. Next, we introduce the concept of an auxiliary market as per \citeA{cvitanic1992convex}. Consider a convex set $K=\mathbb{R}\times\{0\}$. Its support function is given by $\delta:\mathbb{R}^2\to\mathbb{R}\cup\{+\infty\}$ with
    \begin{align*}
    	\delta(x):=-\inf_{y\in K}(x^\top y)=-\inf_{y_1\in\mathbb{R}}(x_1y_1)=\begin{cases}
    		0, & \text{if }x_1=0,\\
    		+\infty, &\text{otherwise}.
    	\end{cases}
    \end{align*}
    
    The barrier cone of $K$ is defined by
    \begin{align*}
    	\tilde{K}:=\{x\in\mathbb{R}^2|\delta(x)<+\infty\}=\{x\in\mathbb{R}^2|x_1=0\}=\{0\}\times\mathbb{R}.
    \end{align*}
    
    For $x\in\tilde{K}$ we have $\delta(x)=0$. Let the class of $\mathbb{R}^2$-valued dual processes be given by
    \begin{align*}
    	\mathcal{D}:=\bigg\{\lambda=(\lambda(t))_{t\in[0,T]}\text{ prog. measurable}\bigg|\mathbb{E}\bigg[\int_0^T\lVert\lambda(t)\rVert^2dt\bigg]<\infty,\;\mathbb{E}\bigg[\int_0^T\delta(\lambda(t))dt\bigg]<\infty\bigg\}.
    \end{align*}
    
    It holds for $\lambda\in\mathcal{D}$ that $\lambda(t)\in\tilde{K}$ $\mathbb{Q}$-a.s for all $t\in[0,T]$, i.e., $\lambda_1(t)=0$ $\mathbb{Q}$-a.s. for all $t\in[0,T]$.\\
    
    Let $\lambda\in\mathcal{D}$. In the auxiliary market $\mathcal{M}_\lambda$ the dynamics of assets is given by:
    \begin{align*}
    	dS_0^\lambda(t)=&S_0^\lambda(t)(r+\delta(\lambda(t)))dt=S_0^\lambda(t)rdt,\\
    	dS_1^\lambda(t)=&S_1^\lambda(t)[(\mu_1+\lambda_1(t)+\delta(\lambda(t)))dt+\sigma_1dW_1(t)]=S_1^\lambda(t)(\mu_1dt+\sigma_1dW_1(t)),\\
    	dS_2^\lambda(t)=&S_2^\lambda(t)[(\mu_2+\lambda_2(t)+\delta(\lambda(t)))dt+\sigma_2(\rho dW_1(t)+\sqrt{1-\rho^2}dW_2(t))]\\
    	=&S_2^\lambda(t)[(\mu_2+\lambda_2(t))dt+\sigma_2(\rho dW_1(t)+\sqrt{1-\rho^2}dW_2(t))],
    \end{align*}
    since $\delta(\lambda(t))=0$ and $\lambda_1(t)=0$.
    We denote the market price of risk and the discount factor (also called pricing kernel) in the auxiliary market by
    \begin{align*}
    	\gamma_\lambda(t):=\gamma+\sigma^{-1}\lambda(t)\,\,\text{and}\,\,
    	\widetilde{Z}_\lambda(t):=\exp\bigg(-rt-\frac{1}{2}\int_0^t\lVert\gamma_\lambda(s)\rVert^2ds-\int_0^t\gamma_\lambda(s)^\top dW(s)\bigg).
    \end{align*}
    
    To apply the generalized martingale method introduced by \citeA{desmettre2016optimal}, we only need the stochastic payoff of the fixed-term security and its price at time $0$, which is not necessarily a price based on the risk-neutral valuation. Therefore, we are only interested in $P(T)$ and $P(0)$. $P(T)$ is a random variable that is $\mathcal{F}_T$-measurable. By risk-neutral valuation, $P(T) = (G_T-V^{v_I,\pi_B}(T))^+$ and $P(0) = \mathbb{E}\left[{\widetilde{Z}(T)(G_T-V^{v_I,\pi_B}(T))^+}\right]$. Note that in general $P(0) = \mathbb{E}[\widetilde{Z}(T)P(T)] \neq \mathbb{E}[\widetilde{Z}_\lambda(T)P(T)]$.
    
    
    The wealth process of the insurer $V^{v_{I,0}(\xi_I,\theta_R),\pi_I}_\lambda$ in $\mathcal{M}_\lambda$ is given by
    \begin{align}
    	dV^{v_{I,0}(\xi_I,\theta_R),\pi_I}_\lambda(t)=&(1-\pi_{I,1}(t)-\pi_{I,2}(t))V^{v_{I,0}(\xi_I,\theta_R),\pi_I}_\lambda(t)\frac{dS_0^\lambda(t)}{S_0^\lambda(t)}\\ 
    	&+\pi_{I,1}(t)V^{v_{I,0}(\xi_I,\theta_R),\pi_I}_\lambda(t)\frac{dS_1^\lambda(t)}{S_1^\lambda(t)}+\pi_{I,2}(t)V^{v_{I,0}(\xi_I,\theta_R),\pi_I}_\lambda(t)\frac{dS_2^\lambda(t)}{S_2^\lambda(t)}\nonumber\\
    	=&(1-\pi_{I,1}(t)-\pi_{I,2}(t))V^{v_{I,0}(\xi_I,\theta_R),\pi_I}_\lambda(t)\frac{dS_0(t)}{S_0(t)}\\ 
    	&+\pi_{I,1}(t)V^{v_{I,0}(\xi_I,\theta_R),\pi_I}_\lambda(t)\frac{dS_1(t)}{S_1(t)}+\pi_{I,2}(t)V^{v_{I,0}(\xi_I,\theta_R),\pi_I}_\lambda(t)\frac{dS_2(t)}{S_2(t)}\nonumber\\
    	&+\underbrace{V^{v_{I,0}(\xi_I,\theta_R),\pi_I}_\lambda(t)\pi_I(t)^\top\lambda(t)dt}_{\text{additional term}}, \nonumber\\
    	V^{v_{I,0}(\xi_I,\theta_R),\pi_I}_\lambda(0)&=v_I-\xi_I(1+\theta_R)P(0)=v_{I,0}(\xi_I,\theta_R). \nonumber
    \end{align}
    
    The unconstrained optimization problem of the insurer in $\mathcal{M}_\lambda$ is given by
    \begin{align} \label{IOP2}
    	\sup_{(\pi_I,\xi_I)\in\Lambda_I^\lambda}\mathbb{E}[U_I(V^{v_{I,0}(\xi_I,\theta_R),\pi_I}_\lambda(T)+\xi_I P(T))], \tag{$\mathcal{P}_I^\lambda$}
    \end{align}
    where
    \begin{align*}
    	\Lambda_I^\lambda:=\{(\pi_I,\xi_I)|\;&\pi_I\text{ self-financing, }V^{v_{I,0}(\xi_I,\theta_R),\pi_I}_\lambda(t)\geq0\;\mathbb{Q}\text{-a.s. }\forall t\in[0,T],\\
    	&\xi_I\in[0,\xi^{\max}]\text{ and }\mathbb{E}[U_I(V^{v_{I,0}(\xi_I,\theta_R),\pi_I}_\lambda(T)+\xi_I P(T))^-]<\infty\}.
    \end{align*}
    
    Note that  \eqref{IOP2} does not have the investment strategy constraint as in \eqref{OPI}, but for the optimal $\lambda^{\ast}$, which has to be found, the solution to the unconstrained problem will satisfy the constraint $\pi_2 = 0$. We denote the solution to \eqref{IOP2} by 
    
    \begin{equation*}
        (\pi_\lambda^\ast,\xi_\lambda^\ast)\in\underset{(\pi_I,\xi_I)\in\Lambda_I^\lambda}{\text{arg sup }}\mathbb{E}[U_I(V^{v_{I,0}(\xi_I,\theta_R),\pi_I}_\lambda(T)+\xi_I P(T))].
    \end{equation*}

    \textbf{Random utility function.} As in \citeA{desmettre2016optimal}, we define the random utility function by
    \begin{align*}
    	\hat{U}_I(x):=U_I(x+\xi_I P(T))
    \end{align*}
    for $x\in [0,\infty)$, where $\xi_I\in[0,\xi^\text{max}]$. The utility function $\hat{U}_I$ is random, since $P(T)$ is a random variable. Hence, $\hat{U}_I:[0,+\infty)\to[U_I(\xi_I P(T)),+\infty)$ and $\hat{U}_I$ is continuously differentiable, strictly increasing and strictly concave. Therefore, it holds $\hat{U}_I':[0,+\infty)\to(0,U_I'(\xi_I P(T))]$ and 
    \begin{align*}
    	\hat{U}_I'(x)=U_I'(x+\xi_I P(T)).
    \end{align*}
    
    We denote the inverse function of $\hat{U}_I'$ by $\hat{I}_I:(0,+\infty)\to[0,+\infty)$. For $y\in(0,U_I'(\xi_I P(T))]$ it is given by $I_I(y)-\xi_I P(T)$, where $I_I$ denotes the inverse of $U_I'$. For $y>U_I'(\xi_I P(T))$ we set $\hat{I}(y):=0$. Hence, the random inverse function $\hat{I}_I$ is bijective on $(0,U_I'(\xi_I P(T))]$.

    \begin{proposition}[Optimal solution to \eqref{IOP2}] \label{SolIOP2}
    	Assume that for all $y\in(0,\infty)$
    	\begin{align*}
    		\mathbb{E}[\widetilde{Z}_\lambda(T)I_I(y\widetilde{Z}_\lambda(T))]<\infty\text{ and }\mathbb{E}[U_I(I_I(y\widetilde{Z}_\lambda(T)))]<\infty
    	\end{align*}
    	holds. Then, there exists a solution $(\pi_\lambda^\ast,\xi_\lambda^\ast)$ to \eqref{IOP2}, where
    	\begin{align*}
    		\xi_\lambda^\ast\in\underset{\xi_I\in[0,\xi^\text{max}]}{\text{arg max }}\nu(\xi_I).
    	\end{align*}
    	The function $\nu$ is given by
    	\begin{align*}
    		\nu(\xi_I):=\mathbb{E}[U_I(\max\{I_I(y^\ast(\xi_I)\widetilde{Z}_\lambda(T)),\xi_I P(T)\})],
    	\end{align*}
    	where the Lagrange multiplier $y^\ast:= y^\ast(\xi_I)$ is given by the budget constraint
    	\begin{align*}
    		\mathbb{E}[\widetilde{Z}_\lambda(T)\hat{I}_I(y^\ast\widetilde{Z}_\lambda(T))]=v_I-\xi_I(1+\theta_R)P(0).
    	\end{align*}
    	The optimal terminal wealth $V^\ast_\lambda(T):=V^{v_{I,0}^\lambda(\xi_{\lambda}^\ast,\theta_R),\pi_\lambda^\ast}_\lambda(T)$ is given by
    	\begin{align*}
    		V^\ast_\lambda(T) = \hat{I}_I(y^\ast(\xi_\lambda^\ast)\widetilde{Z}_\lambda(T))=\max\{I_I(y^\ast(\xi_\lambda^\ast)\widetilde{Z}_\lambda(T))-\xi_\lambda^\ast P(T),0\}
    	\end{align*}
    	and the optimal wealth process $V_\lambda^\ast$ is given by
    	\begin{align*}
    	    V_\lambda^\ast(t)=\widetilde{Z}_\lambda(t)^{-1}\mathbb{E}[\widetilde{Z}_\lambda(T)V_\lambda^\ast(T)|\mathcal{F}_t] \quad \text{for}\quad \tin.
    	\end{align*}\\
    	If $\hat{I}_I$ and $\frac{d\hat{I}_I(y)}{dy}$ are polynomially bounded\footnote{As per \citeA{desmettre2016optimal}, a function $h:(0,\infty)\to\mathbb{R}$ is called polynomially bounded at $0$ and $+\infty$ if there exist $c,k\in(0,\infty)$ such that for all $y\in(0,\infty)$ \begin{align*}|h(y)|\leq c\bigg(y+\frac{1}{y}\bigg)^k\end{align*}} at $0$ and $\infty$, the optimal $\pi_\lambda^\ast$ is given by
    	\begin{align}\label{OptimalPortfolioProcess_Lambda}
    	    \pi_\lambda^\ast(t)V_\lambda^\ast(t)=-(\sigma^\top)^{-1}\gamma_\lambda\widetilde{Z}_\lambda(t)^{-1}\mathbb{E}\bigg[\widetilde{Z}_\lambda(T)y^\ast(\xi_\lambda^\ast)\widetilde{Z}_\lambda(T)\frac{d\hat{I}}{dy}(y^\ast(\xi_I^\ast)\widetilde{Z}_\lambda(T))\bigg|\mathcal{F}_t\bigg]
    	\end{align}
    	$\mathbb{Q}$-a.s. for all $\tin$.
    \end{proposition}
    
    \begin{proof}
        See Appendix \ref{app:proofs_main}.
    \end{proof}
    
    \textcolor{black}{The assumption of the finiteness of expectations in Proposition \ref{SolIOP2} is standard (cf. \citeA{desmettre2016optimal}, \citeA{cvitanic1992convex}) and ensures that the optimal Lagrange multiplier $y^\ast(\xi_I) > 0$ (given an arbitrary but fixed $\xi_I$) in the corresponding problem exists.}
    
    Hence, by Proposition \ref{SolIOP2} we know that under some conditions there exists the optimal solution $(\pi_\lambda^\ast,\xi_\lambda^\ast)$ to the unconstrained optimization problem of the insurer \eqref{IOP2}. In the next proposition, we show how the solutions of the optimization problem of the insurer \eqref{OPI} and the unconstrained optimization problem of the insurer \eqref{IOP2} are linked.
    
    \begin{proposition}[Optimal solution to \eqref{OPI}] \label{SolOPI}
    	Suppose that there exists $\lambda^\ast\in\mathcal{D}$ such that for the optimal solution $(\pi_{\lambda^\ast}^\ast,\xi_{\lambda^\ast}^\ast)$ to $(P_I^{\lambda^\ast})$ we have $\pi_{\lambda^\ast}^\ast(t)\in K$ $\mathbb{Q}$-a.s. for all $t\in[0,T]$. Then $(\pi_I^\ast,\xi_I^\ast):=(\pi_{\lambda^\ast}^\ast,\xi_{\lambda^\ast}^\ast)$ is optimal for the constrained optimization problem of the insurer \eqref{OPI}.
    \end{proposition}
    \begin{proof}
    	See Appendix \ref{app:proofs_main}.
    \end{proof}
    
    \textbf{Remark to Proposition \ref{SolOPI}.} In the case of a deterministic utility function, $\lambda^\ast$ can be found by a minimization criterion, which is given in Example 15.1 of \citeA{cvitanic1992convex}. As we have a random utility function, we cannot apply this minimization criterion, and the search for the optimal $\lambda^\ast$ is more involved. In Section \ref{sec_example}, we have a concrete power utility function for the insurer and we find the explicit form of $\lambda^\ast$ that satisfies Proposition \ref{SolOPI}.

\subsection{Solution to the optimization problem of the reinsurer}\label{subsec_solution_to_OPR}

\hspace{0.5cm} We now know the optimal responses of the insurer to each possible safety loading $\theta_R\in[0,\theta^{\max}]$. We write $(\xi_I^\ast(\theta_R),\pi_I^\ast(t|\theta_R,\xi^\ast_I(\theta_R)))$ to emphasize the dependence of the insurer's optimal strategy on $\theta_R\in[0,\theta^{\max}]$. To derive the Stackelberg equilibrium, we need to solve the optimization problem \eqref{OPR} of the reinsurer. In contrast to \eqref{OPI}, we face only one challenge in \eqref{OPR}: the reinsurer has a fixed short position in the put option. 
    
    In comparison to the insurer, the reinsurer has a larger investment universe and can hedge the put option. Thus, we adapt the approach proposed by \citeA{korn1999optimal} based on replicating strategies. In this approach, we work mainly with trading strategies instead of the relative portfolio processes. A trading strategy is an $\mathbb{R}^{3}$-valued progressively measurable stochastic process $\varphi(t) = (\varphi_0(t), \varphi_1(t), \varphi_2(t))',\,\tin,$ such that $\mathbb{Q}$-a.s.:
    \begin{equation*}
        \int \limits_{0}^{T}|\varphi_0(t)|\,ds <+\infty \quad \text{and}\quad \int \limits_{0}^{T}\varphi_i^2(t)\,ds <+\infty \,\,, i \in \{1,2\}.
    \end{equation*}
    We consider only self-financing trading strategies, which means that they satisfy:
    \begin{equation*}
       \underbrace{\varphi_0(t) S_0(t) + \varphi_1(t) S_1(t) + \varphi_2(t) S_2(t)}_{=:  V(\varphi, t)} = V(\varphi, 0) + \int \limits_{0}^{t}\varphi_0(s)\,d S_0(s) + \int \limits_{0}^{t}\varphi_1(s)\,d S_1(s) + \int \limits_{0}^{t}\varphi_2(s)\,d S_2(s),
    \end{equation*}
    where $V(\varphi, t)$ is the portfolio value and $V(\varphi, 0)$ is the initial wealth of the decision maker. We denote the trading strategy of the reinsurer by $\varphi_R$. If $\varphi_R$ is a self-financing strategy with $V_R(\varphi_R, t)  > 0$ $\mathbb{Q}$ -a.s. $\forall \tin$, it is linked to the relative portfolio process $\pi_R$ via the formula: 
    \begin{align} \label{RelationPortfolioTradingReinsurer}
    	\pi_{R,i}(t)=\frac{\varphi_{R,i}(t)\cdot S_i(t)}{V_R(\varphi_R, t)},\;\forall\tin,i\in\{0,1,2\},
    \end{align}
    with $V_R(\varphi_R, t) = V^{v_{R,0}(\xi_I^\ast(\theta_R),\theta_R),\pi_R}_R(t)$ $\mathbb{Q}$ -a.s. for all $\tin$.
    
    We solve the optimization problem of the reinsurer \eqref{OPR} as follows. First, we transform the optimization problem \eqref{OPR} with respect to (w.r.t.) the portfolio process $\pi_R$ and $\theta_R$ into an optimization problem \eqref{OPRa} w.r.t. the trading strategy $\varphi_R$ and $\theta_R$. Second, we use the idea of \citeA{desmettre2016optimal} and solve the optimization problem \eqref{OPRa} given a fixed $\theta_R\in[0,\theta^{\max}]$. For this, we transform the optimization problem \eqref{OPRa} into the optimization problem \eqref{OPRb} w.r.t. the trading strategy $\varphi_R$ given fixed $\theta_R$ and afterwards transform the optimization problem \eqref{OPRb} into the optimization problem \eqref{OPRc} w.r.t. the trading strategy $\varphi_R$ given fixed $\theta_R$ and a fixed position  $\xi(t)=-\xi_I^\ast(\theta_R)$ in the put option $P$. Then, we solve the optimization problem \eqref{OPRc} by applying the idea of \citeA{korn1999optimal} and use \eqref{RelationPortfolioTradingReinsurer} to calculate $\pi_R^\ast(\cdot|\theta_R)$ from the optimal trading strategy $\varphi_R^\ast(\cdot|\theta_R)$. The optimization problems \eqref{OPRb} and \eqref{OPRc} are stated in Appendix \ref{app:proofs_main} in the proof of Proposition \ref{ThmOpR1} to keep this subsection clearer. Lastly, we solve the optimization problem \eqref{OPR} \textcolor{black}{w.r.t.} $\theta_R$ for the given optimal portfolio process $\pi_R^\ast(\cdot|\theta_R)$.
    
    Therefore, we first transform the optimization problem \eqref{OPR} \textcolor{black}{w.r.t.} the portfolio process $\pi_R$ into the optimization problem \eqref{OPRa} \textcolor{black}{w.r.t.} the trading strategy $\varphi_R$. We use the notation $V^{v_{I,0}(\xi_I,\theta_R),\varphi_R}_R$ for the wealth process of the reinsurer \textcolor{black}{w.r.t.} the trading strategy $\varphi_R$. Hence, the optimization problem of the reinsurer \textcolor{black}{w.r.t.} the trading strategy $\varphi_R$ is given by
    \begin{align}\label{OPRa}
    	\sup_{(\varphi_R,\theta_R)\in\Lambda_R^{\varphi_R}}\mathbb{E}[U_R(V^{v_{R,0}(\xi_I^\ast(\theta_R),\theta_R),\varphi_R}_R(T)-\xi_I^\ast(\theta_R)P(T))], \tag{$\mathcal{P}_R^{\varphi_R,\theta_R}$}
    \end{align}
    where $\Lambda_R^{\varphi_R}$ is the set of all admissible trading strategies and safety loadings of the reinsurer:
    \begin{align*}
    	\Lambda_R^{\varphi_R}:=\{(\varphi_R,\theta_R)|\;&\varphi_R\text{ self-financing, }V^{v_{R,0}(\xi_I^\ast(\theta_R),\theta_R),\varphi_R}_R(t)\geq0\;\mathbb{Q}\text{-a.s. }\forall t\in[0,T],\\
    	&\theta_R\in[0,\theta^{\max}]\text{ and }\mathbb{E}[U_R(V^{v_{R,0}(\xi_I^\ast(\theta_R),\theta_R),\varphi_R}_R(T)-\xi_I^\ast(\theta_R)P(T))^-]<\infty\}.
    \end{align*}
    Let $\theta_R\in[0,\theta^{\max}]$ be arbitrary but fixed. Our aim is to solve the optimization problem 
    \begin{align}\label{OPRb}
    	\sup_{\varphi_R:(\varphi_R,\theta_R)\in\Lambda_R^{\varphi_R}}\mathbb{E}[U_R(V^{v_{R,0}(\xi_I^\ast(\theta_R),\theta_R),\varphi_R}_R(T)-\xi_I^\ast(\theta_R)P(T))], \tag{$\mathcal{P}_R^{\varphi_R|\theta_R}$}
    \end{align}
    where $\theta_R$ in \ref{OPRb} emphasizes that $\theta_R$ is not a control variable there. To find the optimal trading strategy $\varphi_R^\ast$ that solves \eqref{OPRb} we adjust Theorem 5.1 mentioned in Remark 5.2 in \citeA{korn1999optimal}.
    
    \begin{proposition}[Optimal solution to \eqref{OPRb}] \label{ThmOpR1}
        Assume that it holds for all $y\in(0,\infty)$
        \begin{align*}
            \mathbb{E}[\widetilde{Z}(T)I_R(y\widetilde{Z}(T))]<\infty,
        \end{align*}
        where $I_R$ is the inverse function of $U'_R$.
    	\begin{itemize}
    		\item[(a)] There exists an optimal trading strategy $\varphi_R^\ast$ for the optimization problem \eqref{OPRb}. The optimal total terminal wealth in the optimization problem \eqref{OPRb} is given by
    		\begin{align*}
    			V^{v_{R,0}(\xi_I^\ast(\theta_R),\theta_R),\varphi_R^\ast}_R(T)-\xi_I^\ast(\theta_R)P(T)=I_R(y^\ast_R(\theta_R)\widetilde{Z}(T)),
    		\end{align*}
    		where $y^\ast_R\equiv y^\ast_R(\theta_R)$ is the Lagrange multiplier determined by
    		\begin{align}
    			\mathbb{E}[\widetilde{Z}(T)I_R(y^\ast_R\widetilde{Z}(T))]=v_{R}+\xi_I^\ast(\theta_R)\theta_RP(0). \label{LagrangeMultiplierR}
    		\end{align}
    		\item[(b)] Let $\zeta_R^\ast$ be the optimal trading strategy of the stock optimization problem, where the reinsurer only invests in the assets $S_0$, $S_1$ and $S_2$ (i.e., no investment in options). Then, the optimal trading strategy $\varphi_R^\ast$ to the optimization problem \eqref{OPRb} is given by
    		\begin{align*}
    			\varphi_{R,0}^\ast(t)&=\frac{V^{v_{R,0}(\xi_I^\ast(\theta_R),\theta_R),\varphi_R^\ast}(t)-\sum_{i=1}^2\varphi^\ast_{R,i}(t)S_i(t)}{S_0(t)},\quad \varphi_{R,1}^\ast(t)=\zeta_{R,1}^\ast(t);\\
    			\varphi_{R,2}^\ast(t)&=\zeta_{R,2}^\ast(t) + \xi_I^\ast(\theta_R) \frac{\pi^{CM}V^{v_I,\pi_B}(t)(\Phi(d_+)-1)}{S_2(t)},
    		\end{align*}
    		where
    		\begin{align*}
		d_+:=d_+(t,V^{v_I,\pi_B}(t)):=\frac{\ln\Big(\frac{V^{v_I,\pi_B}(t)}{G_T}\Big)+\big(r+\frac{1}{2}(\sigma_2\pi^{CM})^2\big)(T-t)}{\pi^{CM}\sigma_2\sqrt{T-t}}.
	    \end{align*}
    	\end{itemize}	
    \end{proposition}
    \begin{proof}
        See Appendix \ref{app:proofs_main}.
    \end{proof}
    \textbf{Remark to Proposition \ref{ThmOpR1}.}
    \textcolor{black}{The assumption of the finiteness of expectations in Proposition \ref{ThmOpR1} is standard (cf. \citeA{korn1999optimal}) and ensures the existence of the optimal Lagrange multiplier $y^\ast_R > 0$.}
    \\
    
    Proposition \ref{ThmOpR1} yields the the optimal terminal value     $V^{v_{R,0}(\xi_I^\ast(\theta_R),\theta_R),\varphi_R^\ast}_R(T)$ and the optimal trading strategy $\varphi_R^\ast$ in the reinsurer's optimization problem \eqref{OPRa} for an arbitrary but fixed $\theta_R$. The next proposition determines the optimal safety loading.
    \begin{proposition}[Optimal safety loading] \label{OpSt}
    	Let $\varphi^\ast_R(\cdot|\theta_R)$ be the optimal trading strategy for \eqref{OPRa} given an arbitrary but fixed $\theta_R\in[0,\theta^\text{max}]$. Then, the optimal $\theta_R^\ast$ is given by
    	\begin{align*}
    		\theta_R^\ast=\underset{\theta_R\in[0,\theta^{\max}]}{\text{arg max }}\mathbb{E}[U_R(V^{v_{R,0}(\xi_I^\ast(\theta_R),\theta_R),\varphi^\ast_R}_R(T)-\xi_I^\ast(\theta_R)P(T))].
    	\end{align*}
    \end{proposition}
    \begin{proof}
        See Appendix \ref{app:proofs_main}.
    \end{proof}
    
\subsection{Stackelberg equilibrium}
    \begin{proposition}[Stackelberg equilibrium] \label{SE}
    	The Stackelberg equilibrium of the Stackelberg game \eqref{SG} is given by $(\pi_R^\ast(\cdot|\theta_R^\ast),\theta_R^\ast,\pi_I^\ast(\cdot|\theta_R^\ast),\xi_I^\ast(\theta_R^\ast))$, where 
    	\begin{itemize}
    		\item $\pi_R^\ast(\cdot|\theta_R^\ast)$ is given by
    		\begin{align*}
    			\pi_{R,i}^\ast(t|\theta_R^\ast)=\frac{\varphi_{R,i}^\ast(t|\theta_R^\ast)\cdot S_i(t)}{V_R^{v_{R,0}(\xi_I^\ast(\theta_R^\ast),\theta_R^\ast),\varphi_R^\ast}(t)}
    		\end{align*}
    		where $\varphi_R^\ast$ is given by Proposition \ref{ThmOpR1},
    		\item $\theta_R^\ast$ is given by Proposition \ref{OpSt}, and
    		\item $(\pi_I^\ast(\cdot|\theta_R^\ast),\xi_I^\ast(\theta_R^\ast))$ are given by Proposition \ref{SolIOP2} and \ref{SolOPI}, such that
    		\begin{align}
    		    \xi_I^\ast(\theta_R^\ast)=\max\{\xi_I^\ast|\;&\mathbb{E}[U_I(V_I^{v_{I,0}(\xi_I^\ast,\theta_R^\ast),\pi_I^\ast}(T)+\xi_I^\ast P(T))]\nonumber\\
    		    &=\sup_{\xi_I:(\pi_I,\xi_I)\in\Lambda_I}\mathbb{E}[U_I(V_I^{v_{I,0}(\xi_I,\theta_R^\ast),\pi_I^\ast}(T)+\xi_IP(T))]\}. \label{ConditionSE}
    		\end{align}
    	\end{itemize}
    \end{proposition}
    
    \begin{proof}
        See Appendix \ref{app:proofs_main}.
    \end{proof}

    \textbf{Remarks to Proposition \ref{SE}.} 
    \begin{itemize}
        \item[1.] \eqref{ConditionSE} ensures that \eqref{eq:SE_condition_2} is fulfilled, i.e., if there exist several best responses of the insurer to the reinsurer's optimal action, then the insurer chooses among them the best response that is also best from the reinsurer's perspective.
        \item[2.] We get analytical representations of the optimal relative portfolio processes $\pi_R^\ast$ and $\pi_I^\ast$, which depend on $\theta_R^\ast$ and $\xi_I^\ast(\theta_R^\ast)$. For common utility functions, such as power, HARA or logarithmic, the optimal reinsurance amount $\xi_I^\ast(\theta_R^\ast)$, $\xi_I^\ast(\theta_R^\ast) = \bar{\xi}$, but the optimal safety loading $\theta_R^\ast$ has to be numerically computed by solving a non-linear equation. 
    \end{itemize}

\section{Explicit solutions for power utility functions} \label{sec_example}

\hspace{0.5cm} In this section, we give explicit solutions for the special case when the reinsurer and the insurer have power utility functions, i.e., for $x\in(0,\infty)$
    \begin{align*}
        U_R(x):=\frac{1}{b_R}x^{b_R}\text{ and }U_I(x):=\frac{1}{b_I}x^{b_I}
    \end{align*}
    with $b_R,b_I\in(-\infty,1)\backslash\{0\}$. In addition, we assume that the upper limit $\xi^{\max}$ is fixed and equals $\bar{\xi}>0$ with $\bar{\xi}<\frac{v_I}{(1+\theta_R^{\max})P(0)}$.
    
    \begin{corollary}[Response of the insurer] \label{CorPowerI} 
        Assume that the insurer has a power utility function with $b_I\in(-\infty,1)\backslash\{0\}$. Then the optimal $\lambda^\ast\in\mathcal{D}$ is given by 
        \begin{align*}
    		\lambda^\ast=\bigg(0,\frac{\sigma_2\rho}{\sigma_1}(\mu_1-r)-\mu_2+r\bigg)^\top
    	\end{align*}
    	and the Merton strategy $\pi^M_{\lambda^\ast}$ in the auxiliary market $\mathcal{M}_{\lambda^\ast}$ by
    	\begin{align} \label{eq_AuxiliaryMerton}
    		\pi^M_{\lambda^\ast}(b_I)=\frac{1}{1-b_I}(\sigma\sigma^\top)^{-1}(\mu+\lambda^\ast-r\mathbbm{1}).
    	\end{align}
    	Therefore, the optimal response of the insurer to any admissible action of the reinsurer is
        \begin{align}
            \xi_I^\ast(\theta_R)&=\begin{cases}
        		\bar{\xi},&\text{if } \theta_R < \frac{\mathbb{E}[\widetilde{Z}_{\lambda^\ast}(T)P(T)]-P(0)}{P(0)},\\
        		\text{any }\tilde{\xi}\in[0,\bar{\xi}],&\text{if } \theta_R = \frac{\mathbb{E}[\widetilde{Z}_{\lambda^\ast}(T)P(T)]-P(0)}{P(0)},\\
    			0,&\text{if } \theta_R > \frac{\mathbb{E}[\widetilde{Z}_{\lambda^\ast}(T)P(T)]-P(0)}{P(0)},
        		\end{cases} \label{OptimalXi}
        \end{align}
        \begin{align}
            \pi^\ast_I(t|\theta_R)=&\pi_{\lambda^\ast}^M(b_I)\frac{V_I^{v_{I,0}(\xi_I^\ast(\theta_R),\theta_R),\pi_I^\ast}(t)+\xi_I^\ast(\theta_R)\widetilde{Z}_{\lambda^\ast}(t)^{-1}\mathbb{E}[\widetilde{Z}_{\lambda^\ast}(T)P(T)|\mathcal{F}_t]}{V_I^{v_{I,0}(\xi_I^\ast(\theta_R),\theta_R),\pi_I^\ast}(t)},\quad \tin. \label{StrategieInsurer}
        \end{align}
        \end{corollary}
        
        \begin{proof}
            See Appendix \ref{app:proofs_example}. 
        \end{proof}
        
        \textbf{Remarks to Corollary \ref{CorPowerI}.} \eqref{StrategieInsurer} can be rewritten in the following way:
        \begin{align}
            \pi^\ast_I(t|\theta_R)&=\pi^M(b_I)+\underbrace{\frac{1}{1-b_I}(\sigma\sigma^\top)^{-1}\lambda^\ast}_{\text{constraint correction term}}+\underbrace{\pi_{\lambda^\ast}^M(b_I)\frac{\xi_I^\ast(\theta_R)\widetilde{Z}_{\lambda^\ast}(t)^{-1}\mathbb{E}[\widetilde{Z}_{\lambda^\ast}(T)P(T)|\mathcal{F}_t]}{V_I^{v_{I,0}(\xi_I^\ast(\theta_R),\theta_R),\pi_I^\ast}(t)}}_{\text{reinsurance contract correction term}}, \label{StrategieInsurer2}
        \end{align}
        where $\pi^M$ is the Merton portfolio process in the basic market $\mathcal{M}$. If the insurer has no investment strategy constraint, then $\lambda^\ast=(0,0)^\top$ and $\pi^M_{\lambda^\ast}(b_I)=\pi^M(b_I)$. Moreover, if $\lambda^\ast=(0,0)^\top$, then for any $\theta_R>0$ we get $\xi_I^\ast(\theta_R)=0$ due to \eqref{OptimalXi} and, thus, $\pi_I^\ast(t|\theta_R)=\pi_{\lambda^\ast}^M(b_I)=\pi^M(b_I)$. 
        
        Otherwise, the insurer's optimal portfolio process $\pi_I^\ast$ is a Merton strategy with two correction terms that account for the availability of the reinsurance contract and the difference between the insurer's individual portfolio and the reinsured portfolio.
        
        \begin{corollary}[Stackelberg equilibrium] \label{CorPowerR}
            Assume that the insurer and the reinsurer have power utility functions with $b_I,b_R\in(-\infty,1)\backslash\{0\}$. Then the Merton portfolio process in the basic market is given by
            \begin{align*}
                \pi^M(b_R)=\frac{1}{1-b_R}(\sigma\sigma^\top)^{-1}(\mu-r\mathbbm{1}).
            \end{align*}
            Therefore, the Stackelberg equilibrium is given by
            \begin{align}
                \theta_R^\ast&=\min\bigg\{\frac{\mathbb{E}[\widetilde{Z}_{\lambda^\ast}(T)P(T)]-P(0)}{P(0)},\theta^{\max}\bigg\}, \label{OptimalTheta}\\
                \pi_R^\ast(t|\theta_R^\ast)&=\pi^M(b_R)\frac{V_R^{v_{R,0}(\xi_I^\ast(\theta_R^\ast),\theta_R^\ast),\varphi_R^\ast}(t)-\xi_I^\ast(\theta_R^\ast)P(t)}{V_R^{v_{R,0}(\xi_I^\ast(\theta_R^\ast),\theta_R^\ast),\varphi_R^\ast}(t)}+\underbrace{\begin{pmatrix}
            		0\\
            		\frac{\pi^{CM}V^{v_I,\pi_B}(t)(\Phi(d_+)-1)}{V_R^{v_{R,0}(\xi_I^\ast(\theta_R^\ast),\theta_R^\ast),\varphi_R^\ast}(t)}\xi_I^\ast(\theta_R^\ast)
            	\end{pmatrix}}_{\text{put hedging term}}, \label{StrategieReinsurer}
            \end{align}
            $\xi_I^\ast(\theta_R^\ast)=\bar{\xi}$ by Equation \eqref{OptimalXi} and $\pi_I^\ast(\cdot|\theta_R^\ast)$ is given by Equation \eqref{StrategieInsurer}. 
        \end{corollary}
        \begin{proof}
            See Appendix \ref{app:proofs_example}. 
        \end{proof}
        
        \textbf{Remarks to Corollary \ref{CorPowerR}.} The strategy \eqref{StrategieReinsurer} can be seen as a generalized CPPI with the floor being the value of the put-option position.
        
        If the insurer has no portfolio constraint (i.e., $\lambda^\ast=(0,0)^\top$), then $\xi_I^\ast(\theta_R^\ast)=0$ for any $\theta_R>0$ (see Remark to Corollary \ref{CorPowerI}) and $\pi_R^\ast(t|\theta_R^\ast)=\pi^M(b_R)$ for all $t\in[0,T]$. The same equality holds if $G_T=0$, since then $d_+=+\infty$, $P(t) = 0$ for all $\tin$ and put-hedging term disappears.

\section{Numerical studies}\label{sec_numerical_studies}

    \hspace{0.5cm} In this section, we conduct numerical studies calibrated to the German market. The selection of the parameters is stated and explained in Subsection \ref{subsec_parameter}. In Subsection \ref{subsec_SGNS}, the Stackelberg equilibrium for the base parameters from Subsection \ref{subsec_parameter} is calculated and we conduct a sensitivity analysis of the Stackelberg equilibrium w.r.t. the behavior of the parties and w.r.t. a change in the put option price. In Subsection \ref{subsec_deviationSE}, we show the impact of a deviation from the Stackelberg equilibrium on the expected utilities of the parties.
    
\subsection{Parameter selection}\label{subsec_parameter}
    \begin{table}[h]
    	\centering
    	\begin{tabular}[h]{lll}
    		\hline
    		\textbf{Parameter} & \textbf{Symbol} & \textbf{Values}\\
    		\hline
    		Interest rate & $r$ & 1.02\%\\
    		Drift coefficient for $S_1$& $\mu_1$ &  17.52\%\\
    		Drift coefficient for $S_2$ & $\mu_2$ & 12.37\%\\
    		Diffusion coefficient for $S_1$ & $\sigma_1$ & 23.66\%\\
    		Diffusion coefficient for $S_2$ & $\sigma_2$ & 21.98\%\\
    		Correlation coefficient & $\rho$ & 80.12\%\\
    		Benchmark CM strategy & $\pi_B$ & $(0\%,29.48\%)^\top$\\
    		Initial value of $S_1$ & $s_1$ & 1\\
    		Initial value of $S_2$ & $s_2$ & 1\\
    		Guarantee & $G_T$ & 100\\
    		Initial wealth of insure & $v_I$ & 100\\
    		Initial wealth of reinsurer & $v_R$ & 100\\
    		Relative risk aversion of insurer $(RRA_I)$ & $1-b_I$ & 10\\
    		Relative risk aversion of reinsurer $(RRA_R)$ & $1-b_R$ & 10\\
    		Time horizon & $T$ & 10\\
    		Maximal safety loading of reinsurer & $\theta^{\max}$ & 50\%\\
    		Maximal amount of reinsurance & $\xi^{\max}=\bar{\xi}$ & 1.5\\
    		\hline
    	\end{tabular}
    	\caption{Parameters for the numerical analysis}
    	\label{tab:Parameter}
    \end{table}
    
    In this part, we discuss the selection of the parameters summarized in Table \ref{tab:Parameter}. For the majority of parameters, we choose the same values as in \citeA{Escobar}, where they are calibrated to the German market in the period from January 1, 2003, till June 8, 2020.
    
    The risk-free rate is modeled by the Euro OverNight Index Average (EONIA) daily data. For calibrating the parameters of $S_1$, we use the TecDAX daily data and for calibrating the parameters of $S_2$, we use the DAX daily data. 
    
    In this way, we model the following situation: the insurer invests in bonds and prefers the technological sector to a broad stock portfolio. One reason could be that the insurer's asset manager has special knowledge in the technology sector and/or believes that the TecDAX has a better performance than the DAX. In contrast, the reinsurer considers the insurer's technology-focused portfolio too risky to be reinsured. Therefore, the reinsurer offers reinsurance only on a mixed portfolio consisting of bonds and the broad market index DAX.
    
    The relative proportion $\pi^{CM}$ of the constant mix portfolio is selected such that it equals the optimal initial proportion of the insurer in the risky asset without reinsurance.
    
    In the German Life Insurance Market, the capital guarantee for the representative client is usually less than or equal to $100\%$ of the representative client's initial endowment. For example, ERGO offers the equity-linked insurance product \enquote{ERGO Rente Garantie} where the guarantee lies between $80\%$ and $100\%$\footnote{See, e.g.,   \url{https://www.focus.de/finanzen/steuern/ergo-ergo-rente-garantie_id_3550999.html}}. In contrast, Allianz offers a guarantee between $60\%$ and $90\%$\footnote{See, e.g.,  \url{https://www.sueddeutsche.de/wirtschaft/lebensversicherung-allianz-kuenftig-ohne-beitragsgarantie-1.5056917}}. Hence, we assume that the representative client has a $100\%$ capital guarantee of the initial capital and investigate in the sensitivity analysis a $G_T$ varying from $60\%$ to $110\%$ of the initial contribution of the client.
    
    For convenience, we set the insurer's initial wealth $v_I$ to 100. It is natural to assume that the reinsurer is a larger company with more initial capital. If reinsurance is offered on the whole company level, then the initial wealth of the reinsurer should be higher than the initial wealth of the insurer, as stated in \citeA{Chen2018}. Since we consider reinsurance within a single insurance product, we assume that the initial product-related capital of the reinsurer coincides with the initial wealth of the insurer, i.e., $v_R=v_I=100$.
    
    In \citeA{Chen2018}, the authors assume that the reinsurer and insurer have the same risk aversion. In contrast, \citeA{bai2019hybrid} assume that the insurer is more risk averse than the reinsurer. Thus, we first consider the situation where both parties have the same risk aversion. Afterwards, we explore the situations where the parties have different risk aversion. In the base case, we choose $b_I=b_R=-9$, which is consistent with \citeA{Escobar}. 
    
    For the maximal level of safety loading we choose $50\%$. It was chosen in line with \citeA{Chen2018} and \citeA{Chen2019} who chose an upper bound on the safety loading of the reinsurer of $45\%$.
    
    We do not allow that the insurer can speculate with the reinsurance by going short or buying \textcolor{black}{too much} of it. Since the underlying of the put option is not the portfolio of the insurer but a correlated portfolio, we allow that $\xi^{\max}=\bar{\xi}=1.5$.
    
\subsection{Stackelberg equilibrium} \label{subsec_SGNS}

\hspace{0.5cm} First, we calculate the Stackelberg equilibrium in the base case of parametrization and its sensitivity. It is given by
    \begin{align*}
    	\pi_R^\ast(0)=&\textcolor{black}{(31.67\%,-16.42\%)^\top},\quad \theta_R^\ast = 20.86\%,\\
    	\pi_I^\ast(0)=&(31.69\%,0\%)^\top,\quad \quad \quad \xi_I^\ast(\theta_R^\ast)=1.5.
    \end{align*} 
    The optimal reinsurance strategy of the insurer is to buy the maximal amount of put options (i.e., $\xi_I^\ast(\theta_R^\ast)=\xi^{\max}$) and the reinsurer sets the maximal price for the reinsurance such that the insurer still buys reinsurance, which is about 20\% higher than the fair price of the put option.
    
    Next, we describe the results of our sensitivity of the Stackelberg equilibrium w.r.t. the change of risk aversion of each party, the interest rate, the time horizon, and the capital guarantee. The corresponding plots can be found online in {\color{black}Section 2 of supplementary materials}.

    Varying RRA coefficients, we find that the higher the RRA of a party, the less the party invests in or speculates with the risky assets. \textcolor{black}{Interestingly, the values of the reinsurance contract parameters in the equilibrium do not depend on the individual RRA coefficient of each party. This is consistent with Corollary \ref{CorPowerR}, from which we know that the equilibrium safety loading $\theta_R^\ast$ depends on the pricing kernel in the original as well as the auxiliary markets and the put-option parameters. Hence, RRA coefficients do not influence $\theta_R^\ast$, if the benchmark strategy $\pi_B$ is independent of them. Figure \textcolor{black}{2.1 in supplementary materials} corresponds to such a choice of $\pi_B$. However, the insurer's RRA coefficient can have an indirect impact on the reinsurer's choice of $\theta_R^\ast$ through $\pi_B$, because the benchmark portfolio influences the fair price of the put option. If the risky-asset investment share $\pi^{CM}$ in $\pi_B$ decreases with increasing $b_I$, then the equilibrium safety loading increases. The remaining components of the equilibrium have the same trends as indicated in Figure \textcolor{black}{2.1}. The amount of reinsurance $\xi_I^\ast(\theta_R^\ast)$ bought by the insurer in equilibrium equals $\bar{\xi}$ for any considered RRA coefficient, as expected from \eqref{OptimalXi} and \eqref{eq:SE_condition_2}.
    }

    Varying the interest rate $r\in\{-2\%,-1\%,0\%,1\%,2\%\}$, we find that the higher $r$, the higher the optimal safety loading $\theta_R^\ast$, since the fair price of the put option $P(0)$ in the basic market $\mathcal{M}$ decreases faster than the fair price of the put option $\mathbb{E}[\widetilde{Z}_{\lambda^\ast}(T)P(T)]$ in the auxiliary market $\mathcal{M}_{\lambda^\ast}$ (see \eqref{OptimalTheta}). Consistent with Corollary \ref{CorPowerR}, the interest rate has no influence on the reinsurance strategy of the insurer. The higher the interest rate, the less the insurer invests in the risky asset. The interest rate has no notable influence on the reinsurer's optimal investment portfolio. The corresponding plots can be found in Figure \textcolor{black}{2.2 in supplementary materials.}
    

    For increasing time horizon $T\in\{1,5,10,15,20\}$, the equilibrium safety loading increases. In contrast, the time horizon has no influence on the optimal reinsurance strategy of the insurer and no notable influence on the optimal investment strategy of the reinsurer and the insurer. The corresponding plots can be found in in Figure \textcolor{black}{2.3 in supplementary materials}.
    
    
    Finally, we explore the sensitivity of the Stackelberg equilibrium w.r.t. the client guarantee $G_T\in\{0.6\cdot v_I,0.7\cdot v_I,0.8\cdot v_I,\dots, 1.1\cdot v_I\}$, which is also the strike of the put option that models the reinsurance contract. We observe that the optimal safety loading of the reinsurer is decreasing in $G_T$. This follows from the fact that the higher the guarantee, the more expensive the put option, which is why the maximal additional safety loading at which the insurer may still be willing to buy reinsurance decreases. However, $G_T$ has no influence on the reinsurance amount that is purchased by the insurer in the equilibrium, i.e., $\xi_I = \xi^{\max} = 1.5$. This stems from the definition of the Stackelberg equilibrium, where, in the case of several equivalent optimal responses to the leader's action, the follower chooses the response that is also best for the leader's perspective. When $G_T$ increases, the reinsurer's investment in $S_2$ decreases due to the hedge of its put position. There is no notable influence of $G_T$ on the optimal investment strategy of the insurer. Figure \textcolor{black}{2.4 in supplementary materials} contains the corresponding plots.
    
	
\subsection{Impact of deviating from Stackelberg equilibrium} \label{subsec_deviationSE}
	\hspace{0.5cm} Recall from \eqref{OptimalXi} that if the reinsurer chooses the safety loading below the critical value $({\mathbb{E}[\widetilde{Z}_{\lambda^\ast}(T)P(T)]-P(0)})/{P(0)}$, then the best action of the insurer is to buy the maximal amount of reinsurance. However, if the reinsurer chooses $\theta_R$ equal to that critical value, then the insurer's best response is to buy any reinsurance amount from the interval $[0,\bar{\xi}]$. Condition \eqref{eq:SE_condition_2} in the definition of the Stackelberg equilibrium implies that the reinsurer chooses the critical value for safety loading and has an optimistic view on the insurer's response.  In other words, the reinsurer \enquote{hopes} that the insurer selects $\xi_I^\ast(\theta_R^\ast)=\bar{\xi}$. However, in practice, the insurer could also buy no reinsurance or any amount in between, because its set of best responses is $[0,\bar{\xi}]$. Therefore, we study next the effects of a deviation of the reinsurer from the optimal safety loading $\theta_R^\ast$ on the reinsurer and the insurer. In the following, we denote $\xi_I^\ast(\theta_R^\ast)$ only by $\xi_I^\ast$ to simplify notation.
	
	The reinsurer should offer safety loadings equal to or below $\theta_R^\ast$, as otherwise the reinsurance price is too high for the insurer to buy it. We denote a discounted safety loading by $\theta_R(\alpha)=\alpha\cdot\theta_R^\ast$ with $\alpha\in[0,1]$. If the reinsurance premium is higher than the fair price of the put option (i.e., $\alpha>0$), then the reinsurer always gains from offering reinsurance with a discounted safety loading in comparison to not offering reinsurance. In the Stackelberg equilibrium, the reinsurer has the highest gain in its utility from offering reinsurance but is uncertain about the amount of reinsurance the insurer will buy. To get rid of this uncertainty, the reinsurer can offer reinsurance with a safety loading $\theta_R(\alpha)$ with $\alpha\in(0,1)$. Note that a deviation from the optimal safety loading means that the parties are not in the theoretical Stackelberg equilibrium anymore. However, both of them are incentivized to do so. The insurer's expected utility is strictly larger than it is in the Stackelberg equilibrium as well as in the case of not buying reinsurance. The reinsurer's expected utility for a discounted safety loading is strictly larger than it is in the case of not selling reinsurance. See Section 3 of supplementary materials for details.
	
	After the reinsurer has decided to grant a discount on the safety loading to the insurer, we consider the disadvantage/advantage to the reinsurer and the insurer in comparison to the Stackelberg equilibrium. We study this using the concept of the wealth-equivalent utility change (WEUC) of a party, i.e., insurer and reinsurer. In the calculation of WEUC, the expected utility of the party in the Stackelberg equilibrium (or the Stackelberg equilibrium with a discount) is compared as a reference action combination to an alternative action combination. The alternative action combination refers to a specific choice of the safety loading as well as the investment strategy by the reinsurer and the amount of reinsurance as well as the investment strategy by the insurer. WEUC equals the relative change of the party's initial wealth that is required to bring the party in the alternative action combination to the same expected utility as in case of the reference combination of actions. We define it in a way such that our definition is consistent with the definition of the wealth-equivalent utility loss in \citeA{larsen2012costs}.
	
 	To be specific, we denote a reference combination of actions by  $(\tilde{\pi}_R(\cdot|\tilde{\theta}_R,\tilde{\xi}_I),\tilde{\theta}_R,\tilde{\pi}_I(\cdot|\tilde{\theta}_R,\tilde{\xi}_I),\tilde{\xi}_I)$ and an alternative combination of actions by $(\hat{\pi}_R(\cdot|\hat{\theta}_R,\hat{\xi}_I),\hat{\theta}_R,\hat{\pi}_I(\cdot|\hat{\theta}_R,\hat{\xi}_I),\hat{\xi}_I)$.
	
	Then the WEUC of the reinsurer is denoted by 
	\begin{equation*}
	    WEUC_R(\tilde{\theta}_R,\tilde{\xi}_I,\hat{\theta}_R,\hat{\xi}_I):=WEUC_R((\tilde{\pi}_R(\cdot|\tilde{\theta}_R,\tilde{\xi}_I),\tilde{\theta}_R), (\hat{\pi}_R(\cdot|\hat{\theta}_R,\hat{\xi}_I),\hat{\theta}_R))
	\end{equation*}
	and satisfies the relation
	\begin{equation*}
	    \mathbb{E}\left[U_R\left(V_R^{v^{WEUC}_{R,0}(\hat{\xi}_I,\hat{\theta}_R),\hat{\pi}_R}(T)-\hat{\xi}_I P(T)\right)\right]=	\mathbb{E}\left[U_R\left(V_R^{v_{R,0}(\tilde{\xi}_I,\tilde{\theta}_R),\tilde{\pi}_R}(T)-\tilde{\xi}_I P(T)\right)\right],
	\end{equation*}
	where
	\begin{equation*}
	    v^{WEUC}_{R,0}(\hat{\xi}_I,\hat{\theta}_R)=v_R\cdot(1+WEUC_R(\tilde{\theta}_R,\tilde{\xi}_I,\hat{\theta}_R,\hat{\xi}_I))+\hat{\xi}_I(1+\hat{\theta}_R)P(0).
	\end{equation*}
	Analogously, the WEUC of the insurer is denoted by
	\begin{equation*}
	   WEUC_I(\tilde{\theta}_R,\tilde{\xi}_I,\hat{\theta}_R,\hat{\xi}_I):=WEUC_I((\tilde{\pi}_I(\cdot|\tilde{\theta}_R,\tilde{\xi}_I),\tilde{\xi}_I), (\hat{\pi}_I(\cdot|\hat{\theta}_R,\hat{\xi}_I),\hat{\xi}_I)) 
	\end{equation*}
    and satisfies the relation
	\begin{align*}
		\mathbb{E}\left[U_I\left(V_I^{v^{WEUC}_{I,0}(\hat{\xi}_I,\hat{\theta}_R),\hat{\pi}_I}(T)+\hat{\xi}_I P(T)\right)\right]&=	\mathbb{E}\left[U_I\left(V_I^{v_{I,0}(\tilde{\xi}_I,\tilde{\theta}_R),\tilde{\pi}_I}(T)+\tilde{\xi}_I P(T)\right)\right],
	\end{align*}
	where
	\begin{align*}
		v^{WEUC}_{I,0}(\hat{\xi}_I,\hat{\theta}_R)&=v_I\cdot(1+WEUC_I(\tilde{\theta}_R,\tilde{\xi}_I,\hat{\theta}_R,\hat{\xi}_I))-\hat{\xi}_I(1+\hat{\theta}_R)P(0).
	\end{align*}

	The WEUC has an intuitive interpretation. If this quantity is positive, then the reference combination of actions is better for the party of interest than the alternative combination of actions. In this case, the WEUC indicates by which proportion the party has to increase its initial capital in case of the alternative combination of actions so that it has the same expected utility as in the case of the reference combination of actions. If the WEUC is negative, then the reference combination of actions is worse for the considered party than the alternative combination of actions. In this case, the WEUC indicates by which proportion the party can decrease its initial capital in case of the alternative combination of actions so that it has the same expected utility as in the case of the reference combination of actions.

	Figure \ref{subfig:WEUC} shows the loss of the reinsurer and the gain of the insurer in their utilities, when the Stackelberg equilibrium is compared to an alternative situation with all actions being the same as in the Stackelberg equilibrium but with a discounted safety loading. The gain of the insurer is of the same amount as the loss of the reinsurer. Moreover, further investigations reveal that it does not depend on the risk aversions of the parties. Hence, the insurer does not benefit more from deviating from the Stackelberg equilibrium as the reinsurer loses. For example, a discount of 5\% from $\theta_R^\ast$ would increase the WEUC of the reinsurer by $6bp$ of the initial capital.

    \begin{figure}[h]
	    \centering
	    \includegraphics[width=0.5\linewidth]{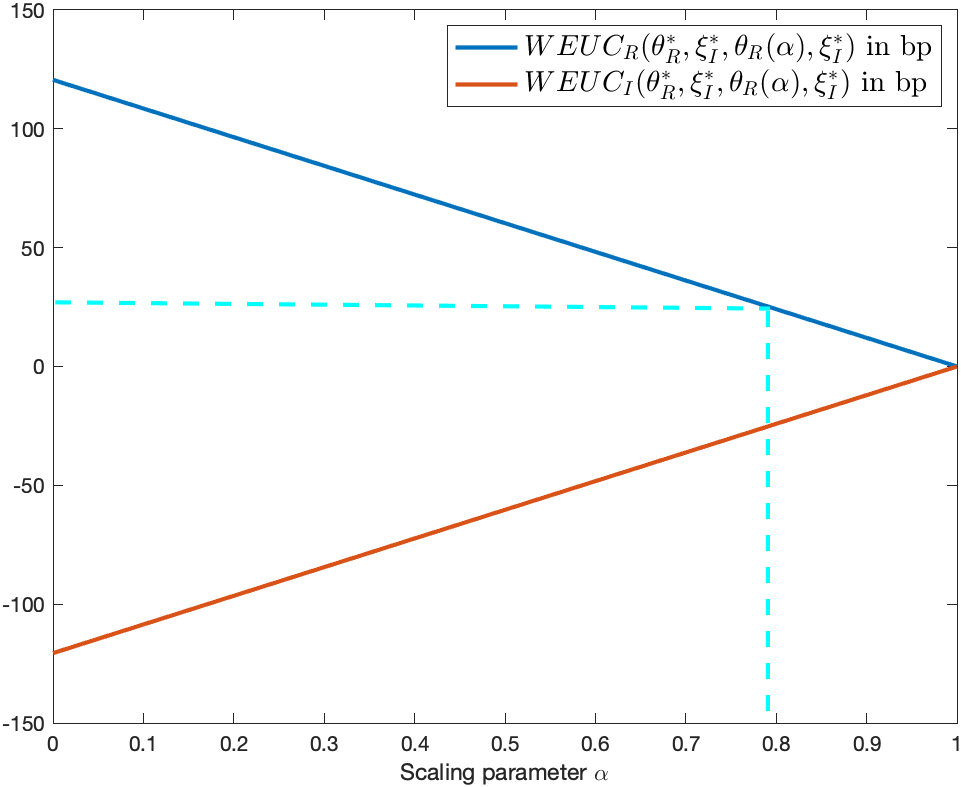}
	    \caption{Comparison of WEUC: Stackelberg equilibrium (reference) and the same combination of actions but with a discounted safety loading (alternative)}
	    \label{subfig:WEUC}
  \end{figure}

    In the next step, the reinsurer needs to choose a discount level $\alpha$ such that its WEUC is not too \enquote{painful}. In the following, we present some ways to choose $\alpha$. One way would be for the reinsurer to choose an acceptable upper bound on its WEUC (see the blue line in Figure \ref{subfig:WEUC}) and to exchange uncertainty for a fixed defined WEUC. For example, if the reinsurer wants its WEUC to be at most $25bp$, then it can choose $\alpha= 79.27\%$ (see turquoise line in Figure \ref{subfig:WEUC}).
    
    \begin{figure}[h]
        \centering
        \includegraphics[width=0.5\linewidth]{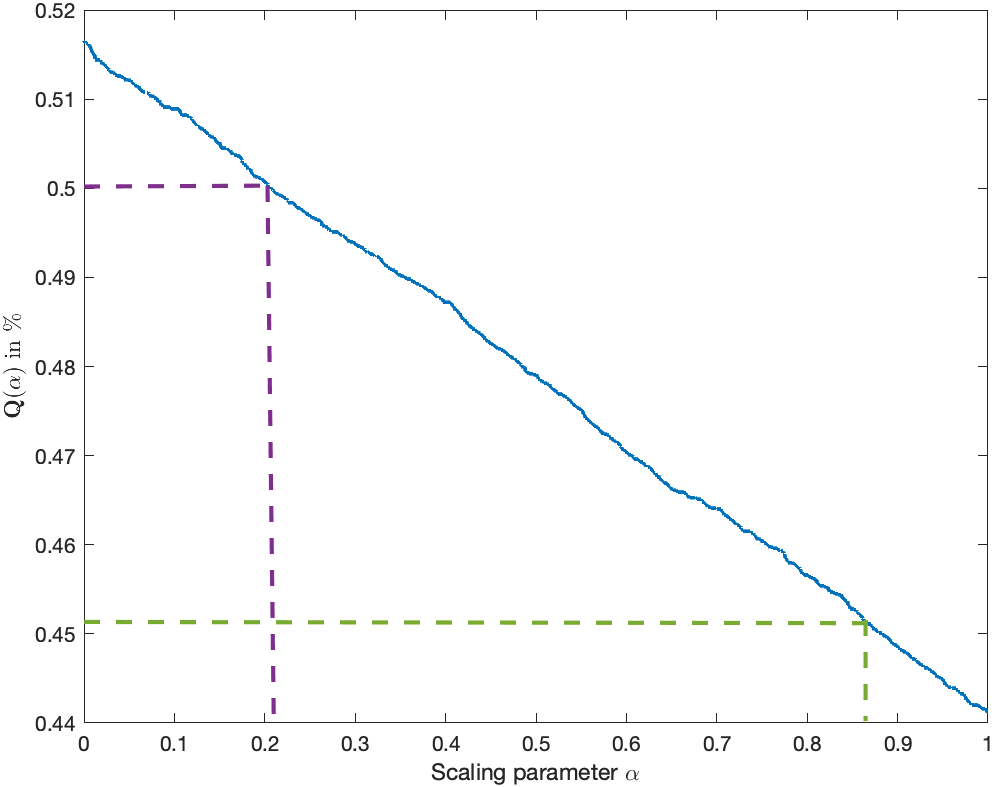}
        \caption{Probability of loss at time $T$ for reinsurer}
        \label{subfig:VaRR}
    \end{figure}
    
    The other ways of choosing $\alpha$ can be related to product profitability. It can be measured by calculating the probability that the reinsurer's total terminal wealth is below its initial wealth:
    \begin{align*}
        \mathbb{Q}(\alpha):=\mathbb{Q}(V_R^{v_{R,0}(\theta_R(\alpha),\xi_I^\ast),\pi_R^\ast(\cdot|\theta_R(\alpha),\xi_I^\ast)}(T)-\xi_I^\ast P(T)< v_R).
    \end{align*}
    This measure shows the product profitability of the reinsurance for the reinsurer and is shown in Figure \ref{subfig:VaRR}. The reinsurer can choose $\alpha$ in several different ways, which are based on:
    \begin{itemize}
        \item a tolerance level for the increase of the loss probability. For example, if the reinsurer is not willing to increase its loss probability by more than 0.01\% (compared to the loss probability of $\mathbb{Q}(1)=0.4413\%$ in the case of the Stackelberg equilibrium), then the reinsurer can choose $\alpha=86.73\%$ (see green line in Figure \ref{subfig:VaRR});
        \item the maximal acceptable probability of loss. For example, if the reinsurer can tolerate at most $0.5\%$ probability of loss, it can choose $\alpha=20.74\%$ (see purple line in Figure \ref{subfig:VaRR}).
    \end{itemize}
    The reinsurance company can also use other criteria for choosing $\alpha$, e.g., other risk measures, as the standard deviation, or different performance measures, as the adjusted Sharpe ratio.
    
    In the insurance industry, the analysis of the product profitability is a common procedure, which is also easier to communicate. Therefore, we choose as an example $\alpha=86.73\%$. In this case, the reinsurer limits the increase of the probability of loss to $0.0001$.
    
    \begin{figure}[h]
        \begin{subfigure}[t]{0.44\textwidth}
            \centering
            \includegraphics[width=\linewidth]{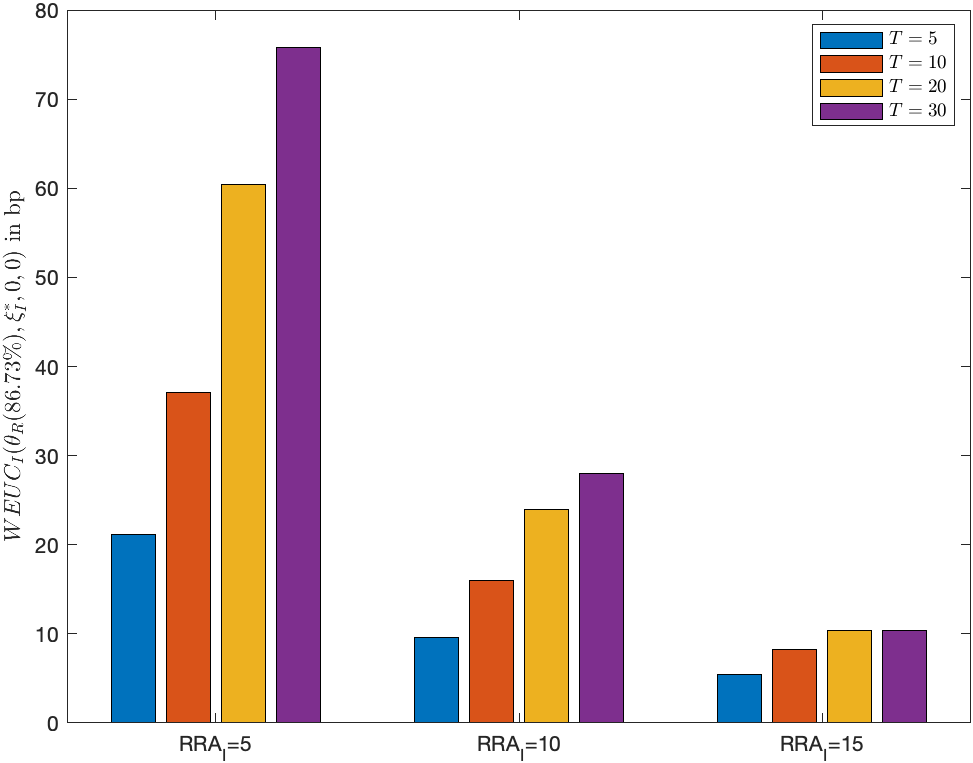}
            \caption{Reference $(\pi_I^\ast(\cdot|\theta_R(86.73\%),\xi_I^\ast),\xi_I^\ast)$, Alternative $(\pi_I^\ast(\cdot|0,0),0)$ }
        \label{subfig:WEUCIMerton}
        \end{subfigure}\hfill
    	\begin{subfigure}[t]{0.45\textwidth}
    	    \centering
    	    \includegraphics[width=\linewidth]{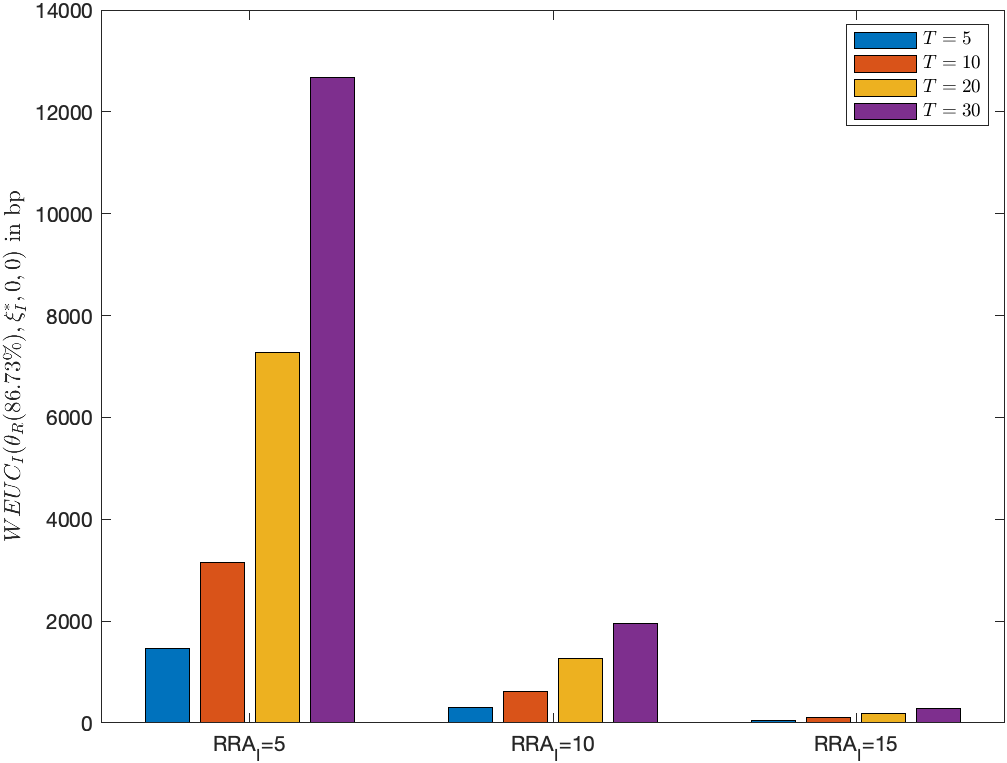}
    	    \caption{Reference $(\pi_I^\ast(\cdot|\theta_R(86.73\%),\xi_I^\ast),\xi_I^\ast)$, Alternative $((15\%,0),0)$}
    	    \label{subfig:WEUCICN}
    	\end{subfigure}
        \caption{Impact of relative risk aversion of insurer and investment horizon on $WEUC_I$}
        \label{fig:WEUCI}
    \end{figure}
    
    In the remaining analysis, we study how the insurer benefits from buying reinsurance with a safety loading of $\theta_R(86.73\%)$ (reference) instead of buying no reinsurance (alternative).
    
    First, we determine the benefit of the insurer if it follows a dynamic portfolio strategy with reinsurance $(\pi_I^\ast(\cdot|\theta_R(86.73\%),\xi_I^\ast),\xi_I^\ast)$ (reference) instead of a dynamic portfolio strategy without reinsurance $(\pi_I^\ast(\cdot|0,0),0)$ (alternative), the latter being its Merton strategy. Therefore, we calculate $WEUC_I(\theta_R(86.73\%),\xi_I^\ast,0,0)$ and consider how the relative risk aversion and the time horizon influence this value. The result is shown in Figure \ref{subfig:WEUCIMerton}. Since $WEUC_I(\theta_R(86.73\%),\xi_I^\ast,0,0)>0$, the insurer needs less initial capital if it follows the dynamic portfolio strategy with reinsurance instead of the dynamic portfolio strategy without reinsurance. This means that the insurer can decrease its product costs by buying reinsurance. The more risk averse the insurer, the lower $WEUC_I(\theta_R(86.73\%),\xi_I^\ast,0,0)$, and the longer the investment period, the higher $WEUC_I(\theta_R(86.73\%),\xi_I^\ast,0,0)$. If the insurer’s relative risk aversion is given by $RRA_I= 5$ and the insurer offers a long-term equity-linked insurance product with $T= 20$, \linebreak then $WEUC_I(\theta_R(86.73\%),\xi_I^\ast,0,0)=60bp$. If the insurer gets more risk averse, then \linebreak $WEUC_I(\theta_R(86.73\%),\xi_I^\ast,0,0)$ decreases strongly. For $RRA_I= 15$ and $T= 20$ we have \linebreak $WEUC_I(\theta_R(86.73\%),\xi_I^\ast,0,0)= 10bp$.
    
    Lastly, we consider the benefit of the insurer if it follows the dynamic investment strategy with reinsurance $(\pi_I^\ast(\cdot|\theta_R(86.73\%),\xi_I^\ast),\xi_I^\ast)$ (reference strategy) instead of the constant-mix strategy without reinsurance $((15\%,0),0)$ (alternative strategy). We consider this constant-mix strategy without reinsurance, since it approximates the long-term investment strategy of an average life insurer due to \citeA{grundl2017evolution}. Therefore, we calculate \linebreak
    $WEUC_I(\theta_R(86.73\%),\xi_I^\ast,0,0)$ and analyze the impact of the relative risk aversion and investment horizon, see Figure \ref{subfig:WEUCICN}. Analogous to the analysis before, the insurer needs less initial capital if it follows the dynamic strategy with reinsurance instead of the constant-mix strategy without reinsurance. The more risk averse the insurer, the lower
    $WEUC_I(\theta_R(86.73\%),\xi_I^\ast,0,0)$, and the longer the investment period, the higher $WEUC_I(\theta_R(86.73\%),\xi_I^\ast,0,0)$. If the insurer’s relative risk aversion is given by $RRA_I= 5$ and the insurer offers a long-term equity-linked insurance product with $T= 20$, then $WEUC_I(\theta_R(86.73\%),\xi_I^\ast,0,0)=7275bp$. If the insurer gets more risk averse, then $WEUC_I(\theta_R(86.73\%),\xi_I^\ast,0,0)$ decreases strongly. For $RRA_I= 15$ and $T= 20$ we have $WEUC_I(\theta_R(86.73\%),\xi_I^\ast,0,0)= 287bp$.

\section{Conclusion and further research}\label{sec_conclusions}

\hspace{0.5cm} In this paper, we derive and analyze a Stackelberg equilibrium in a Stackelberg game between a reinsurer and an insurer, which appears in the context of an equity-linked insurance product with a capital guarantee. We assume that the reinsurance is only purchased at the beginning of the investment horizon and is not continuously adjusted.
	
	To solve the Stackelberg game analytically, we use backward induction and combine in a novel way the concept of auxiliary markets by \citeA{cvitanic1992convex}, the generalized martingale approach by \citeA{desmettre2016optimal} and the replicating strategies approach by \citeA{korn1999optimal}.
	
	We find that in the Stackelberg equilibrium the reinsurer selects the largest safety loading of the reinsurance contract such that the insurer may still be willing to sign it. However, for this equilibrium value of the safety loading, the insurer's expected utility function becomes independent of the amount of reinsurance due to its high price. Thus, in practice, the reinsurer should consider offering a reinsurance contract with a safety loading that is lower than the equilibrium one. It is still beneficial for the reinsurer as long as the safety loading is positive. We present some possible ways how the reinsurer can choose its final (discounted) safety loading. In addition, if the reinsurer offers a discount on the safety loading, the insurer can substantially decrease product costs for the client when switching from an old strategy, e.g., a constant-mix strategy without reinsurance, to the optimal dynamic strategy with reinsurance.
	
	In our numerical studies, the cost benefits for the insurer varied from a few basis points to 126.68\%, depending on the investment horizon, the insurer's risk aversion and the old strategy used in comparison.
	
	The Stackelberg equilibrium depends on the financial market model. It is an open interesting question for future research whether our results can be extended to more complex financial market models or if the addition of mortality or surrender risks prevents the analytical tractability of the Stackelberg game.


    \section*{Acknowledgement}
    
    \hspace{0.5cm} Yevhen Havrylenko and Rudi Zagst acknowledge the financial support of the ERGO Center of Excellence in Insurance at the Technical University of Munich promoted by ERGO Group.
    
    \section*{Competing interests}
    The authors declare none.
    
    \theendnotes

\bibliographystyle{apacite} 
\bibliography{Stackelberg_game_IR.bib}


\begin{appendix}
\section{Proofs of main results}\label{app:proofs_main}

\begin{proof}[Proof of Proposition \ref{SolIOP2}]
    The proof follows immediately from Corollary 3.4 and Theorem 4.3 in \citeA{desmettre2016optimal}.
\end{proof}

\begin{proof}[Proof of Proposition \ref{SolOPI}]
    The proof is based on the proof of Proposition 8.3 in \citeA{cvitanic1992convex} and consists of two parts:
	\begin{itemize}
		\item[1.] We fix $\xi_I\in[0,\xi^{\max}]$ and prove that $\pi_{\lambda^\ast}^\ast$ is optimal for \eqref{OPI} given the fixed $\xi_I$.
		\item[2.] We prove that $\xi_{\lambda^\ast}^\ast$ is optimal for \eqref{OPI} given the optimal portfolio process $\pi_I^\ast\equiv\pi_{\lambda^\ast}^\ast$.
	\end{itemize}
		\textit{Part 1.} Let $\xi_I\in[0,\xi^{\max}]$ be fixed. For the initial wealth it holds that
		\begin{align}
			V^{v_{I,0}(\xi_I,\theta_R),\pi_I}_\lambda(0)=V^{v_{I,0}(\xi_I,\theta_R),\pi_I}_I(0). \label{Initial}
		\end{align}
		Let $\pi_I$ such that $(\pi_I,\xi_I)\in\Lambda_I$, i.e., it holds that $\pi_I(t)\in K$ $\mathbb{Q}$-a.s. for all $t\in[0,T]$. Hence, for all $\lambda\in\mathcal{D}$ and $t\in[0,T]$ we have $\pi_I(t)^\top\lambda(t)=0$, whence
		\begin{align}
			V_\lambda^{v_{I,0}(\xi_I,\theta_R),\pi_I}(t)=V_I^{v_{I,0}(\xi_I,\theta_R),\pi_I}(t)\geq 0 \quad \Q\text{-a.s.}, \label{ProofI1}
		\end{align}
		where the equality follows from $\pi_I(t)^\top\lambda(t)=0$ and \eqref{Initial}. Using \eqref{ProofI1} for $t = T$, we obtain:
		\begin{align*}
				\mathbb{E}[U_I(V^{v_{I,0}(\xi_I,\theta_R),\pi_I}_\lambda(T)+\xi_I 	P(T))^-]=\mathbb{E}[U_I(V^{v_{I,0}(\xi_I,\theta_R),\pi_I}_I(T)+\xi_I P(T))^-]<\infty.
		\end{align*}
		Therefore, $(\pi_I,\xi_I)\in\Lambda_I^\lambda$. It follows that $\Lambda_I\subset\Lambda_I^\lambda$ and
		\begin{align}
			\sup_{\pi_I:(\pi_I,\xi_I)\in\Lambda_I}&\mathbb{E}[U_I(V^{v_{I,0}(\xi_I,\theta_R),\pi_I}_I(T)+\xi_I P(T))]\nonumber\\
			\stackrel{(a)}{=}&\sup_{\pi_I:(\pi_I,\xi_I)\in\Lambda_I}\mathbb{E}[U_I(V^{v_{I,0}(\xi_I,\theta_R),\pi_I}_\lambda(T)+\xi_I P(T))]\nonumber\\
			\stackrel{(b)}{\leq}& \sup_{\pi_I:(\pi_I,\xi_I)\in\Lambda_I^\lambda}\mathbb{E}[U_I(V^{v_{I,0}(\xi_I,\theta_R),\pi_I}_\lambda(T)+\xi_I P(T))],\label{ProofISol1}
		\end{align}
		where (a) follows from Equation \eqref{ProofI1}, and (b) from $\Lambda_I\subset\Lambda_I^{\lambda}$.	Let $\lambda^\ast\in\mathcal{D}$ and the optimal portfolio process $\pi^\ast_{\lambda^\ast}$ for the unconstrained optimization problem of the insurer $(\mathcal{P}_I^{\lambda^\ast})$ given a fixed $\xi_I$ such that $(\pi_{\lambda^\ast}^\ast,\xi_I)\in\Lambda^{\lambda^\ast}_I$ and $\pi_{\lambda^\ast}^\ast(t)\in K$ $\mathbb{Q}$-a.s. for all $t\in[0,T]$. Then:
		\begin{align}
			V_{\lambda^\ast}^{v_{I,0}(\xi_I,\theta_R),\pi^\ast_{\lambda^\ast}}(t)\stackrel{\eqref{ProofI1}}{=}V_I^{v_{I,0}(\xi_I,\theta_R),\pi_{\lambda^\ast}^\ast}(t) \quad \forall t\in[0,T]. \label{ProofI2}
		\end{align}
		Hence, $(\pi_{\lambda^\ast}^\ast,\xi_I)\in\Lambda_I$ and
		\begin{align}
			\sup_{\pi_I:(\pi_I,\xi_I)\in\Lambda_I^{\lambda^\ast}}&\mathbb{E}[U_I(V^{v_{I,0}(\xi_I,\theta_R),\pi_I}_{\lambda^\ast}(T)+\xi_I P(T))]\nonumber\\
			\stackrel{(a)}{=}&\mathbb{E}[U_I(V^{v_{I,0}(\xi_I,\theta_R),\pi_{\lambda^\ast}^\ast}_{\lambda^\ast}(T)+\xi_I P(T))]\label{ProofISol2a}\\
			\stackrel{(b)}{=}&\mathbb{E}[U_I(V^{v_{I,0}(\xi_I,\theta_R),\pi_{\lambda^\ast}^\ast}_I(T)+\xi_I P(T))]\nonumber\\
			\stackrel{(c)}{\leq}&\sup_{\pi_I:(\pi_I,\xi_I)\in\Lambda_I}\mathbb{E}[U_I(V^{v_{I,0}(\xi_I,\theta_R),\pi_I}_I(T)+\xi_I P(T))],\label{ProofISol2}
		\end{align}
		where (a) follows from the definition of $\pi_{\lambda^\ast}^\ast$, (b) from Equation \eqref{ProofI2}, and (c) from $(\pi_{\lambda^\ast}^\ast,\xi_I)\in\Lambda_I$.\\
		All in all, we have
		\begin{align*}
			\mathbb{E}[U_I(V^{v_{I,0}(\xi_I,\theta_R),\pi_{\lambda^\ast}^\ast}_{\lambda^\ast}(T)+\xi_I P(T))]\stackrel{(a)}{=}&\sup_{\pi_I:(\pi_I,\xi_I)\in\Lambda_I^{\lambda^\ast}}\mathbb{E}[U_I(V^{v_{I,0}(\xi_I,\theta_R),\pi_I}_{\lambda^\ast}(T)+\xi_I P(T))]\\
			\stackrel{(b)}{=}&\sup_{\pi_I:(\pi_I,\xi_I)\in\Lambda_I}\mathbb{E}[U_I(V^{v_{I,0}(\xi_I,\theta_R),\pi_I}_I(T)+\xi_I P(T))],
		\end{align*}
		where (a) follows from Equation \eqref{ProofISol2a}, and (b) from Inequalities \eqref{ProofISol1} and \eqref{ProofISol2}. Therefore, $\pi_{\lambda^\ast}^\ast$ is optimal for the optimization problem of the insurer \eqref{OPI} given a fixed $\xi_I$.\\
		
		\textit{Part 2.} Denote by $\xi_{\lambda^\ast}^\ast$ the optimal reinsurance amount in the unconstrained optimization problem of the insurer $(P_I^{\lambda^\ast})$ given $\pi_{\lambda^\ast}^\ast\in K$, i.e.:
		\begin{equation}\label{def:xi_nu_optimal}
		    {\color{black} \xi_{\lambda^\ast}^\ast \in \argmax_{\xi_{I} \in [0, \xi^{\max}]} \mathbb{E}\Bigl[U_I(V^{v_{I,0}(\xi_{I},\theta_R),\pi_{\lambda^\ast}^\ast(\xi_{I})}_{\lambda^\ast}(T) + \xi_{I} P(T))\Bigr].}
		\end{equation}
		Observe that:
		\begin{align}
			\{\xi_I|\;(\pi_{\lambda^\ast}^\ast,\xi_I)\in\Lambda_I^{\lambda^\ast}\}=[0,\xi^{\max}]=\{\xi_I|\;(\pi_{\lambda^\ast}^\ast,\xi_I)\in\Lambda_I\}. \label{ProofI3}
		\end{align}
		Then:
		\begin{align*}
			\mathbb{E}[U_I(V^{v_{I,0}(\xi_{\lambda^\ast}^\ast,\theta_R),\pi_{\lambda^\ast}^\ast}_I(T)+ &  \xi_{\lambda^\ast}^\ast P(T))]\stackrel{\eqref{ProofI2}}{=}\mathbb{E}[U_I(V^{v_{I,0}(\xi_{\lambda^\ast}^\ast,\theta_R),\pi_{\lambda^\ast}^\ast}_{\lambda^\ast}(T)+\xi_{\lambda^\ast}^\ast P(T))]\\
			\stackrel{\eqref{def:xi_nu_optimal}}{=}&\sup_{\xi_I:(\pi_{\lambda^\ast}^\ast,\xi_I)\in\Lambda^{\lambda^\ast}_I}\mathbb{E}[U_I(V^{v_{I,0}(\xi_I,\theta_R),\pi_{\lambda^\ast}^\ast}_{\lambda^\ast}(T)+\xi_I P(T))]\\
			\stackrel[\eqref{ProofI1}]{ \eqref{ProofI3}}{=}&\sup_{\xi_I:(\pi_{\lambda^\ast}^\ast,\xi_I)\in\Lambda_I}\mathbb{E}[U_I(V^{v_{I,0}(\xi_I,\theta_R),\pi_{\lambda^\ast}^\ast}_I(T)+\xi_I P(T))],
		\end{align*}
		Therefore, $\xi_{\lambda^\ast}^\ast$ is optimal for \eqref{OPI} given $\pi_{\lambda^\ast}^\ast\in K$
and we conclude that $(\pi_{\lambda^\ast}^\ast,\xi_{\lambda^\ast}^\ast)$ solves \eqref{OPI}.
\end{proof}

\begin{proof}[Proof of Proposition \ref{ThmOpR1}] 
    The proof is based on the proof of Theorem 4.1 in \citeA{korn1999optimal}.
    
    First, we define a new wealth process of the reinsurer with investment in the assets $S_0$, $S_1$ and $S_2$, and additionally in the put option $P$. We denote by $\xi(t)\equiv-\xi_I^\ast(\theta_R)$ the trading strategy w.r.t. $P$. The wealth process $V_R^{\bar{v}_{R,0}(\xi_I^\ast(\theta_R),\theta_R),(\varphi_R,\xi)}$ is given by
    \begin{align}
    	dV^{\bar{v}_{R,0}(\xi_I^\ast(\theta_R),\theta_R),(\varphi_R,\xi)}_R(t)=&\varphi_{R,0}(t)dS_0(t)+\varphi_{R,1}(t)dS_1(t)+\varphi_{R,2}(t)dS_2(t)+\xi(t)dP(t),\label{RWP2}\\
    	V^{\bar{v}_{R,0}(\xi_I^\ast(\theta_R),\theta_R),(\varphi_R,\xi)}_R(0)=&v_{R}+\xi_I^\ast(\theta_R)\theta_RP(0)=:\bar{v}_{R,0}(\xi_I^\ast(\theta_R),\theta_R).\nonumber
    \end{align}
    Note that $\bar{v}_{R,0}(\xi_I^\ast(\theta_R),\theta_R)$ is not equal to $v_{R,0}(\xi_I^\ast(\theta_R),\theta_R)$:
    	\begin{align*}
    		\bar{v}_{R,0}(\xi_I^\ast(\theta_R),\theta_R)=&v_{R}+\xi_I^\ast(\theta_R)\theta_RP(0)=v_{R}+\xi_I^\ast(\theta_R)(1+\theta_R)P(0)-\xi_I^\ast(\theta_R) P(0)\\
    		=&v_{R,0}(\xi_I^\ast(\theta_R),\theta_R)-\xi_I^\ast(\theta_R) P(0).
    	\end{align*}
    Since 
    \begin{align*}
    	V^{\bar{v}_{R,0}(\xi_I^\ast(\theta_R),\theta_R),(\varphi_R,\xi)}_R(T)=\varphi_{R,0}(T)S_0(T)+\varphi_{R,1}(T)S_1(T)+\varphi_{R,2}(T)S_2(T)+\xi(T)P(T)
    \end{align*}
    and the reinsurer has a short put position $-\xi_I^\ast(\theta_R)$, the optimization problem \eqref{OPRb} is equivalent to the optimization problem given by
    \begin{align}
    	\sup_{\varphi_R\in\Lambda_R^{\varphi_R|\theta_R,\xi(t)=-\xi_I^\ast(\theta_R)}}&\mathbb{E}[U_R(V^{\bar{v}_{R,0}(\xi_I^\ast(\theta_R),\theta_R),(\varphi_R,\xi)}_R(T))]\label{OPRc} \tag{$\mathcal{P}_R^{\varphi_R|\theta_R,\xi(t)=-\xi_I^\ast(\theta_R)}$}\\
    	&\text{s.t. }\xi(t)\equiv-\xi_I^\ast(\theta_R) \;\forall\tin. \nonumber
    \end{align}
    $\Lambda_R^{\varphi_R|\theta_R,\xi(t)=-\xi_I^\ast(\theta_R)}$ is the set of all admissible strategies $\varphi_R$ in the optimization problem \\
    \eqref{OPRc}:
    \begin{align*}
    	\Lambda_R^{\varphi_R|\theta_R,\xi(t)=-\xi_I^\ast(\theta_R)}:=\{\varphi_R\text{ self-financing}|\;&V^{\bar{v}_{R,0}(\xi_I^\ast(\theta_R),\theta_R),(\varphi_R,\xi)}_R(t)\geq0\;\mathbb{Q}\text{-a.s. }\forall t\in[0,T]\\
    	&\text{ and }\mathbb{E}[U_R(V^{\bar{v}_{R,0}(\xi_I^\ast(\theta_R),\theta_R),(\varphi_R,\xi)}_R(T))^-]<\infty\}.
    \end{align*}
    
	As shown on page 6 in \cite{Escobar}, the dynamics of the put is given by
	\begin{align}
		dP(t)=&[V^{v_I,\pi_B}(t)(\Phi(d_+)-1)\pi^{CM}(\mu_2-r)+rP(t)]dt \nonumber\\
		&+V^{v_I,\pi_B}(t)(\Phi(d_+)-1)\sigma_2\pi^{CM}(\rho dW_1(t)+\sqrt{1-\rho^2}dW_2(t)),\nonumber
    \end{align}
    and the trading strategy $\psi(t)=(\psi_{0}(t),\psi_{1}(t),\psi_{2}(t))^\top,\,\tin,$ that replicates the put is
    \begin{align}\label{RepSt}
    	\psi(t) := \bigg(\frac{P(t)-\pi^{CM}V^{v_I,\pi_B}(t)(\Phi(d_+)-1)}{S_0(t)},0,\frac{\pi^{CM}V^{v_I,\pi_B}(t)(\Phi(d_+)-1)}{S_2(t)}\bigg),
    \end{align}
    as shown in Lemma 1.1 in supplementary materials of this paper.
    
    Hence, for the wealth process $V^{\bar{v}_{R,0}(\xi_I^\ast(\theta_R),\theta_R),(\varphi_R,\xi)}_R$ it follows that:
	\begin{align*}
		V^{\bar{v}_{R,0}(\xi_I^\ast(\theta_R),\theta_R),(\varphi_R,\xi)}_R&(t)=\varphi_{R,0}(t)S_0(t)+\varphi_{R,1}(t)S_1(t)+\varphi_{R,2}(t)S_2(t)+\xi(t)P(t)\\
		\overset{\eqref{RepSt}}{=}&\varphi_{R,0}(t)S_0(t)+\varphi_{R,1}(t)S_1(t)+\varphi_{R,2}(t)S_2(t)
		+\xi(t)\psi_0(t)S_0(t)+\xi(t)\psi_2(t)S_2(t)\\
		=&(\varphi_{R,0}(t)+\xi(t)\psi_0(t))S_0(t)+\varphi_{R,1}(t)S_1(t)+(\varphi_{R,2}(t)+\xi(t)\psi_2(t))S_2(t)\\
		=&:\zeta_{R,0}(t)S_0(t)+\zeta_{R,1}(t)S_1(t)+\zeta_{R,2}(t)S_2(t),
	\end{align*}
    where
	\begin{align}
		\zeta_R(t)&=(\zeta_{R,0}(t),\zeta_{R,1}(t),\zeta_{R,2}(t))^\top
		:=(\varphi_{R,0}(t)+\xi(t)\psi_0(t),\varphi_{R,1}(t),\varphi_{R,2}(t)+\xi(t)\psi_2(t))^\top\label{zeta}
	\end{align} 
	is a self-financing trading strategy. Hence, the dynamics is given by
	\begin{align*}
		dV^{\bar{v}_{R,0}(\xi_I^\ast(\theta_R),\theta_R),(\varphi_R,\xi)}_R(t)=&\zeta_{R,0}(t)dS_0(t)+\zeta_{R,1}(t)dS_1(t)+\zeta_{R,2}(t)dS_2(t).
	\end{align*}
	The wealth process $V^{\bar{v}_{R,0}(\xi_I^\ast(\theta_R),\theta_R),(\varphi_R,\xi)}_R$ equals the wealth process of the reinsurer \textcolor{black}{w.r.t.} the trading strategy $\zeta_R$ (i.e., only an investment in the assets $S_0$, $S_1$ and $S_2$ without an investment in the put option $P$). If the trading strategy $\varphi_R$ is admissible for the optimization problem \eqref{OPRc}, then the trading strategy $\zeta_R$ is admissible to the portfolio optimization problem $(P_R^{\zeta_R|\theta_R,\xi(t)=0})$:
	\begin{align*}
		V^{\bar{v}_{R,0}(\xi_I^\ast(\theta_R),\theta_R),(\zeta_R,0)}(t)&=V^{\bar{v}_{R,0}(\xi_I^\ast(\theta_R),\theta_R),(\varphi_R,\xi)}_R(t)\geq0\;\forall t\in[0,T]\\
		&\text{and}\\
		\mathbb{E}[U_R(V^{\bar{v}_{R,0}(\xi_I^\ast(\theta_R),\theta_R),(\zeta_R,0)}(T))^-]&=\mathbb{E}[U_R(V^{\bar{v}_{R,0}(\xi_I^\ast(\theta_R),\theta_R),(\varphi_R,\xi)}_R(T))^-]<\infty. 
	\end{align*}
	By the standard Martingale method (i.e., the investor/reinsurer has only an investment in the assets $S_0$, $S_1$ and $S_2$), there exists an optimal trading strategy $\zeta_R^\ast$ to the optimization problem $(P_R^{\zeta_R|\theta_R,\xi(t)=0})$ and the optimal terminal wealth $V^{\bar{v}_{R,0}(\xi_I^\ast(\theta_R),\theta_R),(\zeta_R^\ast,0)}_R$ is given by 
	\begin{align*}
		V^{\bar{v}_{R,0}(\xi_I^\ast(\theta_R),\theta_R),(\zeta_R^\ast,0)}_R(T)=I_R(y^\ast_R(\theta_R)\widetilde{Z}(T)),
	\end{align*}
	where $y^\ast_R\equiv y^\ast_R(\theta_R)$ is determined by the budget constraint
	\begin{align*}
		\mathbb{E}[\widetilde{Z}(T)I_R(y^\ast_R\widetilde{Z}(T))]=v_R+\xi_I^\ast(\theta_R)\theta_RP(0).
	\end{align*}
	Therefore, there exists an optimal trading strategy $\varphi_R^\ast$ for the optimization problem\\ \eqref{OPRc} and the optimal wealth process $V^{\bar{v}_{R,0}(\xi_I^\ast(\theta_R),\theta_R),(\varphi_R^\ast,\xi)}_R$ is given by
	\begin{align*}
		V^{\bar{v}_{R,0}(\xi_I^\ast(\theta_R),\theta_R),(\varphi_R^\ast,\xi)}_R(T)&=V^{\bar{v}_{R,0}(\xi_I^\ast(\theta_R),\theta_R),(\zeta_R^\ast,0)}_R(T)=I_R(y^\ast_R(\theta_R)\widetilde{Z}(T)).
	\end{align*}
	By \eqref{zeta}, we get for the optimal trading strategy $\varphi_R^\ast$ the following representation:
	\begin{align}
		\varphi_{R,1}^\ast(t)=&\zeta_{R,1}^\ast(t)\label{varphiR1}\\
		\varphi_{R,2}^\ast(t)=&\zeta_{R,2}^\ast(t)-\xi(t)\psi_2(t) 
		=\zeta_{R,2}^\ast(t)+\xi_I^\ast(\theta_R)\psi_2(t) \label{varphiR2}\\
		\varphi_{R,0}^\ast(t)=&\zeta_{0R}^\ast(t)-\xi(t)\psi_0(t)
		=\frac{V^{\bar{v}_{R,0}(\xi_I^\ast(\theta_R),\theta_R),(\varphi_R^\ast,\xi)}(t)-\sum_{i=1}^2\varphi^\ast_{R,i}(t)S_i(t)+\xi_I^\ast(\theta_R)P(t)}{S_0(t)}.\label{varphiR0}
	\end{align}
	Since the optimization problems \eqref{OPRb} and \eqref{OPRc} are equivalent, it holds that there exists an optimal trading strategy $\varphi_R^\ast$ to the optimization problem \eqref{OPRb} given by \eqref{varphiR0}, \eqref{varphiR1} and \eqref{varphiR2}. The optimal terminal wealth of the reinsurer is given by
	\begin{align*}
	    V^{v_{R,0}(\xi_I^\ast(\theta_R),\theta_R),\varphi_R^\ast}_R(T)&=V^{\bar{v}_{R,0}(\xi_I^\ast(\theta_R),\theta_R),(\varphi_R^\ast,\xi)}_R(T)+\xi_I^\ast(\theta_R)P(T)
	    =I_R(y^\ast_R(\theta_R)\widetilde{Z}(T))+\xi_I^\ast(\theta_R)P(T)
	\end{align*}
	and the optimal wealth process by
	\begin{align*}
	    V^{v_{R,0}(\xi_I^\ast(\theta_R),\theta_R),\varphi_R^\ast}_R(t)=&\widetilde{Z}(t)^{-1}\mathbb{E}\left[\widetilde{Z}(T)(V^{\bar{v}_{R,0}(\xi_I^\ast(\theta_R),\theta_R),(\varphi_R^\ast,\xi)}_R(T)+\xi_I^\ast(\theta_R)P(T))|\mathcal{F}_t\right]\\
	    =&V^{\bar{v}_{R,0}(\xi_I^\ast(\theta_R),\theta_R),(\varphi_R^\ast,\xi)}_R(t)+\xi_I^\ast(\theta_R)P(t).
	\end{align*}
	Hence, it follows for the optimal trading strategy $\varphi_R^\ast$ from \eqref{varphiR0}, \eqref{varphiR1} and \eqref{varphiR2}
	\begin{align*}
	    \varphi_{R,1}^\ast(t)=&\zeta_{R,1}^\ast(t),\;\\
		\varphi_{R,2}^\ast(t)=&\zeta_{R,2}^\ast(t)+\xi_I^\ast(\theta_R)\psi_2(t)\\ 
		\varphi_{R,0}^\ast(t)=&\frac{V^{v_{R,0}(\xi_I^\ast(\theta_R),\theta_R),\varphi_R^\ast}(t)-\sum_{i=1}^2\varphi^\ast_{R,i}(t)S_i(t)}{S_0(t)}.
	\end{align*}
\end{proof}

\begin{proof}[Proof of Proposition \ref{OpSt}]
	We define the function $\kappa:[0,\theta^{\max}]\to\mathbb{R}$ by
	\begin{align*}
		\kappa(\theta_R):=&\mathbb{E}[U_R(V^{v_{R,0}(\xi_I^\ast(\theta_R),\theta_R),\varphi^\ast_R}_R(T)-\xi_I^\ast(\theta_R)P(T))]
		=\mathbb{E}[U_R(I_R(y^\ast_R(\theta_R)\widetilde{Z}(T)))],
	\end{align*}
	where $y^\ast_R\equiv y^\ast_R(\theta_R)$ is the Lagrange multiplier determined by the budget constraint
	\begin{align*}
		\mathbb{E}[\widetilde{Z}(t)I_R(y^\ast_R\widetilde{Z}(T))]=v_{R,0}+\xi_I^\ast(\theta_R)\theta_RP(0).
	\end{align*}
	We show that the map $\theta_R\mapsto\kappa(\theta_R)$ is continuous. As per Proposition \ref{SolIOP2}, $\xi_I^\ast(\cdot)$ is given by
	\begin{align*}
		\xi_I^\ast(\theta)=\underset{\xi\in[0,\xi^{\max}(\theta)]}{\text{arg max }}\nu(\xi,\theta)
	\end{align*}
	with $\xi^{\max}(\theta)=\min\{\bar{\xi},\frac{v_I}{(1+\theta)P(0)}\}$ and $\nu(\xi,\theta):=\mathbb{E}[U_I(\max\{I_I(y_I^\ast(\xi,\theta)\widetilde{Z}_{\lambda^\ast}(T)),\xi P(T)\})]$,
	where the Lagrange multiplier $y^\ast_I\equiv y_I^\ast(\xi,\theta)$ is given by the budget constraint of the insurer
	\begin{align*}
		\mathbb{E}[\widetilde{Z}_{\lambda^\ast}(T)\hat{I}_I(y_I^\ast\widetilde{Z}_{\lambda^\ast}(T))]=v_I-\xi(1+\theta_R)P(0).
	\end{align*}
	The Lagrange multiplier $y_I^\ast(\xi,\theta)$ is continuous \textcolor{black}{w.r.t.} $\xi$ and $\theta$, since $\hat{I}_I$ is a continuous function and $v_I-\xi(1+\theta_R)P(0)$ is continuous \textcolor{black}{w.r.t.} $\xi$ and $\theta$. Furthermore, we have that the functions $I_I$, $\max$ and $U_I$ are continuous. Therefore, the function $\nu$ is continuous \textcolor{black}{w.r.t.} $\xi$ and $\theta$. In addition $\nu$ is strictly concave \textcolor{black}{w.r.t.} $\xi$ (by Lemma A.3 in \citeA{desmettre2016optimal}), if $U_I$ is strictly concave. Since $U_I$ is a utility function, it is strictly concave. Furthermore, the map $\theta\mapsto\xi^{\max}(\theta)$ is continuous. By the Berges Maximum Theorem we conclude that the map $\theta\mapsto\xi_I^\ast(\theta)$ is continuous too.
	
	Next, we argue that $\theta\mapsto\kappa(\theta)$ is continuous. The Lagrange multiplier $y_R^\ast(\theta)$ is continuous, since $I_R$ is a continuous function and $v_{R,0}+\xi_I^\ast(\theta)\theta P(0)$ is continuous \textcolor{black}{w.r.t.} $\theta$. Furthermore, we have that the functions $I_R$ and $U_R$ are continuous. Therefore, the function $\kappa$ is continuous \textcolor{black}{w.r.t.} $\theta$. Since $[0,\theta^{\max}]$ is compact, it follows from the Weierstrass Theorem that there exists $\theta_R^\ast$ such that $\theta_R^\ast=\underset{\theta_R\in[0,\theta^{\max}]}{\text{arg max }}\kappa(\theta_R)$.
\end{proof}

\begin{proof}[Proof of Proposition \ref{SE}]
    We show that Conditions \eqref{eq:SE_condition_1} and \eqref{eq:SE_condition_2} are fulfilled:
        
        By Propositions \ref{SolIOP2} and \ref{SolOPI}, the optimal response $(\pi_I^\ast(\cdot|\theta_R^\ast),\xi_I^\ast(\theta_R^\ast))$ solves the optimization problem of the insurer. Thus, \eqref{eq:SE_condition_1} holds by definition.
        
        By Propositions \ref{ThmOpR1} and \ref{OpSt}, the optimal strategy $(\pi_R^\ast(\cdot),\theta_R^\ast)$ solves the optimization problem of the reinsurer. If the insurer has several replies to the reinsurer's strategy $(\pi_R^\ast,\theta_R^\ast)$, then we need to find $(\pi_I^\ast(\cdot|\theta_R^\ast),\xi_I^\ast(\theta_R^\ast))$ from the set of best replies of the insurer, i.e.,
        \begin{align}
            \mathbb{E}[U_R(V_R^{v_{R,0}(\xi_I,\theta_R^\ast),\pi_R^\ast}(T)-\xi_I P(T))]\leq \mathbb{E}[U_R(V_R^{v_{R,0}(\xi_I^\ast(\theta_R^\ast),\theta_R^\ast),\pi_R^\ast}(T)-\xi_I^\ast(\theta_R^\ast) P(T))] \label{Inequ}
        \end{align}
        for all $\xi_I$ in the set of best replies of the insurer. Note that the reinsurer's value function does not depend on $\pi_I$. Thus, we focus only on $\xi_I$ and show now that the reinsurer's value function is increasing in $\xi_I$. For the total terminal wealth of the reinsurer it holds for $\xi_I\in[0,\xi^{\max}]$
        \begin{align}
            V_R^{v_{R,0}(\xi_I,\theta_R^\ast),\pi_R^\ast}(T)-\xi_I P(T)=I_R(y_R^\ast\widetilde{Z}(T)), \label{TotalTerminalWealthR}
        \end{align}
        where $y_R^\ast$ solves the budget constraint
        \begin{align}
            \mathbb{E}[\widetilde{Z}(T)I_R(y_R^\ast\widetilde{Z}(T))]=v_R+\xi_I\theta_R^\ast P(0). \label{LagrangeR}
        \end{align}
        It holds
        \begin{align*}
            &\mathbb{E}\left[U_R(V_R^{v_{R,0}(\xi_I,\theta_R^\ast),\pi_R^\ast}(T)-\xi_I P(T))\right]\text{ increasing w.r.t. }\xi_I\\
            \overset{(a)}{\Leftrightarrow}\;&V_R^{v_{R,0}(\xi_I,\theta_R^\ast),\pi_R^\ast}(T)-\xi_I P(T)\text{ increasing w.r.t. }\xi_I\;\overset{(b)}{\Leftrightarrow}\;y_R^\ast\text{ decreasing w.r.t. }\xi_I
        \end{align*}
        where (a) follows from the fact that $U_R$ is an increasing function and (b) from \eqref{TotalTerminalWealthR} and the fact that $I_R$ is a decreasing function due to the strict concavity of $U_R$. By \eqref{LagrangeR}, the Lagrange multiplier decreases if and only if $\xi_I$ increases, since $I_R$ is a decreasing function. Hence, the value function of the reinsurer increases if $\xi_I$ increases. Therefore, \eqref{Inequ} is fulfilled if the optimal reinsurance strategy of the insurer in the Stackelberg equilibrium is given by
        \begin{align*}
            \xi_I^\ast(\theta_R^\ast)=\max\{\xi_I^\ast|\;&\mathbb{E}[U_I(V_I^{v_{I,0}(\xi_I^\ast,\theta_R^\ast),\pi_I^\ast}(T)+\xi_I^\ast P(T))]\nonumber\\
    		&=\sup_{\xi_I:(\xi_I,\pi_I)\in\Lambda_I}\mathbb{E}[U_I(V_I^{v_{I,0}(\xi_I,\theta_R^\ast),\pi_I^\ast}(T)+\xi_IP(T))]\}.
        \end{align*}
\end{proof}

\section{Proofs of example results}\label{app:proofs_example}
    \begin{proof}[Proof of Corollary \ref{CorPowerI}]
        For $\lambda\in\mathcal{D}$, the Merton portfolio process $\pi^M_\lambda$ in the auxiliary market $\mathcal{M}_\lambda$ is given by
        \begin{align*}
            \pi^M_\lambda(b_I)=\frac{1}{1-b_I}(\sigma\sigma^\top)^{-1}(\mu+\lambda(t)-r\mathbbm{1}).
        \end{align*}
        By Proposition \ref{SolIOP2}, the insurer's optimal portfolio process $\pi_\lambda^\ast$ in the auxiliary market $\mathcal{M}_\lambda$ satisfies
        \begin{align*}
            \pi_\lambda^\ast(t)V_\lambda^\ast(t)=\pi_\lambda^M(b_I)(V_\lambda^\ast(t)+\xi_\lambda^\ast \widetilde{Z}_\lambda(t)^{-1}\mathbb{E}[\widetilde{Z}_\lambda(T)P(T){1}_{\{V^\ast_\lambda(T)>0\}}|\mathcal{F}_t]),
        \end{align*}
        where $V^\ast_\lambda$ is the insurer's optimal wealth process in $\mathcal{M}_\lambda$. Since $\xi^{\max}<\frac{v_I}{P(0)}$ and $\xi_\lambda^\ast \leq \xi^{\max}$, it holds $v_{I,0}^\lambda(\xi_\lambda^\ast,\theta_R)>0$ and, therefore, $V^\ast_\lambda(t)>0$ for all $\tin$. Therefore:
        \begin{align*}
            \pi_\lambda^\ast(t)=\pi_\lambda^M(b_I)\frac{V_\lambda^\ast(t)+\xi_\lambda^\ast \widetilde{Z}_\lambda(t)^{-1}\mathbb{E}[\widetilde{Z}_\lambda(T)P(T)|\mathcal{F}_t]}{V_\lambda^\ast(t)}.
        \end{align*}
        We find now $\lambda^\ast\in\mathcal{D}$ such that $\pi_{\lambda^\ast}^\ast(t)\in K=\mathbb{R}\times\{0\}$ for all $\tin$.
        
        Since $\pi_{\lambda}^\ast$ is given by $\pi_\lambda^M$ multiplied by a random variable bigger than zero, it is sufficient to find $\lambda^\ast\in\mathcal{D}$ such that $\pi_{\lambda^\ast}^M\in K$. Hence,
        \begin{align*}
            \pi_{\lambda^\ast}^M(b_I)&\in K\Leftrightarrow\frac{1}{1-b_I}(\sigma\sigma^\top)^{-1}\begin{pmatrix}\mu_1-r\\ \mu_2+\lambda_2^\ast(t)-r\end{pmatrix}\in K.
        \end{align*}
        Since $\lambda^\ast\in\mathcal{D}$ has to hold we have $\lambda^\ast_1(t)=0$. It follows
        \begin{align*}
            \lambda^\ast(t)\equiv\lambda^\ast=\begin{pmatrix}0 \\ \frac{\sigma_2\rho}{\sigma_1}(\mu_1-r)-\mu_2+r \end{pmatrix}.
        \end{align*}
        
        From Lemma 1.2 in supplementary materials, we get that for a power-utility function:
        \begin{equation}\label{eq:nu_power_utility}
            \nu(\xi)=\mathbb{E}\left[\frac{1}{b_I}(y_I^\ast(\xi)\widetilde{Z}_{\lambda^\ast}(T))^{\frac{b_I}{b_I-1}}\right],
        \end{equation}
        where the Lagrange multiplier $y^\ast_I(\xi)$ is determined by the budget constraint
        \begin{align*}
            \mathbb{E}\left[\widetilde{Z}_{\lambda^\ast}(T)\left((y^\ast_I(\xi)\widetilde{Z}_{\lambda^\ast}(T))^{\frac{1}{b_I-1}}-\xi P(T)\right)\right]=v_I-\xi(1+\theta_R)P(0).
        \end{align*}
        Hence, we obtain
        \begin{equation}\label{eq:yStar_power_utility}
            y^\ast_I(\xi)=(v_I-\xi(1+\theta_R)P(0)+\xi\mathbb{E}[\widetilde{Z}_{\lambda^\ast}(T)P(T)])^{b_I-1}\mathbb{E}[\widetilde{Z}_{\lambda^\ast}(T)^\frac{b_I}{b_I-1}]^{1-b_I}
        \end{equation}
        and, therefore,
        \begin{align*}
            \nu(\xi)
            &\stackrel[\eqref{eq:yStar_power_utility}]{\eqref{eq:nu_power_utility}}{=}\frac{1}{b_I}(v_I-\xi(1+\theta_R)P(0)+\xi\mathbb{E}[\widetilde{Z}_{\lambda^\ast}(T)P(T)])^{b_I}\mathbb{E}\left[\widetilde{Z}_{\lambda^\ast}(T)^\frac{b_I}{b_I-1}\right]^{1-b_I}.
        \end{align*}
        It follows for the optimal reinsurance strategy $\xi_I^\ast=\xi_I^\ast(\theta_R)$ that:
    	\begin{align}
    		\xi_I^\ast=&\underset{\xi_I\in[0,\bar{\xi}]}{\text{arg max }}\bigg(\frac{1}{b_I}(v_I-\xi_I(1+\theta_R)P(0)+\xi_I\mathbb{E}[\widetilde{Z}_{\lambda^\ast}(T)P(T)])^{b_I}\mathbb{E}\Big[\widetilde{Z}_{\lambda^\ast}(T)^{\frac{b_I}{b_I-1}}\Big]^{1-b_I}\bigg) \nonumber \\
    		=&\begin{cases}
    			\bar{\xi},&\text{if }-(1+\theta_R)P(0)+\mathbb{E}[\widetilde{Z}_{\lambda^\ast}(T)P(T)]>0,\\
    			\text{any } \tilde{\xi}\in[0,\bar{\xi}],&\text{if }-(1+\theta_R)P(0)+\mathbb{E}[\widetilde{Z}_{\lambda^\ast}(T)P(T)]=0,\\
    			0,&\text{if }-(1+\theta_R)P(0)+\mathbb{E}[\widetilde{Z}_{\lambda^\ast}(T)P(T)]<0.
    		\end{cases} \nonumber\\
    		=&\begin{cases}
    			\bar{\xi},&\text{if } \theta_R < \frac{\mathbb{E}[\widetilde{Z}_{\lambda^\ast}(T)P(T)]-P(0)}{P(0)},\\
    			\text{any } \tilde{\xi}\in[0,\bar{\xi}],&\text{if } \theta_R = \frac{\mathbb{E}[\widetilde{Z}_{\lambda^\ast}(T)P(T)]-P(0)}{P(0)},\\
    			0,&\text{if } \theta_R > \frac{\mathbb{E}[\widetilde{Z}_{\lambda^\ast}(T)P(T)]-P(0)}{P(0)}.
    		\end{cases} \label{eq:xi_optimal}
    	\end{align}
    \end{proof}

    \begin{proof}[Proof of Corollary \ref{CorPowerR}]
        Due to Corollary \ref{CorPowerI}, the optimal reinsurance strategy $\xi_I^\ast(\theta_R)$ is given by \eqref{eq:xi_optimal}.
        From Proposition \ref{OpSt} and \eqref{TotalTerminalWealthR} it follows for the optimal safety loading $\theta_R^\ast$ that
    	\begin{align*}
    		\theta_R^\ast=&\underset{\theta_R\in[0,\theta^{\max}]}{\text{arg max }}\mathbb{E}\left[\frac{1}{b_R}(y^\ast_R(\theta_R)\widetilde{Z}(T))^{\frac{b_R}{b_R-1}}\right]\\
    		\overset{\eqref{LagrangeMultiplierR}}{=}&\underset{\theta_R\in[0,\theta^{\max}]}{\text{arg max }}\frac{1}{b_R}(v_{R}+\xi_I^\ast(\theta_R)\theta_RP(0))^{b_R}\mathbb{E}\left[\widetilde{Z}(T)^{{\frac{b_R}{b_R-1}}}\right]^{1-b_R}
    		=\underset{\theta_R\in[0,\theta^{\max}]}{\text{arg max }}\xi_I^\ast(\theta_R)\theta_R.
    	\end{align*}
    	Hence, the reinsurer chooses the largest $\theta_R\in[0,\theta^{\max}]$ such that $\xi_I^\ast(\theta_R)=\bar{\xi}$, i.e.,
    	\begin{align*}
    		\theta_R^\ast=\min\bigg\{\frac{\mathbb{E}[\widetilde{Z}_{\lambda^\ast}(T)P(T)]-P(0)}{P(0)},\theta^{\max}\bigg\}.
    	\end{align*}
    	By the definition of the Stackelberg equilibrium (see Conditions \eqref{eq:SE_condition_1} and \eqref{eq:SE_condition_2}), the optimal reinsurance strategy of the insurer is given by $\xi_I^\ast(\theta_R^\ast)=\bar{\xi}$ (i.e., the insurer chooses the response to the optimal safety loading, which is the best for the reinsurer).
    \end{proof}

\begin{center}
\Huge Supplementary materials
\end{center}
Section \ref{sm:auxiliary_lemmata} contains two lemmas that we use to prove the main results of the paper. Section \ref{sm_sec:SA_figures} contains plots that complement the numerical studies section from the paper. In particular, Subsection \ref{subsec:SA_plots} contains plots related to the sensitivity analysis of the Stackelberg equilibrium w.r.t. RRA coefficients, $r$, $T$, $G_T$, whereas the description of these plots was included in the article in Section 5.2. In Subsection \ref{subsec:dynamic_investment_strategies}, we provide plots that illustrate the dynamic relative-portfolio process of each company over the entire investment horizon. In Section \ref{sm:reinsurer_incentive} we prove that the reinsurance company is better off when it sells reinsurance with a discounted safety loading $\theta(\alpha) = \alpha \cdot \theta_R^\ast$ in comparison to not selling reinsurance at all. We also illustrate the corresponding monetary benefit with the help of the wealth-equivalent utility gain, defined in Section 5.3. of the paper. In Section \ref{sm_sec:other_utilities}, we derive the insurer's optimal strategies in the Stackelberg equilibrium, when the insurance company has a logarithmic-utility and a HARA-utility function. Finally, in Section \ref{sm_sec:actuarial_risks} we discuss ways of adding mortality and surrender risks to the model and the corresponding potential challenges to the derivation of the Stackelberg equilibrium.  

\section{Auxiliary lemmas}\label{sm:auxiliary_lemmata}
\begin{lemma}[Put-replicating trading strategies] \label{RepStLem}
The replicating strategy $\psi(t)$, $t\in[0,T]$, of the put option $P$ is given by
	\begin{align*}
		\psi(t)=\bigg(\frac{P(t)-\pi^{CM}V^{v_I,\pi_B}(t)(\Phi(d_+)-1)}{S_0(t)},0,\frac{\pi^{CM}V^{v_I,\pi_B}(t)(\Phi(d_+)-1)}{S_2(t)}\bigg)^\top,
	\end{align*}
	where
	\begin{align*}
		d_+:=d_+(t,V^{v_I,\pi_B}(t)):=\frac{\ln\big(\frac{V^{v_I,\pi_B}(t)}{G_T}\big)+\big(r+\frac{1}{2}(\sigma_2\pi^{CM})^2\big)(T-t)}{\pi^{CM}\sigma_2\sqrt{T-t}}.
	\end{align*}
	The dynamics of the put option $P$ is given by
	\begin{align*}
		dP(t)=&[V^{v_I,\pi_B}(t)(\Phi(d_+)-1)\pi^{CM}(\mu_2-r)+rP(t)]dt \\ 
		&+V^{v_I,\pi_B}(t)(\Phi(d_+)-1)\sigma_2\pi^{CM}(\rho dW_1(t)+\sqrt{1-\rho^2}dW_2(t)). \nonumber
    \end{align*}
\end{lemma}
\begin{proof}
    The price of the put option $P$ at time $\tin$ is given by
	\begin{align}
		P(t)=&\widetilde{Z}(t)^{-1}\mathbb{E}[\widetilde{Z}(T)P(T)|\mathcal{F}_t]\nonumber\\
		=&\exp \rBrackets{-r(T - t)}G_T\Phi(-d_1(t,V^{v_I,\pi_B}(t)))-V^{v_I,\pi_B}(t)\Phi(-d_2(t,V^{v_I,\pi_B}(t))),\nonumber 
	\end{align}
	where $\Phi$ is the cumulative distribution function of the standard normal distribution and
	\begin{align*}
		d_1(t,V^{v_I,\pi_B}(t))&:=\frac{\ln\big(\frac{V^{v_I,\pi_B}(t)}{G_T}\big)+\big(r-\frac{1}{2}(\pi^{CM}\sigma_2)^2\big)(T-t)}{\pi^{CM}\sigma_2\sqrt{T-t}},\\
		d_2(t,V^{v_I,\pi_B}(t))&:=d_1(t,V^{v_I,\pi_B}(t))+\pi^{CM}\sigma_2\sqrt{T-t}.
	\end{align*}
	The stock price $S_2$ and the constant-mix portfolio value $V^{v_I,\pi_B}$ are given by
	\begin{align*}
		S_2(t)=&S_2(0)\exp\Big(\Big(\mu_2-\frac{1}{2}\sigma_2^2\Big)t+\sigma_2(\rho W_1(t)+\sqrt{1-\rho^2}W_2(t))\Big),\\
		V^{v_I,\pi_B}(t)=&v_I\exp\Big(\Big(r+\pi^{CM}(\mu_2-r)-\frac{1}{2}(\sigma_2\pi^{CM})^2\Big)t+\sigma_2\pi^{CM}(\rho W_1(t)+\sqrt{1-\rho^2}W_2(t))\Big).
	\end{align*}
	Hence, the relation between $S_2$ and $V^{v_I,\pi_B}$ is given by
	\begin{align*}
		V^{v_I,\pi_B}(t)=\frac{v_I}{S_2(0)}\exp\Big(\Big(r+\frac{1}{2}\pi^{CM}\sigma_2^2\Big)(1-\pi^{CM})t\Big)S_2(t)^{\pi^{CM}}.
	\end{align*}
	Therefore,
	\begin{align*}
		\frac{\partial P(t)}{\partial S_1(t)}=0
	\end{align*}
	and
	\begin{align*}
		\frac{\partial P(t)}{\partial S_2(t)}=&\exp \rBrackets{-r(T - t)}G_T\frac{\partial}{\partial S_2(t)}\Phi(-d_1(t,V^{v_I,\pi_B}(t)))-\frac{\partial V^{v_I,\pi_B}(t)}{\partial S_2(t)}\Phi(-d_2(t,V^{v_I,\pi_B}(t)))\\
		&-V^{v_I,\pi_B}(t)\frac{\partial}{\partial S_2(t)}\Phi(-d_2(t,V^{v_I,\pi_B}(t))).
	\end{align*}
	We have for $i=1,2$
	\begin{align*}
		\frac{\partial V^{v_I,\pi_B}(t)}{\partial S_2(t)}=&V^{v_I,\pi_B}(t)\pi^{CM}S_2(t)^{-1},\\
		\frac{\partial (d_i(t,V^{v_I,\pi_B}(t)))}{\partial S_2(t)}=&\frac{1}{\sigma_2\sqrt{T-t}}S_2(t)^{-1},\\
		\frac{\partial\Phi(-d_i(t,V^{v_I,\pi_B}(t)))}{\partial S_2(t)}=&-\phi(-d_i(t,V^{v_I,\pi_B}(t)))\frac{\partial (d_i(t,V^{v_I,\pi_B}(t)))}{\partial S_2(t)},
	\end{align*}
	where $\phi$ is the density of the standard normal distribution. Hence,
	\begin{align*}
		\frac{\partial P(t)}{\partial S_2(t)}=&\exp \rBrackets{-r(T - t)}G_T\frac{\partial}{\partial S_2(t)}\Phi(-d_1(t,V^{v_I,\pi_B}(t)))-\frac{\partial V^{v_I,\pi_B}(t)}{\partial S_2(t)}\Phi(-d_2(t,V^{v_I,\pi_B}(t)))\\
		&-V^{v_I,\pi_B}(t)\frac{\partial}{\partial S_2(t)}\Phi(-d_2(t,V^{v_I,\pi_B}(t)))\\
		=&-\exp \rBrackets{-r(T - t)}G_T\phi(-d_1(t,V^{v_I,\pi_B}(t)))\frac{1}{\sigma_2\sqrt{T-t}}S_2(t)^{-1}\\
		&-\frac{\pi^{CM}V^{v_I,\pi_B}(t)\Phi(-d_2(t,V^{v_I,\pi_B}(t)))}{S_2(t)}\\
		&+V^{v_I,\pi_B}(t)\phi(-d_2(t,V^{v_I,\pi_B}(t)))\frac{1}{\sigma_2\sqrt{T-t}}S_2(t)^{-1}\\
		=&-\frac{\pi^{CM}V^{v_I,\pi_B}(t)\Phi(-d_2(t,V^{v_I,\pi_B}(t)))}{S_2(t)}.
	\end{align*}
	We define $d_2(t,V^{v_I,\pi_B}(t))=:d_+$. Since $\Phi(-x)=1-\Phi(x)$, we get
	\begin{align*}
		\frac{\partial P(t)}{\partial S_2(t)}=-\frac{\pi^{CM}V^{v_I,\pi_B}(t)\Phi(-d_2(t,V^{v_I,\pi_B}(t)))}{S_2(t)}=\frac{\pi^{CM}V^{v_I,\pi_B}(t)(\Phi(d_+)-1)}{S_2(t)}.
	\end{align*}
\end{proof}

\bigskip

\begin{lemma} \label{HLemma1}
        Let $\xi^{\max} = \bar{\xi} < \frac{v_I}{(1 + \theta_R^{\max}) P(0)}$. Then the function $\nu$ from Proposition 3.1 is 
        \begin{align} 
            \nu(\xi)=\mathbb{E}[U_I(I_I(y^\ast_\lambda(\xi)\widetilde{Z}_{\lambda^\ast}(T)))] \nonumber 
        \end{align}
        for $\xi\in[0,\xi^{\max}]$, where $I_I$ is the inverse of $U'_I$ and $y^\ast_\lambda(\xi)$ is the Lagrange multiplier given by
        \begin{align*}
            \mathbb{E}[\widetilde{Z}_\lambda(T)\hat{I}_I(y^\ast_\lambda(\xi)\widetilde{Z}_\lambda(T))]=v_I-\xi(1+\theta_R)P(0).
        \end{align*}
    \end{lemma}
    \begin{proof}
        Recall from Proposition 3.1:
        \begin{align*}
            \nu(\xi):=\mathbb{E}[U_I(\max\{I_I(y^\ast_\lambda(\xi)\widetilde{Z}_{\lambda^\ast}(T)),\xi P(T)\})],
        \end{align*}
        where the Lagrange multiplier $y^\ast_\lambda(\xi)$ is given by
        \begin{align*}
            \mathbb{E}[\widetilde{Z}_\lambda(T)\hat{I}_I(y^\ast_\lambda(\xi)\widetilde{Z}_\lambda(T))]=v_I-\xi(1+\theta_R)P(0).
        \end{align*}
        $\hat{I}_I$ is the inverse function of $\hat{U}_I'$. This function is bijective on $(0,U'_I(\xi P(T))]$ and equals
        \begin{align*}
            \hat{I}_I(y)=I_I(y)-\xi P(T)
        \end{align*}
        for $y\in(0,U_I'(\xi P(T))]$. From this we have the following:
        \begin{align}
            \hat{I}_I(y)>0 \Leftrightarrow y\in(0,U_I'(\xi P(T))). \label{Fact1}
        \end{align}
        For the Lagrange multiplier $y^\ast_\lambda(\xi)$ it follows that
        \begin{align*}
            \mathbb{E}[\widetilde{Z}_{\lambda}(T)\hat{I}_I(y^\ast_\lambda(\xi)\widetilde{Z}_{\lambda}&(T))]=\underbrace{v_I-\xi(1+\theta_R)P(0)}_{>0}\Leftrightarrow\hat{I}_I(y^\ast_\lambda(\xi)\widetilde{Z}_{\lambda}(T))>0\;\mathbb{Q}\text{-a.s.}\\
    		\Leftrightarrow y^\ast_\lambda(\xi)\widetilde{Z}_{\lambda}(T)&<U_I'(\xi P(T))\;\mathbb{Q}\text{-a.s.}\Leftrightarrow I_I(y^\ast_\lambda(\xi)\widetilde{Z}_{\lambda}(T))>\xi P(T)\;\mathbb{Q}\text{-a.s.},
        \end{align*}
        where the second equivalence holds from \eqref{Fact1} and the third from the fact that $I_I$ is strictly decreasing. Hence, we get
        \begin{equation*}
            \nu(\xi)=\mathbb{E}[U_I(\max\{I_I(y^\ast_\lambda(\xi)\widetilde{Z}_{\lambda^\ast}(T)),\xi P(T)\})]=\mathbb{E}[U_I(I_I(y^\ast_\lambda(\xi)\widetilde{Z}_{\lambda^\ast}(T)))].
        \end{equation*}
    \end{proof}

\section{Additional figures}\label{sm_sec:SA_figures}
{\color{black} This section contains two subsections. In Subsection \ref{subsec:SA_plots}, we provide plots related to the sensitivity analysis of the Stackelberg equilibrium at the start of the product. In Subsection \ref{subsec:dynamic_investment_strategies}, we plot the optimal relative-portfolio processes of the players.}

\subsection{Sensitivity analysis of Stackelberg equilibrium at $t=0$}\label{subsec:SA_plots}
{\color{black} Figures \ref{fig:SensitivityRRA}, \ref{fig:SensitivityIR}, \ref{fig:SensitivityTH} and \ref{fig:SensitivityG} graphically illustrate the sensitivity of the Stackelberg equilibrium at $t = 0$ w.r.t. the RRA coefficients of parties, the interest rate, the time to product maturity and the level of the capital guarantee. The description of these plots can be found in Section 5.2 of the article.}

\begin{figure}[!ht]
   	\begin{minipage}[h]{.45\linewidth} 
   		\includegraphics[width=\linewidth]{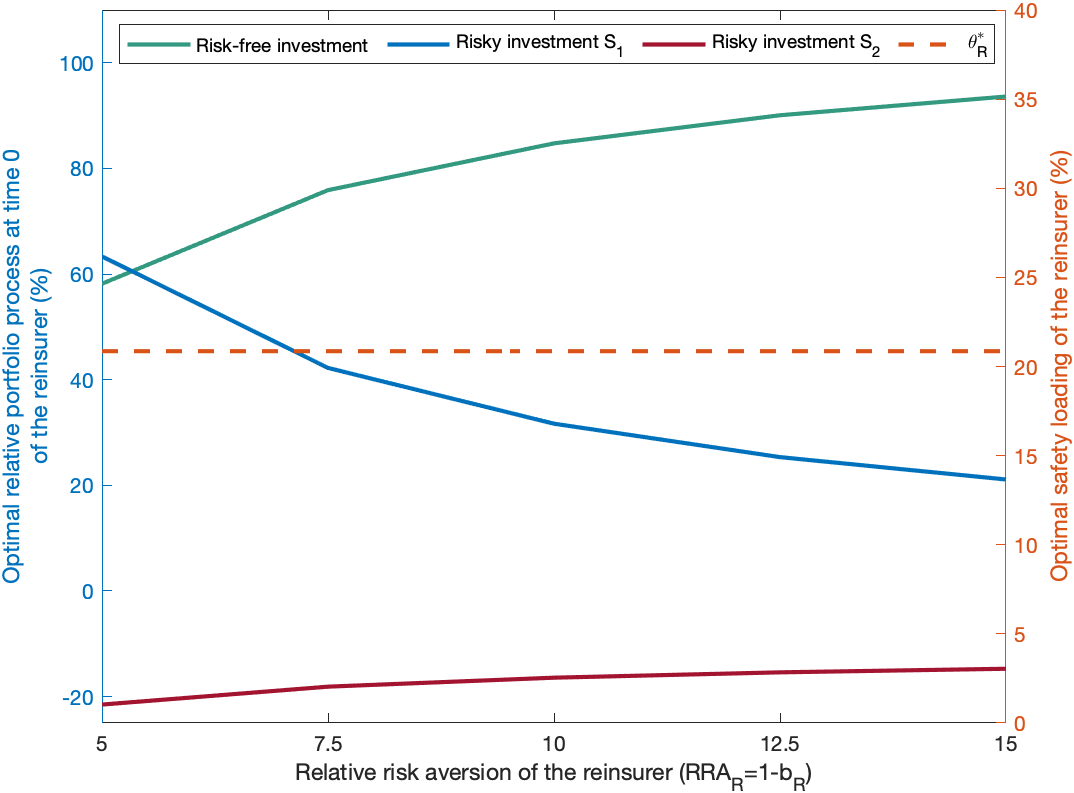}
   	\end{minipage}
   	\hspace{.05\linewidth}
   	\begin{minipage}[h]{.45\linewidth} 
   		\includegraphics[width=\linewidth]{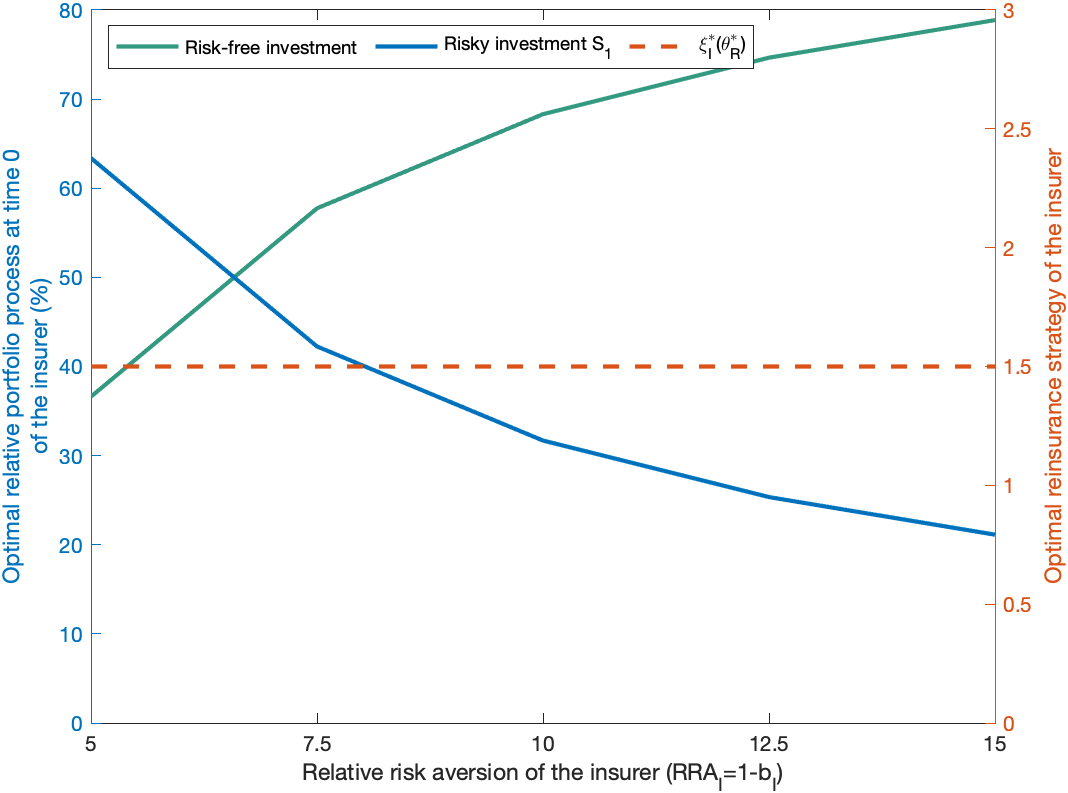}
   	\end{minipage}
   	\caption{Sensitivity of the Stackelberg equilibrium w.r.t. $RRA_R$ and $RRA_I$}
   	\label{fig:SensitivityRRA}
\end{figure}

\begin{figure}[!ht]
   	\begin{minipage}[h]{.45\linewidth}
   		\includegraphics[width=\linewidth]{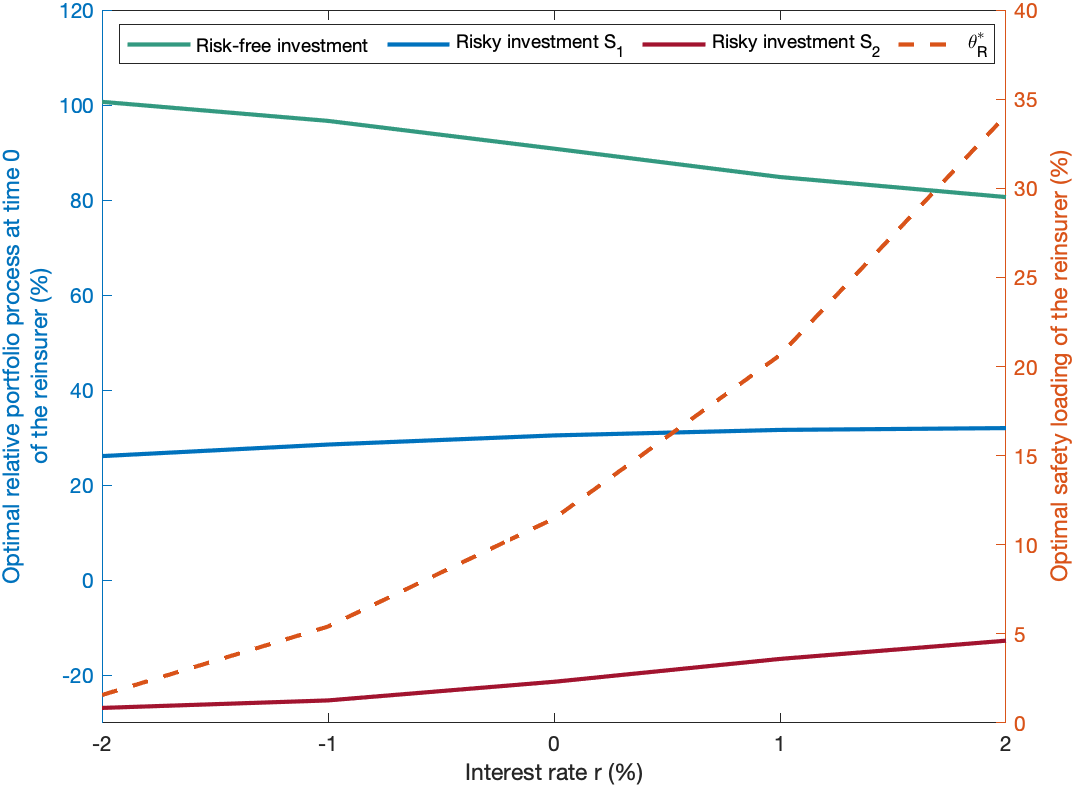}
   	\end{minipage}
   	\hspace{.05\linewidth}
   	\begin{minipage}[h]{.45\linewidth} 
   		\includegraphics[width=\linewidth]{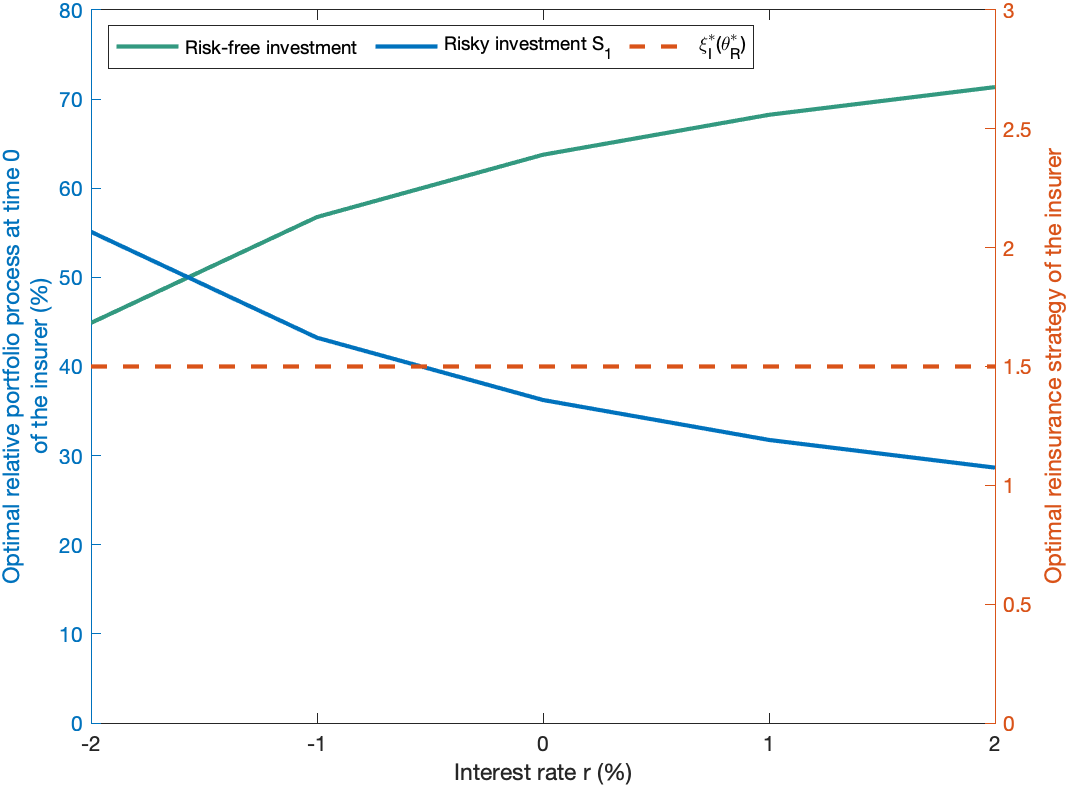}
   	\end{minipage}
   	\caption{Sensitivity of the Stackelberg equilibrium w.r.t. $r$}
   	\label{fig:SensitivityIR}
\end{figure}

\begin{figure}[!ht]
    	\begin{minipage}[h]{.45\linewidth} 
    		\includegraphics[width=\linewidth]{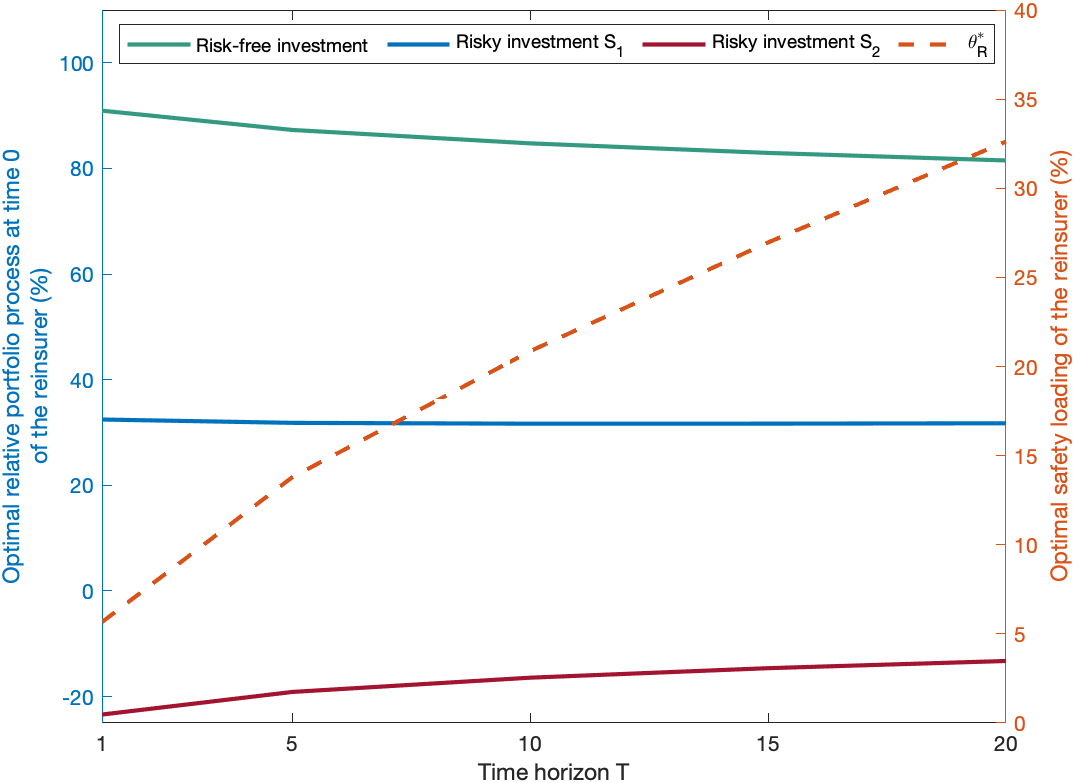}
    	\end{minipage}
    	\hspace{.05\linewidth}
    	\begin{minipage}[h]{.45\linewidth} 
    		\includegraphics[width=\linewidth]{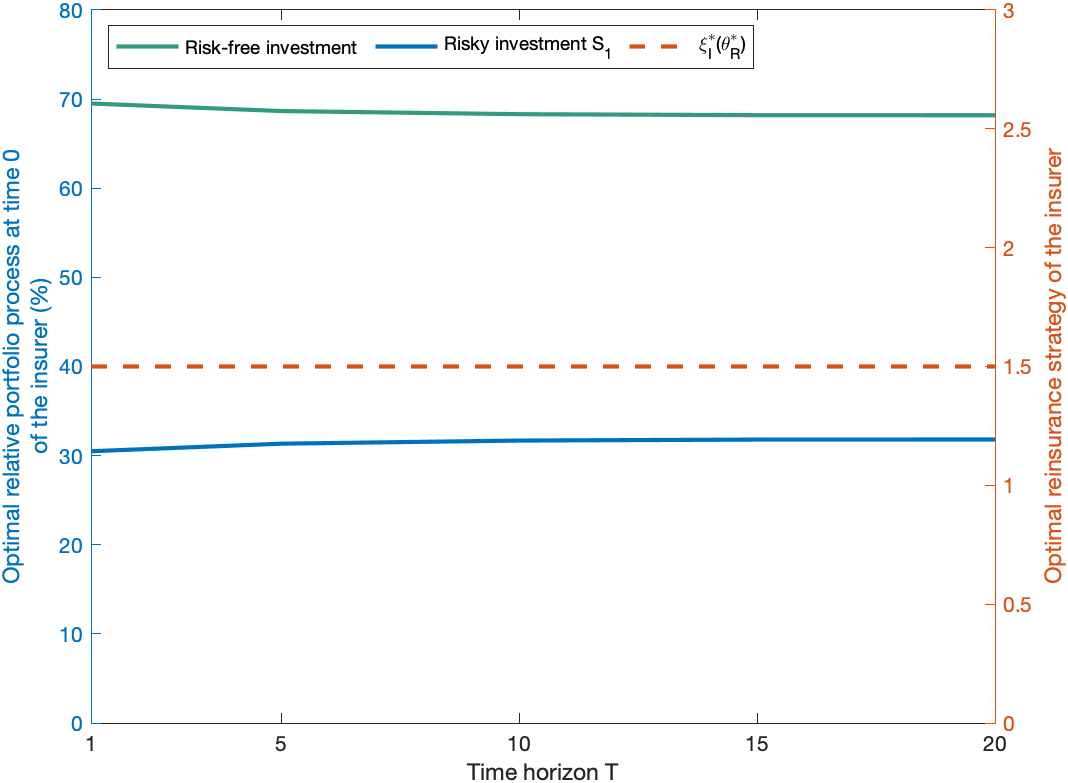}
    	\end{minipage}
    	\caption{Sensitivity of the Stackelberg equilibrium w.r.t. $T$}
    	\label{fig:SensitivityTH}
    \end{figure}
    
    \begin{figure}[!ht]
    	\begin{minipage}[h]{.45\linewidth} 
    		\includegraphics[width=\linewidth]{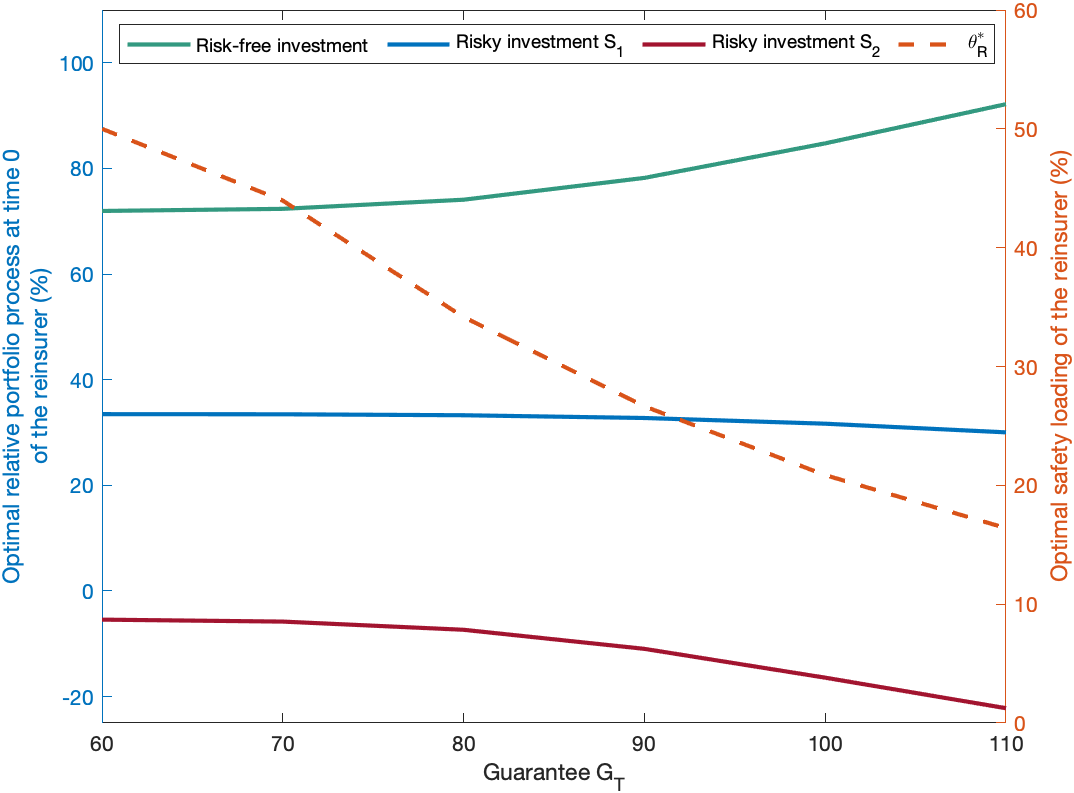}
    	\end{minipage}
    	\hspace{.05\linewidth}
    	\begin{minipage}[h]{.45\linewidth} 
    		\includegraphics[width=\linewidth]{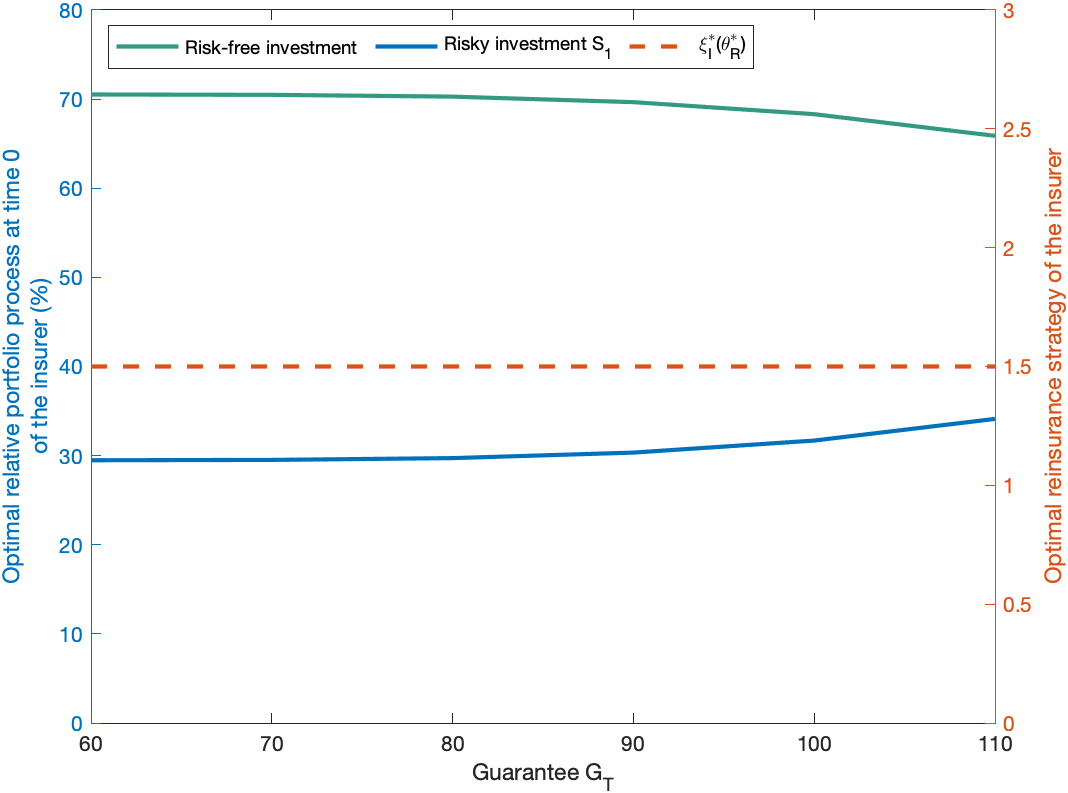}
    	\end{minipage}
    	\caption{Sensitivity of the Stackelberg equilibrium w.r.t. $G_T$}
    	\label{fig:SensitivityG}
    \end{figure}
\color{black}

    \subsection{Investment strategies over whole time horizon}\label{subsec:dynamic_investment_strategies}$\,$\\
    In this subsection, we illustrate how the investment strategies of the reinsurer and the insurer can develop over the entire time horizon $[0,T]$. In Figures \ref{fig:Invest_R} and \ref{fig:Invest_I}, we show $100$ exemplary paths of the relative-portfolio process (semi-transparent lines) as well as the average of the investment strategies (thick nontransparent lines) over the whole time horizon. The average values are calculated using $10000$ simulated paths.
    
    In Figure \ref{fig:Invest_R}, we see that the reinsurer invests on average around $30\%$ of its wealth in the first risky asset and sells between $5\%$ and $15\%$ of the second risky asset. The closer the maturity of the equity-linked product, the more the reinsurer invests in $S_1$ and the less it sells $S_2$. The share of capital invested in the risk-free asset $S_0$ decreases on average.
    
    In contrast, the insurer's fraction of wealth invested in the first risky asset decreases on average over the time horizon and therefore the fraction invested in the risk-free asset increases, as Figure \ref{fig:Invest_I} indicates.
    
    \begin{figure}[h!]
    	\begin{minipage}[h]{.45\linewidth} 
    		\includegraphics[width=\linewidth]{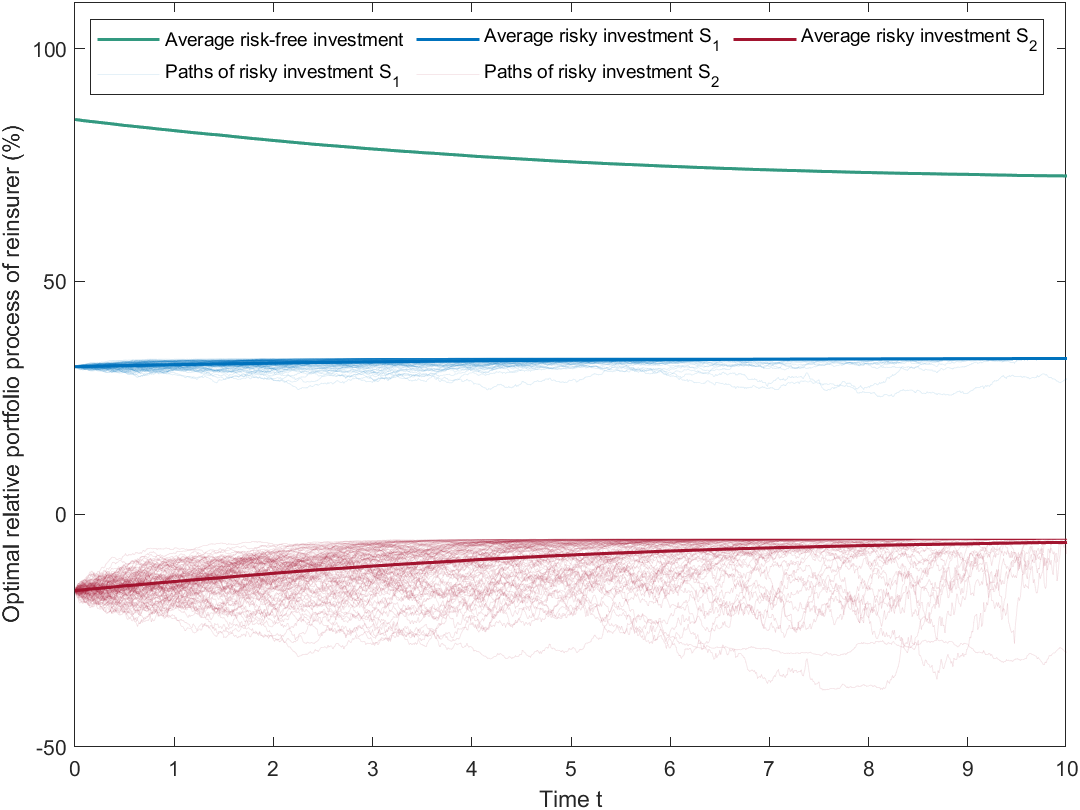}
    		\caption{Reinsurer's investment strategy over the time horizon $[0,T]$}
    	    \label{fig:Invest_R}
    	\end{minipage}
    	\hspace{.05\linewidth}
    	\begin{minipage}[h]{.45\linewidth} 
    		\includegraphics[width=\linewidth]{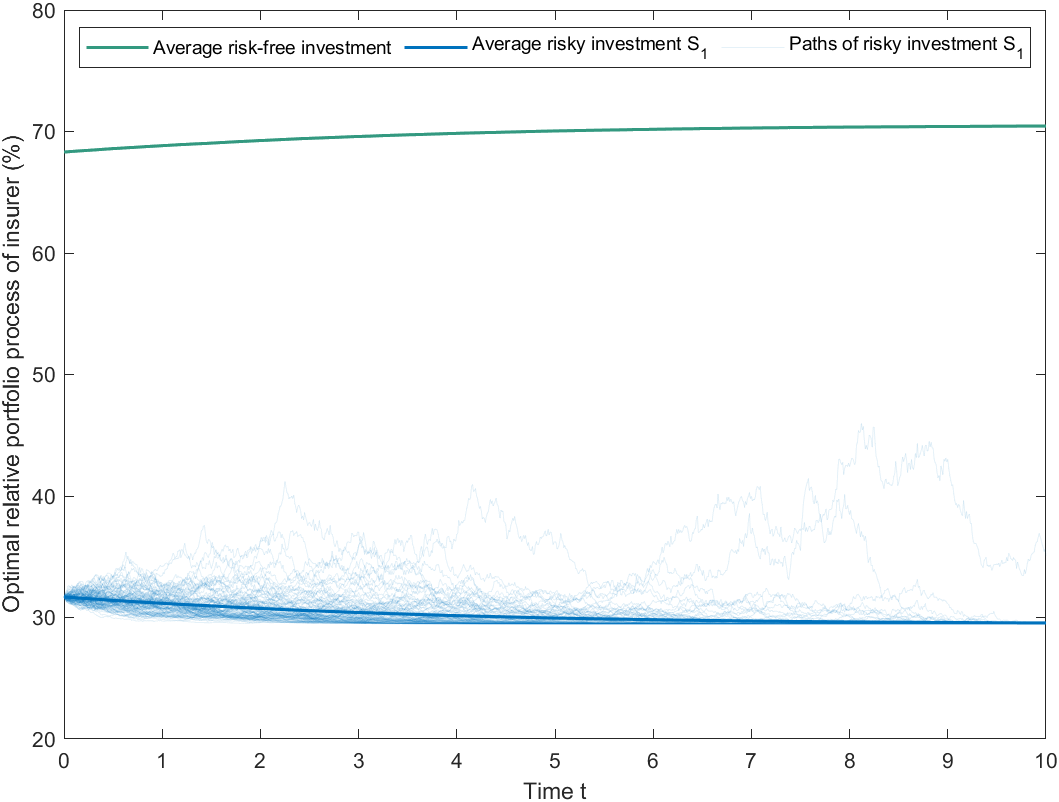}
    		\caption{Insurer's investment strategy over the time horizon $[0,T]$}
    	    \label{fig:Invest_I}
    	\end{minipage}
    \end{figure}

\color{black}
\section{Reinsuirer's incentive to charge a discounted safety loading}\label{sm:reinsurer_incentive}$\,$

    In this section, we answer the question whether the expected utility of the reinsurer is higher when it charges a discounted safety loading in comparison to the situation when it does not sell any reinsurance.  In other words, we explore whether the participation constraint is satisfied for the reinsurer when it charges a lower (discounted) safety loading than the equilibrium one.
    
    When $\theta_R\in(0,\theta_R^\ast)$, the reinsurer benefits from offering reinsurance in comparison to the case of not selling reinsurance at all. The benefit increases as $\theta_R > 0$ increases and $\theta_R < \theta_R^\ast$. To see this, recall that the reinsurer's optimal total terminal wealth without reinsurance is given by
	\begin{align*}
		V_R^{v_{R,0}(0,0),\pi_R^\ast}(T)=I_R(y_R^\ast\widetilde{Z}(T)),
	\end{align*}
	where $y^\ast_R$ solves $\mathbb{E}[\widetilde{Z}(T)I_R(y_R^\ast\widetilde{Z}(T))]=v_R$. In contrast, the reinsurer's optimal total terminal wealth with reinsurance that is priced with a safety loading that is slightly lower than the equilibrium one is given by
	\begin{align*}
		V_R^{v_{R,0}(\xi_I^\ast(\theta_R(\alpha)),\theta_R(\alpha)),\pi_R^\ast}(T)-\xi_I^\ast(\theta_R(\alpha)) P(T)=I_R(y_R^\ast(\theta_R(\alpha))\widetilde{Z}(T)),
	\end{align*}
	where $y^\ast_R(\theta_R(\alpha))$ solves $\mathbb{E}[\widetilde{Z}(T)I_R(y_R^\ast(\theta_R(\alpha))\widetilde{Z}(T))]=v_R+\xi_I^\ast(\theta_R(\alpha))\theta_R(\alpha) P(0)$. As we can see, the reinsurer's optimal total terminal wealth without reinsurance is the same as with reinsurance with full discount (i.e., $\alpha=0$ and, therefore, $\theta_R=0$). If $\xi_I^\ast(\theta_R(\alpha))\theta_R(\alpha)$ strictly increases when $\alpha$ increases, then $I_R(y_R^\ast(\theta_R(\alpha)\widetilde{Z}(T))$ increases strictly up to $I_R(y_R^\ast(\theta_R^\ast)\widetilde{Z}(T))$. Since the reinsurer's expected utility is strictly concave in the terminal wealth and $I_R(y_R^\ast(\theta_R^\ast)\widetilde{Z}(T))$ maximizes it, it is strictly increasing on the interval $[0,I_R(y_R^\ast(\theta_R^\ast)\widetilde{Z}(T))]$. Hence, it holds for $\alpha>0$ that
	\begin{align*}
		\mathbb{E}\left[U_R\left(V_R^{v_{R,0}(0,0),\pi_R^\ast}(T)\right)\right]&=\mathbb{E}\left[U_R\left(V_R^{v_{R,0}(\xi_I^\ast(\theta_R(0)),\theta_R(0)),\pi_R^\ast}(T)-\xi_I^\ast(\theta_R(0)) P(T)\right)\right]\\
		&=\mathbb{E}\left[U_R\left(I_R(y_R^\ast(\theta_R(0)\widetilde{Z}(T))\right)\right]<\mathbb{E}\left[U_R\left(I_R(y_R^\ast(\theta_R(\alpha)\widetilde{Z}(T))\right)\right]\\
		&=\mathbb{E}\left[U_R\left(V_R^{v_{R,0}(\xi_I^\ast(\theta_R(\alpha)),\theta_R(\alpha)),\pi_R^\ast}(T)-\xi_I^\ast(\theta_R(\alpha)) P(T)\right)\right].
	\end{align*}

	In our case, we have $\xi_I^\ast(\theta_R(\alpha))\theta_R(\alpha)=\bar{\xi}\alpha\theta_R^\ast$, which is increasing in $\alpha \in [0,1]$.
	
	Figure \ref{subfig:WEUC_2} illustrates in monetary terms the benefit of a discounted safety loading in comparison to the case of not selling reinsurance at all. 
	
	\begin{figure}[h]
	    \centering
	    \includegraphics[width=0.6\linewidth]{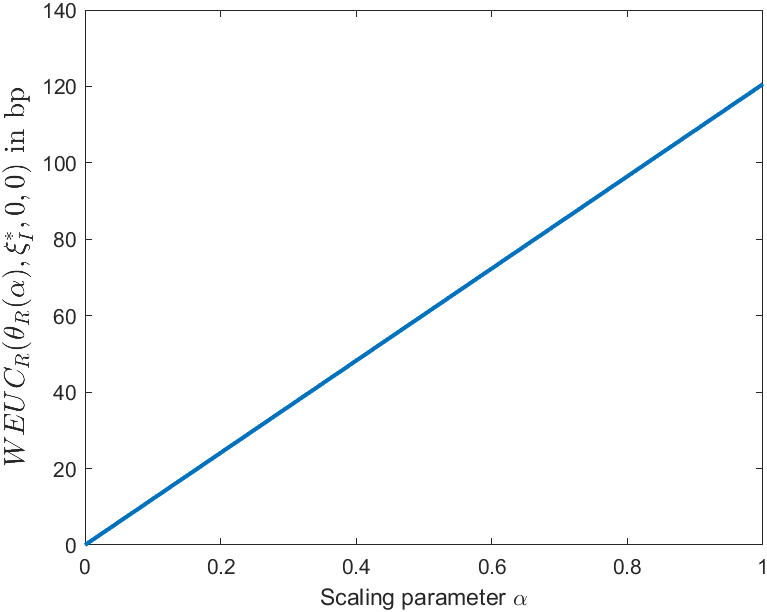}
		\caption{Impact of $\alpha$ on the reinsurer's WEUC with the reference combination $(\pi_R^\ast(\cdot|\theta_R(\alpha)),\theta_R(\alpha))$ and the alternative combination $(\pi_R^\ast(\cdot|0),0)$}
		\label{subfig:WEUC_2}
   \end{figure}

\section{Discussion on other utility functions}\label{sm_sec:other_utilities}$\,$

This section consists of two subsections. In Subsection \ref{subsec:insurers_best_responses}, we show that an insurer with a logarithmic or with a HARA-utility function becomes indifferent in the amount of reinsurance when the reinsurance company charges the equilibrium safety loading. In Subsection \ref{subsec:insurers_relative_portfolio}, we derive the optimal relative-portfolio process in the Stackelberg equilibrium for an insurer whose preferences are described by a logarithmic-utility or a HARA-utility function.

\subsection{Insurer's best response to the equilibrium safety loading}\label{subsec:insurers_best_responses}$\,$

In this subsection, we show that for utility functions, other than the power-utility function, the insurer is also indifferent in the amount of reinsurance to purchase in the equilibrium. 


As in the paper, we assume $\xi^{\max} = \bar{\xi}<\frac{v_I}{(1 + \theta_R^{\max})P(0)}$ and denote by $\lambda^\ast$ the dual process characterizing the optimal auxiliary market. Recall that by Lemma \ref{HLemma1}: 
\begin{align*}
    \nu(\xi)=\mathbb{E}[U_I(I_I(y^\ast_\lambda(\xi)\widetilde{Z}_{\lambda^\ast}(T)))].
\end{align*}

     \textbf{Logarithmic utility function.} We have $U_I(x):=\ln(x)$ and $I_I(x)=\frac{1}{x}$. Hence, it holds
    \begin{align*}
        \nu(\xi)=-\mathbb{E}[\ln(\widetilde{Z}_{\lambda^\ast}(T))]-\ln(y^\ast_{\lambda^\ast}(\xi))
    \end{align*}
    with the Lagrange multiplier (determined by the budget constraint)
	\begin{align*}
	    y^\ast_{\lambda^\ast}(\xi)=(v_I-\xi(1+\theta_R)P(0)+\xi\mathbb{E}[P(T)\widetilde{Z}_{\lambda^\ast}(T)])^{-1}.
	\end{align*}
	
	Hence, $\nu$ is of the form
	\begin{align*}
		\nu(\xi)=-C+\ln(C_1-\xi C_2)
	\end{align*}
	where $C,C_1,C_2\in\mathbb{R}$ are constants with $C_2=(1+\theta_R)P(0)-\mathbb{E}[\widetilde{Z}_{\lambda^\ast}(T)P(T)]$. For $C_2>0$,  $\nu$ strictly decreases if $\xi$ increases, i.e., $\xi_I^\ast=0$. For $C_2=0$, $\nu$ is independent of $\xi$, i.e., $\xi_I^\ast=\tilde{\xi}$ for any $\tilde{\xi}\in[0,\bar{\xi}]$. For $C_2<0$,  $\nu$ is strictly increasing in $\xi$, i.e., $\xi_I^\ast=\bar{\xi}$.
	
	For the equilibrium safety loading $\theta_R^\ast$, $C_2 = 0$, thus, the insurer's best response in terms of the reinsurance amount is $\xi_I^\ast=\tilde{\xi}$ for any $\tilde{\xi}\in[0,\bar{\xi}]$, i.e., the insurance company is indifferent in the amount of reinsurance in equilibrium.
	
	 \textbf{HARA-utility function.} We have $U_I(x):=\frac{(x+a)^{b_I}}{b_I}$ with $a\in\mathbb{R}$ s.t. $x+a>0$ and $b_I\in(-\infty,1)\backslash\{0\}$ and $I_I(x)=x^{\frac{1}{b_I-1}}-a$. The function $\nu$ is given by
	\begin{align*}
		\nu(\xi)=\frac{1}{b_I}\mathbb{E}\left[(y^\ast_{\lambda^\ast}(\xi)\widetilde{Z}_{\lambda^\ast}(T))^{\frac{b_I}{b_I-1}}\right]
	\end{align*}
	with the Lagrange multiplier (determined by the budget constraint)
	\begin{align*}
		y^\ast_{\lambda^\ast}(\xi)=\left(\frac{v_I-\xi(1+\theta_R)P(0)+a\mathbb{E}[\widetilde{Z}_{\lambda^\ast}(T)]+\xi\mathbb{E}[P(T)\widetilde{Z}_{\lambda^\ast}(T)]}{\mathbb{E}[\widetilde{Z}_{\lambda^\ast}(T)^{\frac{b_I}{b_I-1}}]}\right)^{b_I-1}.
	\end{align*}
	
	Hence, $\nu$ is of the form
	\begin{align*}
		\nu(\xi)=\frac{1}{b_I}\mathbb{E}\left[\widetilde{Z}_{\lambda^\ast}(T)^{\frac{b_I}{b_I-1}}\right](C_1-\xi C_2)^{b_I},
	\end{align*}
	where $C_1,C_2\in\mathbb{R}$ are constants with $C_2=(1+\theta_R)P(0)-\mathbb{E}[\widetilde{Z}_{\lambda^\ast}(T)P(T)]$, which is the same as in the case of a logarithmic-utility function. Thus, we conclude that for $\theta_R^\ast$, the insurer's best response is $\xi_I^\ast=\tilde{\xi}$ for any $\tilde{\xi}\in[0,\bar{\xi}]$.

\subsection{Insurer's relative portfolio process in the equilibrium}\label{subsec:insurers_relative_portfolio}$\,$

In this subsection, we derive the insurer's relative portfolio process in the Stackelberg equilibrium. We show that for insurers with logarithmic-utility and with HARA-utility functions the optimal auxiliary market is the same as the one for the insurer with a corresponding power-utility function. We do not derive here the reinsurer's optimal relative portfolio process, because its derivation is much simpler than the derivation for the insurer and boils down to applying Proposition 3.3 to get the optimal trading strategy and to converting it to the relative portfolio process via Relation (3.4) from the paper.

    \textbf{Logarithmic utility function.} We have that $U_I(x):=\ln(x)$ and $\hat{I}_I(x)=\frac{1}{x}-\xi P(T)$. We prove that $\hat{I}_I$ and $\frac{d\hat{I}_I(x)}{d x}$ are polynomially bounded at $0$ and $\infty$: For $x>0$ it holds
    \begin{align*}
        |\hat{I}_I(x)|&=\left|\frac{1}{x}-\xi P(T)\right|\overset{(i)}{\leq}\frac{1}{x}+\xi P(T)\overset{(ii)}{\leq} (1+\xi P(T))\left(\frac{1}{x}+x\right);\\
        \left|\frac{d\hat{I}_I(x)}{dx}\right|&=\left|-\frac{1}{x^2}\right|\overset{x>0}{=}\left(\frac{1}{x}\right)^2\overset{x>0}{\leq}\left(x+\frac{1}{x}\right)^2,
    \end{align*}
    where in $(i)$ we use the Cauchy-Schwartz inequality and in $(ii)$ $\frac{1}{x}+x\geq 1$ for all $x>0$. By Lemma \ref{HLemma1}, we have 
	\begin{align*}
		I_I(y^\ast(\xi)\widetilde{Z}_\lambda(T))>\xi P(T).
	\end{align*}
	
	From this result and by Proposition 3.1, the optimal terminal wealth (before reinsurance payment, i.e., {not} the {total} terminal wealth) is given by
	\begin{align}
		V_\lambda^\ast(T)=&I_I(y^\ast(\xi)\widetilde{Z}_\lambda(T))-\xi P(T) = (y^\ast(\xi)\widetilde{Z}_\lambda(T))^{-1}-\xi P(T).\label{eq1}
	\end{align}
	
	Thus, using $\frac{d\hat{I}(x)}{dx} = -x^{-2}$ and Proposition 3.1, we get that the insurer's optimal relative portfolio process in $\mathcal{M}_\lambda$ satisfies: 
	\begin{align*}
		\pi_\lambda^\ast(t)V_\lambda^\ast(t)=&-(\sigma^\top)^{-1}\gamma_\lambda\widetilde{Z}_\lambda(t)^{-1}\mathbb{E}[\widetilde{Z}_\lambda(T)y^\ast(\xi)\widetilde{Z}_\lambda(T)(-y^\ast(\xi)\widetilde{Z}_\lambda(T))^{-2}|\mathcal{F}_t]\\
		=&(\sigma^\top)^{-1}\gamma_\lambda\widetilde{Z}_\lambda(t)^{-1}\mathbb{E}[\widetilde{Z}_\lambda(T)(y^\ast(\xi)\widetilde{Z}_\lambda(T))^{-1}|\mathcal{F}_t]\\
		\overset{\eqref{eq1}}{=}&(\sigma^\top)^{-1}\gamma_\lambda\widetilde{Z}_\lambda(t)^{-1}\mathbb{E}[\widetilde{Z}_\lambda(T)(V_\lambda^\ast(T)+\xi P(T))|\mathcal{F}_t]\\
		=&(\sigma^\top)^{-1}\gamma_\lambda(V_\lambda^\ast(t)+\xi \widetilde{Z}_\lambda(t)^{-1}\mathbb{E}[\widetilde{Z}_\lambda(T)P(T)|\mathcal{F}_t]).
	\end{align*}
	
	Therefore:
	\begin{align*}
		\pi_\lambda^\ast(t) = (\sigma^\top)^{-1}\gamma_\lambda(1 +\xi \widetilde{Z}_\lambda(t)^{-1} \mathbb{E}[\widetilde{Z}_\lambda(T)P(T)|\mathcal{F}_t] / V_\lambda^\ast(t) ).
	\end{align*}
	
	The optimal dual process $\lambda^\ast$ is the same as in the case of a power-utility function, since $\pi_\lambda^\ast(t)\in K \Leftrightarrow (\sigma^\top)^{-1}\gamma_\lambda\in K \Leftrightarrow \pi^M_{\lambda} \in K $ with $\pi^M_{\lambda}$ defined in Corollary 4.1 in the paper.
	
	\textbf{HARA-utility function.} We have that $U_I(x):=\frac{(x+a)^{b_I}}{b_I}$ with $a\in\mathbb{R}$ s.t. $x+a>0$ and $b_I\in(-\infty,1)\backslash\{0\}$ and $\hat{I}_I(x)=x^{\frac{1}{b_I-1}}-a-\xi P(T)$. Now we prove that $\hat{I}_I$ and $\frac{d\hat{I}_I(x)}{dx}$ are polynomially bounded at $0$ and $\infty$. 
 
    For $x>0$ it holds
	\begin{align*}
	    |\hat{I}_I(x)|=&\left|x^{\frac{1}{b_I-1}}-a-\xi P(T)\right|\overset{(i)}{\underset{x>0}{\leq}}x^{\frac{1}{b_I-1}}+|a+\xi P(T)|=\left(\frac{1}{x}\right)^{\frac{1}{1-b_I}}+|a+\xi P(T)|\\
	    \overset{(ii)}{\underset{(iii)}{\leq}}&\left(x+\frac{1}{x}\right)^{\frac{1}{1-b_I}}+|a+\xi P(T)|\overset{(iv)}{\leq}(1+|a+\xi P(T)|)\left(x+\frac{1}{x}\right)^{\frac{1}{1-b_I}};\\
	    \left|\frac{d\hat{I}_I(x)}{dx}\right|=&\left|\frac{1}{b_I-1}x^{\frac{2-b_I}{b_I-1}}\right|\overset{b_I-1<0}{\underset{x>0}{=}}\frac{1}{1-b_I}x^{\frac{2-b_I}{b_I-1}}=\frac{1}{1-b_I}\left(\frac{1}{x}\right)^{\frac{2-b_I}{1-b_I}}\overset{(ii)}{\underset{(iii)}{\leq}}\frac{1}{1-b_I}\left(x+\frac{1}{x}\right)^{\frac{2-b_I}{1-b_I}},
	\end{align*}
	where we use in
	\begin{itemize}
	    \item[$(i)$] Cauchy-Schwartz inequality;
	    \item[$(ii)$] $x+\frac{1}{x}>\frac{1}{x}$ since $x>0$;
	    \item[$(iii)$] $x\mapsto x^k$ increases for $k>0$ with $k=\frac{1}{1-b_I}>0$ or $k=\frac{2-b_I}{1-b_I}>0$ for all $b_I\in(-\infty,1)\backslash\{0\}$;
	    \item[$(iv)$] $1\leq \left(x+\frac{1}{x}\right)^{\frac{1}{1-b_I}}$ for all $x>0$.
	\end{itemize}
	
	By Lemma \ref{HLemma1}, we have 
	\begin{align*}
		I_I(y^\ast(\xi)\widetilde{Z}_\lambda(T))>\xi P(T).
	\end{align*}
	
	From this result and by Proposition 3.1, the optimal terminal wealth (before reinsurance payment, i.e., {not} the {total} terminal wealth) is given by
	\begin{align}
		V_\lambda^\ast(T)=&I_I(y^\ast(\xi)\widetilde{Z}_\lambda(T))-\xi P(T)
		=(y^\ast(\xi)\widetilde{Z}_\lambda(T))^{\frac{1}{b_I-1}}-a-\xi P(T).\label{eq2}
	\end{align}
	Thus, using $\frac{d\hat{I}(x)}{dx}=\frac{1}{b_I-1}x^{\frac{2-b_I}{b_I-1}}$ and Proposition 3.1, we get that the insurer's optimal relative portfolio process in $\mathcal{M}_\lambda$ satisfies:
	\begin{align*}
		\pi_\lambda^\ast(t)V_\lambda^\ast(t)=&-(\sigma^\top)^{-1}\gamma_\lambda\widetilde{Z}_\lambda(t)^{-1}\mathbb{E}[\widetilde{Z}_\lambda(T)y^\ast(\xi)\widetilde{Z}_\lambda(T)\frac{1}{b_I-1}(y^\ast(\xi)\widetilde{Z}_\lambda(T))^{\frac{2-b_I}{b_I-1}}|\mathcal{F}_t]\\
		=&\frac{1}{1-b_I}(\sigma^\top)^{-1}\gamma_\lambda\widetilde{Z}_\lambda(t)^{-1}\mathbb{E}[\widetilde{Z}_\lambda(T)(y^\ast(\xi)\widetilde{Z}_\lambda(T))^{\frac{1}{b_I-1}}|\mathcal{F}_t]\\
		\overset{\eqref{eq2}}{=}&\frac{1}{1-b_I}(\sigma^\top)^{-1}\gamma_\lambda\widetilde{Z}_\lambda(t)^{-1}\mathbb{E}[\widetilde{Z}_\lambda(T)(V_\lambda^\ast(T)+a+\xi P(T))|\mathcal{F}_t]\\
		=&\frac{1}{1-b_I}(\sigma^\top)^{-1}\gamma_\lambda(V_\lambda^\ast(t)+\xi \widetilde{Z}_\lambda(t)^{-1}\mathbb{E}[\widetilde{Z}_\lambda(T)P(T)|\mathcal{F}_t]+a).
	\end{align*}
	
	Thus:
	\begin{align*}
		\pi_\lambda^\ast(t)= \pi^M_\lambda(1 + \xi \widetilde{Z}_\lambda(t)^{-1}\mathbb{E}[\widetilde{Z}_\lambda(T)P(T)|\mathcal{F}_t]/ V_\lambda^\ast(t) + a / V_\lambda^\ast(t)).
	\end{align*}
	
	The optimal dual process $\lambda^\ast$ is the same as in the case of a power-utility function, since $\pi_\lambda^\ast(t)\in K  \Leftrightarrow \pi^M_{\lambda} \in K $ with $\pi^M_{\lambda}$ defined in Corollary 4.1 in the paper.

\section{Discussion on the surrender and mortality risk} \label{sm_sec:actuarial_risks} $\,$

Mortality and/or surrender risks could be added in various ways to the insurer's (follower's) optimization problem. If these risks are independent of the financial risks and the insurer pays money to the buyer of the equity-linked product at time $T$ only if the buyer does not surrender and does not die before $T$, then the Stackelberg equilibrium does not change. However, once the dependence between financial risks and surrender/mortality risks is introduced and/or the payoff of the insurer is more intertwined with surrender/mortality risks, the insurer's optimization problem becomes much more complicated and deserves a separate treatment. 

Consider a random variable $\tau$ that indicates the time of death or surrender of the policyholder. We could consider the following problem of the insurer:
\begin{equation}\label{OP:tau_easy}
    \max_{(\pi_I, \xi_I) \in \Lambda_I} \mathbb{E}\left[ U_I(V_I^{v_{I,0}(\xi_I,\theta_R),\pi_I}(T) + \xi_I P(T))\mathbbm{1}_{\{ \tau \geq T \}} \right],
\end{equation}
where the set of admissible strategies of the insurer is:
    \begin{align*}
    	\Lambda_I:=\{(\pi_I,\xi_I)|\;&\pi_I\text{ self-financing, }\pi_I(t)\in K\;\mathbb{Q}\text{-a.s. }\forall t\in[0,T],\;\xi_I\in[0,\xi^{\max}],\\
    	&V^{v_{I,0}(\xi_I,\theta_R),\pi_I}_I(t)\geq0\;\mathbb{Q}\text{-a.s. }\forall t\in[0,T]\text{ and }\\
    	&\mathbb{E}[(U_I(V_I^{v_{I,0}(\xi_I,\theta_R),\pi_I}(T)+\xi_I P(T))\mathbbm{1}_{\{ \tau \geq T \}})^-]<\infty]\}.
    \end{align*}

If $\tau$ is independent of financial risks, then:
\begin{align*}
   \mathbb{E}& \left[ U_I(V_I^{v_{I,0}(\xi_I,\theta_R),\pi_I}(T) + \xi_I P(T))\mathbbm{1}_{\{ \tau \geq T \}} \right] \\
   & \stackrel{(i)}{=} \mathbb{E}\left[\mathbb{E}\left[ U_I((V_I^{v_{I,0}(\xi_I,\theta_R),\pi_I}(T) + \xi_I P(T))\mathbbm{1}_{\{ \tau \geq T \}} 
   | \mathcal{F}_T\right]\right]\\
   & \stackrel{(ii)}{=} \mathbb{E}\left[U_I(V_I^{v_{I,0}(\xi_I,\theta_R),\pi_I}(T) + \xi_I P(T))\mathbb{E}\left[ \mathbbm{1}_{\{ \tau \geq T \}} 
   | \mathcal{F}_T\right] \right]\\
   & \stackrel{(iii)}{=} \mathbb{E}\left[U_I(V_I^{v_{I,0}(\xi_I,\theta_R),\pi_I}(T) + \xi_I P(T)) \right] \mathbb{E}\left[ \mathbbm{1}_{\{ \tau \geq T \}} \right] 
\end{align*}
where we use in (i) we use a tower rule of conditional expectations, in (ii) the $\mathcal{F}_T$-measurability of the financial risks, in (iii) the independence between $\tau$ and $\mathcal{F}_T$. Thus, the insurer's optimization problem we considered in the paper has the same solution as the solution to \eqref{OP:tau_easy}, because they differ only by a positive constant multiplier in the objective function. As a result, the Stackelberg equilibrium remains the same as before.

The situation when the policyholder also receives some money at the time of death or surrender can be modelled by incorporating it into the insurer's objective function as follows:
\begin{equation}\label{OP:tau_medium}
    \max_{(\pi_I, \xi_I) \in \Lambda_I} \mathbb{E}\left[U_I(V_I^{v_{I,0}(\xi_I,\theta_R),\pi_I}(\tau) + \xi_I P(\tau))\mathbbm{1}_{\{ \tau < T \}} + U_I(V_I^{v_{I,0}(\xi_I,\theta_R),\pi_I}(T) + \xi_I P(T))\mathbbm{1}_{\{ \tau \geq T \}} \right].
\end{equation}

This way of modeling implies that the put option is of an American type, not a European one, as in our article. This introduces several additional levels of complexity. The first challenge is pricing this option. The second issue is solving the insurer's problem \eqref{OP:tau_medium} with two terms in the insurer's problem with uncertainty about both the time and the amount of the payoff. The third challenge is solving the optimization problem of the reinsurance company, as it has to dynamically hedge its short position in the put option with an uncertain exercise time. The Stackelberg equilibrium would most probably change under this model, but we do not see a direct way of getting the intuition about the direction of change. 

A different approach to modeling surrender and death benefit can be found in \citeA{Kronborg2015}. There, the researchers take a standpoint of a policyholder who dynamically controls the death benefit and the investment strategy. The controlled process that models the death benefit appears in the drift part of the wealth process. Following their approach in our article would require new tools to solving the bi-level optimization problem that models the overall Stackelberg game.

\end{appendix}

\end{document}